\documentclass[oneside,english]{amsart}
\usepackage[T1]{fontenc}
\usepackage[latin9]{inputenc}
\usepackage{textcomp}
\usepackage{amstext}
\usepackage{amsthm}
\usepackage{amssymb}

\makeatletter

\providecommand{\tabularnewline}{\\}

\numberwithin{equation}{section}
\numberwithin{figure}{section}
\theoremstyle{plain}
\newtheorem{thm}{\protect\theoremname}
  \theoremstyle{remark}
  \newtheorem{rem}[thm]{\protect\remarkname}
  \theoremstyle{plain}
  \newtheorem{lem}[thm]{\protect\lemmaname}

\usepackage{babel}

\makeatother

\usepackage{babel}
  \providecommand{\lemmaname}{Lemma}
  \providecommand{\remarkname}{Remark}
\providecommand{\theoremname}{Theorem}

\begin{document}

\title{Heterotic-$\mathbf{F}$-theory Duality with Wilson Line Symmetry-breaking}
\begin{abstract}
We begin with an $E_{8}\times E_{8}$ Heterotic model broken to an
$SU(5)_{gauge}$ and a mirror $SU(5)_{gauge}$, where one $SU(5)$
and its spectrum is identified as the visible sector while the other
can be identified as a hidden mirror world. In both cases we obtain
the minimal supersymmetric standard model spectrum after Wilson-line
symmetry-breaking enhanced by a low energy R-parity enforced by a
local (or global) $U(1)_{X}$-symmetry. Using Heterotic/$F$-theory
duality, we show how to eliminate the vector-like exotics which were
obtained in previous constructions. In these constructions, the Calabi-Yau
{[}CY{]} four-fold was defined by an elliptic fibration with section
over a base $B_{3}$ and a GUT surface given by $K3/\mathbb{Z}_{2}=$
Enriques surface. In the present paper we construct a quotient CY
four-fold fibered by tori with two elliptic structures given by a
pair of sections fibered over the Enriques surface. Using Heterotic/$F$-theory
duality we are able to define the cohomologies used to derive the
massless spectrum.

Our model for the 'correct' $F$-theory dual of a Heterotic model
with Wilson-line symmetry-breaking builds on prior literature but
employs the stack-theoretic version of the dictionary between the
Heterotic semi-stable $E_{8}$-bundles with Yang-Mills connection
and the $dP_{9}$-fibrations used to construct the $F$-theory dual. 
\end{abstract}

\author{Herbert Clemens and Stuart Raby }

\address{Mathematics Department, Ohio State University, Columbus OH 43210,
USA}

\email{clemens.43@osu.edu}

\address{Physics Department, Ohio State University, Columbus OH 43210, USA}

\email{raby.1@osu.edu}

\date{August 2, 2019}

\maketitle
\tableofcontents{}

\section{Introduction}

\subsection{The physics}

Supersymmetric grand unified theories {[}SUSY GUTs{]} \cite{Dimopoulos:1981yj,Dimopoulos:1981zb,Ibanez:1981yh}
have many nice properties. These include an explanation of the family
structure of quarks and leptons with the requisite charge assignments
under the Standard Model {[}SM{]} gauge group $SU(3)_{C}\times SU(2)_{L}\times U(1)_{Y}$
and a prediction of gauge coupling unification at a scale of order
$10^{16}$ GeV. The latter is so far the only direct hint for the
possible observation of supersymmetric particles at the LHC. UV completions
of SUSY GUTs in string theory also provide a consistent quantum mechanical
description of gravity. As a result of this golden confluence, many
groups have searched for SUSY GUTs in string theory. In fact, it has
been shown that by demanding SUSY GUTs in string constructions one
can find many models with features much like that of the minimal supersymmetric
Standard Model {[}MSSM{]} \cite{Lebedev:2006kn,Lebedev:2007hv,Kim:2007mt,Lebedev:2008un,Blumenhagen:2008zz,Anderson:2011ns,Anderson:2012yf}. 
\begin{flushleft}
The past several years have seen significant attention devoted to
the study of supersymmetric GUTs in $F$-theory \cite{Donagi:2008ca,Beasley:2008dc,Beasley:2008kw,Donagi,Blumenhagen,Grimm:2010ez,Tatar}.
Both local and global $SU(5)$ $F$-theory GUTs have been constructed
where $SU(5)$ is spontaneously broken to the SM via non-flat hypercharge
flux. One problem with this approach for GUT breaking is that large
threshold corrections are generated at the GUT scale due to the non-vanishing
hypercharge flux \cite{Donagi,Mayrhofer:2013ara,Blumenhagen,Blumenhagen2}.
An alternative approach to breaking the GUT group is using a Wilson
line in the hypercharge direction, i.e. a so-called flat hypercharge
line bundle. In this case it is known that large threshold corrections
are not generated at the GUT scale (or, in fact, leading to precise
gauge coupling unification at the compactification scale in orbifold
GUTs)\cite{Krippendorf:2013dqa,Raby:2009sf} and \cite{Ross:2004mi,Hebecker:2004ce,Trapletti:2006xv,Anandakrishnan}.
\par\end{flushleft}

In a previous paper, the present authors and collaborators constructed
a global $SU(5)$ $F$-theory model with Wilson line breaking \cite{Marsano:2012yc}.
The model contained the vector multiplets for the MSSM gauge group,
3 families of quarks and leptons, 4 pairs of Higgs doublets, and in
addition, a vector-like pair of chiral multiplets in the representation
$(\mathbf{3},\mathbf{2})_{-5/6}\oplus(\mathbf{\overline{3}},\mathbf{2})_{+5/6}$.
In terms of the model defined on an elliptically fibered CY 4-fold
with GUT surface defined as an Enriques surface, $K_{3}/\mathbb{Z}_{2}$,
the massless spectrum is given in terms of cohomologies of the flux
line bundle (or twists of the flux line bundle) on the GUT surface.
It was then shown that the holomorphic Euler character of any flat
bundle on $S_{{\rm GUT}}$ is equivalent to its Todd genus \cite{Beasley:2008kw,Donagi:2008ca,Marsano:2012yc},
through 
\begin{equation}
c_{1}(L_{Y})=0\quad\implies\quad\chi(S_{2},L_{Y})=\int_{S_{{\rm GUT}}}\text{Td}(TS_{{\rm GUT}}),
\end{equation}
so we have that 
\begin{equation}
\chi(S_{{\rm GUT}},L_{Y})=h^{0}(S_{{\rm GUT}},L_{Y})-h^{1}(S_{{\rm GUT}},L_{Y})+h^{2}(S_{{\rm GUT}},L_{Y})=1.
\end{equation}
Since all $h^{m}(S_{{\rm GUT}},L_{Y})$ cannot be vanishing, we are
guaranteed to get some massless vector or chiral states, $(\mathbf{3},\mathbf{2})_{-5/6}$'s
and $(\mathbf{\overline{3}},\mathbf{2})_{+5/6}$'s. We emphasize that
the presence of some kind of vector-like exotic matter is not a specific
issue with this Enriques model but rather a general property of any
model that breaks $SU(5)_{{\rm GUT}}\rightarrow SU(3)\times SU(2)\times U(1)_{Y}$
with a flat $U(1)_{Y}$ bundle on a holomorphic surface $S_{{\rm GUT}}$.
We note that this derivation is, however, only valid if the $F$-theory
is compactified on a CY 4-fold with section. In the present paper
we show how to eliminate the vector-like exotics and evade this theorem.

We need to understand how to modify the description of the exotics
(and the matter content) in a situation where there is no single distinguished
section, or, more precisely, where the torus fibration has two sections
`on equal footing.' The prescription we use gives exactly such a description.
The purpose of this paper is to present a model for Heterotic/$F$-theory
duality in which $SU(5)$ symmetry is broken (on both sides) by Wilson
lines. The work derives from the previous global $F$-theory model
with Wilson line symmetry-breaking \cite{Marsano:2012yc}. It modifies
the previous model so as to allow the construction of a Heterotic
dual. It adapts previously known Heterotic techniques for eliminating
undesirable features of the model, such as vector-like exotics, by
constructing the torus-fibration (on both sides) with two sections
\cite{Donagi:2004ub}. This allows us to use the Heterotic technique
of translation by the difference of the two sections to form the requisite
$\mathbb{Z}_{2}$-action in order to evade the above-mentioned theorem
and eliminate the vector-like exotics. The model has some very nice
features. It contains the gauge group $SU(5)$ broken via a Wilson
line in the hypercharge direction to $SU(3)\times SU(2)\times U(1)_{}$.
It contains three families of quarks and leptons and one pair of Higgs
doublets. Furthermore our $F$-theory model exhibits a $\mathbb{Z}_{4}$
$\mathbf{R}$-symmetry in the semi-stable limit of the $F$-theory
model.

In addition it has a local $U(1)_{X}$ gauge symmetry, given by a
global (4 + 1) split spectral cover, where $U(1)_{X}$ is the $U(1)$
in $SO(10)$ commuting with $SU(5)_{gauge}$; Higgs pairs have charge,
$\pm2$, while $SU(5)$ \textbf{10}'s and \textbf{$\bar{5}$}'s have
charge, +1 and -3, respectively. This gauge symmetry preserves $R$-parity,
eliminating baryon and lepton number violating dimension-4 operators
\cite{Grimm:2010ez}. $U(1)_{X}$ can also be used to identify possible
right-handed neutrino states which are $SU(5)$ singlet matter states
with $U(1)_{X}$ charge +5 \cite{Tatar,Blumenhagen}.

Some problems for this construction are as follows. The gauge symmetry
$U(1)_{X}$ also forbids Majorana masses for the right-handed neutrinos.
In addition, in Subsection \ref{sec:The-D-term} we calculate the
$D$-term for $U(1)_{X}$ (see Eqn. 3.30 in \cite{Grimm:2010ez,Grimm:2011tb}),
requiring $D_{U(1)_{X}}=0$ so that this symmetry is not spontaneously
broken by fields derived from the adjoint representation of $E_{8}$.
In addition, $U(1)_{X}$ is not sufficient to prevent dimension-5
baryon and lepton number violating operators. It is possible that
the local $U(1)_{X}$ symmetry is broken down to a global $U(1)_{X}$
symmetry via a Stueckelberg mechanism \cite{Tatar,Grimm:2010ez},
but this is beyond the scope of the present paper. The $U(1)_{X}$
symmetry may also be broken to a $\mathbb{Z}_{2}$ matter parity by
a non-perturbative effect at the GUT scale. This would then allow
for right-handed neutrino Majorana masses at the GUT scale.

Finally, a very novel feature of the model is that it contains a twin/mirror
$SU(5)$ symmetry broken to a mirror SM with three families of mirror
quarks and leptons and a pair of mirror Higgs multiplets.\footnote{Mirror world defined as the parity transform of the Standard Model
has been reviewed in the paper by L.B. Okun \cite{Okun}. This paper
has many references which we refer to the reader. Some of these references
however include another related definition of the mirror world given
by a generalized $\mathbb{Z}_{2}$-symmetry which takes the Standard
Model into the twin or mirror sector with states having identical
charges but different masses and couplings. Papers in this genre include,
for example, \cite{Berezhiani}, \cite{Chacko}, and \cite{Barbieri-1}.}This is a direct consequence of the fact that the $\mathrm{GUT}$
surface, $S_{\mathrm{GUT}}^{\vee}=$ Enriques, is a branched (therefore
irreducible) double cover of the base $B_{2}$. The gauge and Yukawa
couplings in the visible and mirror worlds are determined by volume
moduli which must still be stabilized and supersymmetry broken. As
a result mirror matter does not necessarily have the same mass as
visible matter or the same value of their gauge and Yukawa couplings.
This mirror sector is a possible candidate for the dark matter in
the universe.

\subsection{The mathematics}

We next give an idea of the mathematics of our model for Heterotic/$F$-theory
duality in which $SU\left(5\right)$-symmetry is broken (on both sides)
by a Wilson line construction. As mentioned, the mathematical model
derives from a previous global $F$-theory model with Wilson line
symmetry-breaking \cite{Marsano:2012yc}. It modifies the previous
model so as to allow the construction of a Heterotic dual. It also
adapts previously known Heterotic techniques for eliminating undesirable
features of the model, such as vector-like exotics, by constructing
the torus-fibration (on both sides) that admits two sections. More
properly, the two sections, taken together, should be thought of as
determining an invariant $\mathbf{g}_{2}^{1}$ linear series on the
torus fiber.

Our goal for this paper is to present a global phenomenonologically
consistent Heterotic model $V_{3}^{\vee}/B_{2}^{\vee}$ and $F$-theory
dual $W_{4}^{\vee}/B_{3}^{\vee}$. These will be constructed as equivariant
$\mathbb{Z}_{2}$-quotients of another set of dual Heterotic/$F$-theory
models $V_{3}/B_{2}$ and $W_{4}/B_{3}$ with respective involutions
\begin{equation}
\tilde{\beta}_{3}/\beta_{2}:V_{3}/B_{2}\rightarrow V_{3}/B_{2}\label{eq:b3}
\end{equation}
and 
\begin{equation}
\tilde{\beta}_{4}/\beta_{3}:W_{4}/B_{3}\rightarrow W_{4}/B_{3}\label{b4}
\end{equation}
whose quotients $V_{3}^{\vee}/B_{2}^{\vee}$ and $W_{4}^{\vee}/B_{3}^{\vee}$
are our ultimate goal. However much of the work will center on $V_{3}/B_{2}$
and $W_{4}/B_{3}$ and their respective involutions, passing to the
quotients only late in the story. Furthermore much of the work has
already been completed in our companion papers \cite{Clemens-1,Clemens-2},
including in \cite{Clemens-2} the complete construction of 
\[
B_{3}=\mathbb{\mathbb{P}}^{1}\times B_{2}
\]
and the involution $\beta_{3}$ and calculation of their numerical
invariants. The groundwork contained in those papers will be referred
to as needed in what follows.

In short, the fundamental challenge is to construct compatible involutions
(\ref{eq:b3}) and (\ref{b4}) so that both leave their top-degree
holomorphic forms invariant, that is, so that their respective quotients
are Calabi-Yau manifolds. Because of other necessary characteristics
of the Heterotic and $F$-theory models, we showed in \cite{Clemens-1}
that (\ref{eq:b3}) must act as 
\begin{equation}
\frac{dx}{y}\mapsto\frac{dx}{y}\label{eq:y}
\end{equation}
on the relative (Weierstrass) one-form on $V_{3}/B_{2}$ whereas (\ref{b4})
must act as 
\begin{equation}
\frac{dx}{y}\mapsto-\frac{dx}{y}\label{eq:-y}
\end{equation}
on the relative (Weierstrass) one-form on $W_{4}/B_{3}$. We showed
this necessity on the $F$-theory side in \cite{Clemens-1} by tracing
the Tate form 
\begin{equation}
wy^{2}=x^{3}+a_{5}wxy+a_{4}zwx^{2}+a_{3}z^{2}w^{2}y+a_{2}z^{3}w^{2}x+a_{0}z^{5}w^{3}\label{eq:Tate}
\end{equation}
back to its $E_{8}$-origins, namely 
\begin{equation}
wy^{2}=x^{3}+a_{0}z^{5}.\label{eq:E8 sing}
\end{equation}
We showed in \cite{Clemens-2} that the coefficients $a_{j}$ of the
Tate form, as well as $z$ and $y/x$ must go to minus themselves
under the $\mathbb{Z}_{2}$-action. One identifies the configuration
of exceptional curves in the crepant resolution of (\ref{eq:E8 sing})
with the configuration of the positive simple roots in the $E_{8}$-Dynkin
diagram. A consequence of this identification is that (\ref{eq:-y})
sends each positive simple root to its negative. We preserve $E_{8}$-symmetry
by counteracting this reversal of roots by the operation of complex
conjugation on the complex algebraic group $E_{8}^{\mathbb{C}}$,
an operation that leaves untouched the compact real form $E_{8}$.
This last is reflected in the fact that tracing real roots back to
$E_{8}$ requires that \eqref{eq:Tate} and \eqref{eq:E8 sing} be
rescaled by dividing both sides by $a_{0}^{6}$. This rescales both
$z$ and $y/x$ by $a_{0}^{-1}$ and $y/a_{0}x$ indeed becomes invariant
under the $\mathbb{Z}_{2}$-action.

\subsection{The organization of the paper}

In Section \ref{sec:-and-its} we derive the Weierstrass equation
for the Tate form. We introduce the construction of a second section
$\tau$ of $W_{4}/B_{3}$ in addition to the tautological section
$\zeta$ `at infinity.' We discuss the action of the involution $\tilde{\beta}_{4}/\beta_{3}$
as reflecting complex conjugation on the complex algebraic groups
$SL\left(5,\mathbb{C}\right)$ and $E_{8}^{\mathbb{C}}$ whose compact
real forms are $SU\left(5\right)$ and $E_{8}$ respectively. 

In Section \ref{sec:The-spectral-divisor} we derive the spectral
variety from the Tate form and discuss its decomposition into components
of degree four and one respectively.

In Sections \ref{sec:Semi-stable-degeneration-and} and \ref{sec:NC}
we discuss the semi-stable limit and the relation to the Heterotic
dual. In particular we show how to build the normal-crossing $K3$
from an elliptic curve with two flat $E$-bundles.

In Section \ref{sec:Geometric-Model-Double-Cover} we project the
fourfold $W_{4}/B_{3}$ to the $\mathbb{P}^{1}$-bundle $Q/B_{3}$
whose fiber is the degree-$2$ linear series determined by the two
points $\tau\left(b_{3}\right)$ and $\zeta\left(b_{3}\right)$ where
the two sections intersect the fiber. We do this by projecting each
fiber of $W_{4}/B_{3}$ from the third point of intersection with
the line joining $\tau\left(b_{3}\right)$ and $\zeta\left(b_{3}\right)$
. This projection, that we name $\overline{W}_{4}$ blows up the singular
locus of $W_{4}$. In this way we obtain a commutative diagram 
\[
\begin{array}{ccc}
\overline{W}_{4} & \rightarrow & W_{4}\\
\downarrow &  & \downarrow\\
Q & \rightarrow & B_{3}
\end{array}
\]
where the top horizontal map is crepant partial resolution, the left-hand
vertical map is $2-1$ and the bottom horizontal map is a $\mathbb{P}^{1}$-fibration
whose fibers correspond to the degree-$2$ linear system on the fibers
of the right-hand vertical map determined by the lifts of the two
sections $\left(\tau\right)$ and $\left(\zeta\right)$ to $\overline{W}_{4}/B_{3}$.
The Calabi-Yau fourfold $\overline{W}_{4}$ is a branched double cover
of $Q$. The affine fiber coordinate $\vartheta_{0}$ of $Q/B_{3}$
is a section of $K_{B_{3}}^{-1}$ as are the coefficients $a_{j}$
in the Tate form and $t:=y/x$.

We next construct a full crepant resolution of $\overline{W}_{4}$
to obtain a smooth model $\tilde{W}_{4}$. This is accomplished in
Section \ref{sec:Desingularization-of}. The first step is the partial
resolution $W_{4}^{\left(1\right)}/B_{3}$ of $\overline{W}_{4}/B_{3}$
that will create an divisor $D_{0}$ connecting the divisors $D_{1}$
(the inherited component from $W_{4}/B_{3}$) and $D_{4}$ (the exceptional
divisor of $\overline{W}_{4}/W_{4})$. Together these three divisors
comprise 
\[
S_{\mathrm{GUT}}\times_{B_{3}}W_{4}^{\left(1\right)}.
\]
This will be followed by a third partial desingularization $W_{4}^{\left(2\right)}/B_{3}$
of $W_{4}^{\left(1\right)}/B_{3}$ extending over a general point
of $S_{\mathrm{GUT}}$. Its exceptional divisor will be reducible
so that 
\[
S_{\mathrm{GUT}}\times_{B_{3}}W_{4}^{\left(2\right)}=D_{0}\cup D_{1}\cup D_{2}\cup D_{3}\cup D_{4}
\]
that configure themselves over a general point of $S_{\mathrm{GUT}}$
as an extended $A_{4}$-Dynkin diagram. That will in turn be followed
by a final codimension-$2$ desingularization $\tilde{W}_{4}/B_{3}$
of $W_{4}^{\left(2\right)}/B_{3}$ over the curve $\Sigma_{\mathbf{\bar{5}}}^{\left(44\right)}\subseteq S_{\mathrm{GUT}}$.

It should be noted that, in the process of putting $\left(\zeta\right)$
and $\left(\tau\right)$ on equal footing as the first step in the
desingularization, neither can be given preference as the one passing
through the inherited component 
\[
S_{\mathrm{GUT}}\times_{B_{3}}\overline{W}_{4}.
\]
The `inherited' role is assumed by $D_{0}$ while, over a general
point of $S_{\mathrm{GUT}}$, the proper transform $\left(\tilde{\zeta}\right)$
of $\left(\zeta\right)$ intersects $D_{1}$ and the proper transform
$\left(\tilde{\tau}\right)$ of $\left(\tau\right)$ intersects $D_{4}$.

In Section \ref{sec:Higgs-line-bundle} we calculate the Higgs line
bundle that will govern the computation of the chiral spectrum. In
Subsections \ref{subsec:The--G-flux-in} and \ref{subsec:Numerical-conditions-on}
we compute the $G$-flux. In Subsection \ref{sec:The-D-term} we establish
the vanishing of the $D$-term for a suitably chosen $\tilde{\beta}$-symmetric
Kähler metric.

In Section \ref{subsec:Symmetry-breaking} we discuss the symmetry-breaking
induced by wrapping the Wilson line on the involution $\tilde{\beta}_{4}$
on $\tilde{W}_{4}$. The Wilson line breaks $SU(5)$-symmetry to the
$SU(3)_{C}\times SU(2)_{L}\times U(1)_{Y}$ symmetry of the Standard
Model.

It is in Section \ref{sec:Skew-twisted-cohomology} that we calculate
the complete spectrum of the theory. The desired invariants follow
rather directly from the results of \cite{Clemens-2} where the toric
presentation of $B_{3}$ and its involution $\beta_{3}$ were explored
in detail. The present paper demonstrates a method for eliminating
vector-like exotics when breaking the GUT symmetry with the Wilson
line. In this section we discuss just one $SU(5)_{gauge}\times SU(5)_{Higgs}$
sector in the semi-stable limit of the $F$-theory model, however,
the massless spectrum in the hidden $SU(5)_{gauge}\times SU(5)_{Higgs}$
is identical.

Finally in Section \ref{sec:Asymptotic--R-symmetry} discuss the $\mathbb{Z}_{4}$
R-symmetry on the semi-stable limit of our $F$-theory model.
\begin{rem}
Throughout this paper, we will employ the following notational convention.
Projections with image space $A$ will in general be denoted as $\pi_{A}$.
This notation will be employed regardless of the domain of the map,
which (hopefully) will be clear from the context. 
\end{rem}

\begin{rem}
Throughout this paper, we will let 
\[
\mathbb{P}_{\left[i_{1},\ldots,,i_{d}\right]}^{d-1}
\]
denote the weighted complex projective $\left(d-1\right)$-space with
integer weights $\left[i_{1},\ldots,,i_{d}\right]$ and we will let
\[
\mathbb{P}_{\left[u_{1},\ldots,u_{d}\right]}
\]
denote the (unweighted) complex projective space with homogeneous
coordinates $\left[u_{1},\ldots,,u_{d}\right]$. 
\end{rem}

\section{$W_{4}/B_{3}$, $E_{8}$-unfolding and its symmetries\label{sec:-and-its}}

Our starting point is the Fano threefold 
\[
B_{3}=\mathbb{P}_{\left[u_{0},v_{0}\right]}\times B_{2}
\]
where $B_{2}$ is the $D_{2}$ del Pezzo surface studied in Section
6 of \cite{Clemens-2}. $B_{2}$ is a double cover of the projective
plane 
\begin{equation}
D_{2}\rightarrow\mathbb{P}_{\left[n_{-1},m_{i},m_{-i}\right]}\label{eq:D2}
\end{equation}
branched along a specific smooth quartic curve admitting a $\mathbb{Z}_{4}$-action.
Denote 
\[
N:=c_{1}\left(K_{B_{3}}^{-1}\right).
\]

We form 
\begin{equation}
P:=\mathbb{P}\left(\mathcal{O}_{B_{3}}\oplus\mathcal{O}_{B_{3}}\left(2N\right)\oplus\mathcal{O}_{B_{3}}\left(3N\right)\right)\label{eq:P}
\end{equation}
with homogeneous fiber coordinates $\left[w,x,y\right]$ and canonical
bundle 
\[
\mathcal{O}_{P}\left(-3\right)\otimes\pi_{B_{3}}^{\ast}\left(-6N\right).
\]
For 
\begin{equation}
a_{j},\,z,\,\frac{y}{x}=t\in H^{0}\left(\mathcal{O}_{B_{3}}\left(N\right)\right),\label{eq:divcl}
\end{equation}
we write the Tate form (\ref{eq:Tate}) for an elliptically fibered
fourfold $W_{4}/B_{3}$ in determinantal form as 
\begin{equation}
\left|\begin{array}{cc}
x^{3}+a_{4}zwx^{2}+a_{2}z^{3}w^{2}x+a_{0}z^{5}w^{3} & 1\\
wy^{2}-\left(a_{5}wx+a_{3}z^{2}w^{2}\right)y & 1
\end{array}\right|=0.\label{eq:T}
\end{equation}
Since 
\[
K_{P}^{-1}=\mathcal{O}_{P}\left(3\left(\left\{ w=0\right\} +6N\right)\right),
\]
$W_{4}$ is Calabi-Yau. Here as in \cite{Clemens-2} $B_{3}$ admits
an involution $\beta_{3}$ with finite fixpoint set with respect to
which 
\[
t:=\frac{y}{x},z,a_{0},a_{2},a_{3},a_{4},a_{5}\in H^{0}\left(\mathcal{O}_{B_{3}}\left(N\right)\right)^{\left[-1\right]},
\]
the $\left(-1\right)$-eigenspace, and 
\[
\begin{array}{c}
w\in H^{0}\left(\mathcal{O}_{P}\left(1\right)\otimes\pi_{B_{3}}^{\ast}\mathcal{O}_{B_{3}}\right)^{\left[+1\right]}\\
x\in H^{0}\left(\mathcal{O}_{P}\left(1\right)\otimes\pi_{B_{3}}^{\ast}\mathcal{O}_{B_{3}}\left(2N\right)\right)^{\left[+1\right]}\\
y\in H^{0}\left(\mathcal{O}_{P}\left(1\right)\otimes\pi_{B_{3}}^{\ast}\mathcal{O}_{B_{3}}\left(3N\right)\right)^{\left[-1\right]}.
\end{array}
\]
with respect to the induced involution on $W_{4}/B_{3}$. The equation
$z=0$ defines a smooth surface that, by the adjunction formula, must
be a $K3$-surface.

\subsection{Weierstrass form with respect to $\zeta$ \label{subsec:Weierstrass-form-for}}

Let 
\[
\zeta:B_{3}\rightarrow W_{4}
\]
be the standard section given by 
\[
\zeta\left(b_{3}\right)=\left\{ \left[w,x,y\right]=\left[0,0,1\right]\right\} .
\]
Referring to (\ref{eq:T}) we change the equation of $W_{4}/B_{3}$
into Weierstrass form based at the section $\zeta$ in the standard
way. Namely we complete the square with respect to $y$ as follows.
\[
\begin{array}{c}
wy^{2}=x^{3}+a_{5}wxy+a_{4}zwx^{2}+a_{3}z^{2}w^{2}y+a_{2}z^{3}w^{2}x+a_{0}z^{5}w^{3}\\
w\left(y^{2}-\left(a_{5}x+a_{3}z^{2}w\right)y+\frac{\left(a_{5}x+a_{3}z^{2}w\right)^{2}}{4}\right)=\\
x^{3}+a_{4}zwx^{2}+a_{2}z^{3}w^{2}x+a_{0}z^{5}w^{3}+w\frac{\left(a_{5}x+a_{3}z^{2}w\right)^{2}}{4}\\
w\left(y-\frac{a_{5}x+a_{3}z^{2}w}{2}\right)^{2}=x^{3}+\left(a_{4}z+\frac{a_{5}^{2}}{4}\right)wx^{2}+z^{2}\left(a_{2}z+\frac{a_{3}a_{5}}{2}\right)w^{2}x+z^{4}\left(a_{0}z+\frac{a_{3}^{2}}{4}\right)w^{3}\\
=x^{3}+Awx^{2}+Bz^{2}w^{2}x+Cz^{4}w^{3}
\end{array}
\]
where 
\[
\begin{array}{c}
A=a_{4}z+\frac{a_{5}^{2}}{4}\\
B=a_{2}z+\frac{a_{3}a_{5}}{2}\\
C=a_{0}z+\frac{a_{3}^{2}}{4}.
\end{array}
\]
Then we eliminate the $x^{2}$-term by setting 
\begin{equation}
\begin{array}{c}
\underline{x}=x+\frac{a_{5}^{2}+4a_{4}z}{12}w=x+\frac{A}{3}w\\
\underline{y}=y-\frac{a_{5}x+a_{3}z^{2}w}{2}
\end{array}\label{eq:Wnfcoor}
\end{equation}
finally yielding the Weierstass form 
\begin{equation}
w\underline{y}^{2}=\underline{x}^{3}+\left(Bz^{2}-\frac{A^{2}}{3}\right)w^{2}\underline{x}+\left(Cz^{4}-\frac{AB}{3}z^{2}+\frac{2A^{3}}{27}\right)w^{3}.\label{eq:Wnf}
\end{equation}
for our Calabi-Yau fourfold $W_{4}$.

The discriminant of (\ref{eq:Wnf}) is given by 
\[
\begin{array}{c}
4\left(Bz^{2}-\frac{A^{2}}{3}\right)^{3}+27\left(Cz^{4}-\frac{AB}{3}z^{2}+\frac{2A^{3}}{27}\right)^{2}=\\
A^{2}\left(4AC-B^{2}\right)z^{4}+2B\left(2B^{2}-9AC\right)z^{6}+27C^{2}z^{8}.
\end{array}
\]
Expanding the discriminant in powers of $z$ the coefficient of $z^{4}$
becomes $\left(\frac{a_{5}^{2}}{4}\right)^{2}\left(\frac{a_{3}^{2}a_{5}^{2}}{4}-\left(\frac{a_{3}a_{5}}{2}\right)^{2}\right)=0$
and the coefficient of $z^{5}$ does not in general equal zero so
that indeed, by Kodaira's classification, we have $A_{4}$-singularities
over $S_{\mathrm{GUT}}=\left\{ z=0\right\} $.

\subsection{The second section\label{subsec:The-second-section}}

Beside the standard section $\zeta:B_{3}\rightarrow W_{4}$ we require
a second section 
\[
\tau:B_{3}\rightarrow W_{4}
\]
defined by 
\[
b_{3}\mapsto\left[w,x,y\right]=\left[1,z^{2},z^{3}\right].
\]
Substituting this section of $P/B_{3}$ into Tate form (\ref{eq:Tate}),
as in \cite{Clemens-1,Clemens-2} one concludes that the condition
that it lies in $W_{4}$ is 
\begin{equation}
a_{5}+a_{4}+a_{3}+a_{2}+a_{0}=0.\label{eq:cond}
\end{equation}
This section allows a change of the group structure on the fibers
of $W_{4}/B_{3}$ by a translation. Intertwining this translation
with the action of $\tilde{\beta}_{4}$ allows us both to eliminate
vector-like exotics from the $F$-theory model $W_{4}^{\vee}/B_{3}^{\vee}$
and to introduce a $\left(4+1\right)$-split in the spectral divisor
giving the $U\left(1\right)_{X}$ discussed in the Introduction. Translation
of fibers by this section of course leaves the Weierstrass form on
smooth elliptic fibers and $I_{1}$-fibers invariant, that is, all
fibers over $\left(B_{3}-S_{\mathrm{GUT}}\right)$.

Notice that, over a general point $b_{3}\in B_{3}$, $\tau-\zeta$
is not of finite order on $Pic^{0}$ of the cuspidal curve $wy^{2}=x^{3}$
since the parameter $t=\frac{y}{x}=z$ takes all values. So $\tau-\zeta$
is not of finite order in $Pic^{0}\left(\pi^{-1}\left(b_{3}\right)\right)$
if the $a_{j}$ are sufficiently small. Furthermore, if $a_{5}=-a_{0}$
are small and $a_{2}=a_{3}=a_{4}=0$, the same argument gives
\begin{lem}
For a crepant resolution $\tilde{W}_{4}/B_{3}$ of $W_{4}/B_{3}$,
$\tau-\zeta$ is not of finite order in $Pic^{0}\left(\tilde{W}_{4}/B_{3}\right)$
over a general point of $S_{\mathrm{GUT}}:=\left\{ z=0\right\} \subseteq B_{3}$. 
\end{lem}

So these same assertions hold for a general allowable choice of the
coefficients $a_{j}$. This fact is essential to the proof in \cite{Clemens-2}
of Lemma \ref{lem:With-respect-to}ii) below. It says that our $F$-theory
model has a single Higgs doublet.

Over $b_{3}\in B_{3}$ the line in $\pi^{-1}\left(b_{3}\right)$ between
$\zeta\left(b_{3}\right)$ and $\tau\left(b_{3}\right)$ is given
by the equation $x-z^{2}w=0$. Letting $a_{jkl}:=a_{j}+a_{k}+a_{l}$
and using that $a_{420}=-a_{53}$ we obtain by (\ref{eq:cond}) that
the third point of intersection with $W_{4}$ is given by substituting
$x=z^{2}w$ in (\ref{eq:T}) to obtain 
\[
w\left|\begin{array}{cc}
z^{6}w^{2}+a_{420}z^{5}w^{2} & 1\\
y^{2}+a_{420}z^{2}wy & 1
\end{array}\right|=0,
\]
that is, 
\[
y^{2}+a_{420}z^{2}wy-z^{5}w^{2}\left(z+a_{420}\right)=\left(y-z^{3}w\right)\left(y+\left(z+a_{420}\right)z^{2}w\right)=0.
\]
We denote the third section as 
\[
\upsilon\left(b_{3}\right)=\left[w,z^{2}w,\left(-z-a_{420}\right)z^{2}w\right].
\]

Finally, we therefore have a fourth section 
\[
\mu:B_{3}\rightarrow W_{4}
\]
defined by the third point of intersection of the tangent line to
\begin{equation}
\left|\begin{array}{cc}
x^{3}+a_{4}zwx^{2}+a_{2}z{}^{3}w^{2}x+a_{0}z{}^{5}w^{3} & 1\\
wy^{2}-\left(a_{5}wx+a_{3}z{}^{2}w^{2}\right)y & 1
\end{array}\right|=0\label{eq:3rdsec}
\end{equation}
at $\upsilon\left(b_{3}\right)$.

Taken together these calculations yield the following. 
\begin{lem}
\label{lem:The-sections-,Mordell}The sections $\zeta$, $\tau$,
$\upsilon$ and $\mu$ of $W_{4}/B_{3}$ satisfy the following relations
in $\mathrm{Pic}{}^{3}\left(W_{4}/B_{3}\right)$: 
\[
\begin{array}{c}
3\text{·}\zeta\in\left|\mathcal{O}_{W_{4}/B_{3}}\left(1\right)\right|\\
\zeta+\tau+\upsilon\in\left|\mathcal{O}_{W_{4}/B_{3}}\left(1\right)\right|\\
2\upsilon+\mu\in\left|\mathcal{O}_{W_{4}/B_{3}}\left(1\right)\right|.
\end{array}
\]
The last two relations imply that 
\[
\zeta+\tau\equiv\upsilon+\mu,
\]
that is in classical language, the two divisors are members of the
same distinguished $g_{2}^{1}$ on $W_{4}/B_{3}$. 
\end{lem}

\subsection{The quotient Calabi-Yau manifolds}

In order to wrap a Wilson line, we must require that involution $\tilde{\beta}_{4}/\beta_{3}$
in \eqref{b4} with quotient $W_{4}^{\vee}/B_{3}^{\vee}$ be such
that $\beta_{3}$ acts fixpoint-freely on the smooth anti-canonical
divisor $S_{\mathrm{GUT}}\subseteq B_{3}$ yielding an Enriques surface
\[
S_{\mathrm{GUT}}^{\vee}\subseteq B_{3}^{\vee}
\]
and that the induced involution $\tilde{\beta}_{3}/\beta_{2}$ on
the dual Heterotic with quotient $V_{3}^{\vee}/B_{2}^{\vee}$ be such
that $\tilde{\beta}_{3}$ acts freely. In order that $W_{4}^{\vee}$
be Calabi-Yau, the involution $\beta_{3}$ on $B_{3}=\mathbb{P}_{\left[u_{0},v_{0}\right]}\times B_{2}$
must have only finite fixpoint-set since $S_{\mathrm{GUT}}$ is an
ample divisor in the Fano manifold $B_{3}$. This forces $\beta_{3}$
to act as $\left(-1\right)$ on the meromorphic two-form on $B_{3}$
with pole on $S_{\mathrm{GUT}}$. That in turn forces $\tilde{\beta}_{4}$
with quotient $W_{4}^{\vee}/B_{3}^{\vee}$ to act as $\left(-1\right)$
on the relative one-form $dx/y$ on $W_{4}/B_{3}$ since otherwise
$W_{4}^{\vee}$ would not be Calabi-Yau. On the other hand, the induced
involution $\tilde{\beta}_{3}/\beta_{2}$ on the Heterotic threefold
$V_{3}/B_{2}$ must act as $\left(+1\right)$ on the relative one-form
$dx/y$ since otherwise $V_{3}^{\vee}$ would not be Calabi-Yau. The
possibility, even necessity, of the sign-reversal is explained in
\cite{Clemens-1}.

\subsection{Three Calabi-Yau fourfolds related by quotienting }

Following \cite{Clemens-2} we use root systems and toric geometry
to actually define three base threefolds and associated Calabi-Yau
fourfolds that we denote by 
\[
\begin{array}{c}
W_{4}^{\wedge}\rightarrow B_{3}^{\wedge}\\
W_{4}\longrightarrow B_{3}=\mathbb{P}_{\left[u_{0},v_{0}\right]}\times B_{2}\\
W_{4}^{\vee}\longrightarrow B_{3}^{\vee}:=B_{3}/\left\{ C_{u,v}\right\} .
\end{array}
\]
$B_{3}^{\wedge}$ is defined to be the resolution of the graph of
the Cremona involution on $\mathbb{P}\left(\mathfrak{h}_{SU\left(5\right)}\right)$
with respect to a basis given by the choice of a system of simple
roots of $SU\left(5\right)$ balanced between positive and negative
Weyl chambers. We have reserved the least cumbersome notation for
the intermediate one $B_{3}$ because, as we have already mentioned,
it is computationally most convenient to work in that setting.

On the $F$-theory side our ultimate target is the (orbifold) Calabi-Yau
fourfold $W_{4}^{\vee}$ with smooth Heterotic dual Calabi-Yau threefold
$V_{3}^{\vee}$ having two bundles with Yang-Mills connections with
structure group symmetry-breaking 
\[
E_{8}\overset{Tate}{\Longrightarrow}SU\left(5\right)_{gauge}\times SU\left(5\right)_{Higgs}\overset{Higgs}{\Longrightarrow}SU\left(5\right)_{gauge}\overset{Wilson-line}{\Longrightarrow}SU(3)_{C}\times SU(2)_{L}\times U(1)_{Y},
\]
as described in \cite{Clemens-2}. Notice that the initial fourfolds
and the involutions do not involve any choice of Weyl chamber. It
is only the crepant resolutions of the Calabi-Yau fourfolds and threefolds
that imply such choices. This is explained in \cite{Clemens-1}.

In order to wrap a Wilson line, we require that $W_{4}/B_{3}$ given
in \eqref{eq:Tate} admit an equivariant involution $\tilde{\beta}_{4}/\beta_{3}$
with quotient $W_{4}^{\vee}/B_{3}^{\vee}$ such that $\beta_{3}$
acts fixpoint-free on the smooth anti-canonical divisor $S_{\mathrm{GUT}}\subseteq B_{3}$
yielding an Enriques surface 
\[
S_{\mathrm{GUT}}^{\vee}\subseteq B_{3}^{\vee}.
\]
In order that $W_{4}^{\vee}$ be Calabi-Yau, the involution $\beta_{3}$
on $B_{3}=\mathbb{P}_{\left[u_{0},v_{0}\right]}\times B_{2}$ must
have only finite fixpoint-set since $S_{\mathrm{GUT}}$ is an ample
divisor in the Fano manifold $B_{3}$. As explained in \cite{Clemens-1}
this forces $\beta_{3}$ to act as $\left(-1\right)$ on the meromorphic
two-form on $B_{3}$ with pole on $S_{\mathrm{GUT}}$.

\subsection{Unfolding the $E_{8}$-singularity}

A basic principle in the mathematics of String Theory is that the
geometric model \eqref{eq:Tate} of $F$-theory must be considered
as having evolved according to the unfolding of the $E_{8}$-surface
singularity 
\[
wy^{2}=x^{3}+a_{0}z^{5}.
\]

In \cite{Clemens-1} we have observed that principle to the letter,
tracing the equivariant crepant resolution implicit in the Tate form
(\ref{eq:Tate}) back to the Brieskorn-Grothendieck equivariant crepant
resolution \cite{Brieskorn,Slodowy} of the semi-universal deformation
of the rational double point singularity \eqref{eq:E8 sing} by requiring
that the section defining $S_{\mathrm{GUT}}$ be given by a formula
\[
z=\sum_{j=2}^{5}\kappa_{j}a_{j}
\]
for generic $\kappa_{j}$.

The assumption \eqref{eq:cond} will, as we will see, reduces the
maximal subgroup decomposition 
\[
\frac{SU\left(5\right)_{gauge}\times SU\left(5\right)_{Higgs}}{\mathbb{Z}_{5}}\subseteq E_{8}
\]
of \cite{Clemens-1} to 
\[
SU\left(5\right)_{gauge}\times U\left(1\right)_{X}\times SU\left(4\right)_{Higgs}\subseteq E_{8}
\]
so that, on the $F$-theory side, one begins with the identification
of maximal tori compatible with the three-dimensional commutative
diagram obtained by pasting the top and bottom morphisms of 
\begin{equation}
\begin{array}{ccc}
\dot{SL}\left(5;\mathbb{C}\right)_{gauge}\times\dot{SL}\left(4;\mathbb{C}\right)_{Higgs}\times\mathbb{C}^{\ast} & \overset{\dot{\kappa}}{\hookrightarrow} & \dot{E}_{8}^{\mathbb{C}}\\
\uparrow &  & \uparrow\\
SU\left(5\right)_{gauge}\times SU\left(4\right)_{Higgs}\times U\left(1\right)_{X} & \hookrightarrow & E_{8}\\
\downarrow &  & \downarrow\\
\ddot{SL}\left(5;\mathbb{C}\right)_{gauge}\times\ddot{SL}\left(4;\mathbb{C}\right)_{Higgs}\times\mathbb{C}^{\ast} & \overset{\ddot{\kappa}}{\hookrightarrow} & \ddot{E}_{8}^{\mathbb{C}}
\end{array}\label{eq:spin}
\end{equation}
to the top and bottom morphisms, respectively, of the commuting diagram
\begin{equation}
\begin{array}{ccc}
\dot{SL}\left(5;\mathbb{C}\right)_{gauge}\times\dot{SL}\left(4;\mathbb{C}\right)_{Higgs}\times\mathbb{C}^{\ast} & \overset{\dot{\kappa}}{\hookrightarrow} & \dot{E}_{8}^{\mathbb{C}}\\
\updownarrow\iota &  & \updownarrow\iota\\
\ddot{SL}\left(5;\mathbb{C}\right)_{gauge}\times\ddot{SL}\left(4;\mathbb{C}\right)_{Higgs}\times\mathbb{C}^{\ast} & \overset{\ddot{\kappa}}{\hookrightarrow} & \ddot{E}_{8}^{\mathbb{C}}
\end{array}\label{eq:iota}
\end{equation}
of real analytic outer complex conjugation involutions. \cite{Clemens-1}
and \cite{Clemens-2} are built around the necessity of a choice of
a positive-negative pair of Weyl chambers of $E_{8}$ with the requirement
that every step of constructions must commute with the passage between
these two chambers. In particular, as we have shown in \cite{Clemens-1},
the equivariant crepant resolution $\tilde{W}_{4}/B_{3}$ of $W_{4}/B_{3}$
depends on the choice of Weyl chamber.

This means that we will have two copies of the crepantly resolved
$\tilde{W}_{4}/B_{3}$ that we designate by letting $\dot{W}_{4}/B_{3}$
denote the $F$-theory model with a choice of positive chamber and
$\ddot{W}_{4}/B_{3}$ with its negative as the choice of positive
Weyl chamber. The action of $\tilde{\beta}_{4}/\beta_{3}$ on $\dot{W}_{4}/B_{3}$
will be a holomorphic involution that acts on roots as the longest
element of the Weyl group $W\left(SU\left(5\right)\right)$ yielding
a Calabi-Yau quotient, and similarly for the quotient of the action
on $\ddot{W}_{4}/B_{3}$. However these quotients are only real-analytically
equivalent, not complex-analytically equivalent. As explained in \cite{Clemens-2}
the exceptional curves over the quotient $S_{\mathrm{GUT}}^{\vee}$
are `flopped' when passing from one to the other. The flop is essentially
invisible on the Heterotic side since it tracks only the real $E_{8}$-bundles
and the `flop' becomes the passage between the two possible complexifications
of the same real $E_{8}$-bundle.

Identifying exceptional components over $S_{\mathrm{GUT}}$ with positive
simple roots forces the involution $\tilde{\beta}_{4}$ to act as
the non-trivial involution on the $A_{4}$-Dynkin diagram, that is,
by the longest element of the Weyl group $W\left(SU\left(5\right)\right)$
on the exceptional divisors of the crepant resolution of $W_{4}/B_{3}$
. Again as shown in \cite{Clemens-1}, it is the commutativity of
the geometric involutions $\tilde{\beta}_{4}/\beta_{3}$ and $\tilde{\beta}_{3}/\beta_{2}$
with the complex conjugate involution $\iota$ in \eqref{eq:iota}
that allows us to incorporate both in the simultaneous quotienting
on both the Heterotic and $F$-theory models that preserves initial
$E_{8}$-symmetry and subsequent $SU\left(5\right)$-symmetry since
\eqref{eq:iota} acts trivially on $SU\left(5\right)$ and $E_{8}$.

\section{The spectral divisor\label{sec:The-spectral-divisor}}
\begin{flushleft}
The crepant resolution $\tilde{W}_{4}/B_{3}$ of $W_{4}/B_{3}$ will
have $I_{5}$-type fibers over generic points of 
\[
S_{\mathrm{GUT}}:=\left\{ z=0\right\} \subseteq B_{3}.
\]
This $I_{5}$-fibration over $S_{\mathrm{GUT}}$ carries the $SU\left(5\right)_{gauge}$-symmetry.
On the other hand, $SU\left(5\right)_{Higgs}$-symmetry is broken
on a five-sheeted branched covering of $B_{3}$ given by the lift
of 
\begin{equation}
\mathcal{C}_{Higgs}:=W_{4}\text{·}\left(\left\{ wy^{2}=x^{3}\right\} -\left\{ w^{4}=0\right\} \right)\label{eq:specdiv}
\end{equation}
to a divisor $\tilde{\mathcal{C}}_{Higgs}\subseteq\tilde{W}_{4}$.
Its symmetry is broken by assigning non-trivial eigenvalues to the
fundamental representation $SU\left(5\right)_{Higgs}$ using the spectral
construction with respect to the push-forward to $B_{3}$ of a line
bundle $\mathcal{L}_{Higgs}$ on $\tilde{\mathcal{C}}_{Higgs}$. We
see this as follows. 
\par\end{flushleft}

\begin{flushleft}
We form 
\[
\hat{P}:=\mathbb{P}\left(\mathcal{O}_{B_{3}}\oplus\mathcal{O}_{B_{3}}\left(N\right)\oplus\mathcal{O}_{B_{3}}\left(2N\right)\oplus\mathcal{O}_{B_{3}}\left(3N\right)\right)
\]
using the fiber coordinate $t$ for $\mathcal{O}_{B_{3}}\left(N\right)$.
The natural projection 
\begin{equation}
\hat{P}\dashrightarrow P\label{eq:ratmp}
\end{equation}
with center $\mathbb{P}\left(\mathcal{O}_{B_{3}}\left(N\right)\right)$
is defined except along the section where $x=y=w=0$. $\hat{P}$ contains
the smooth five-dimensional incidence hypersurface $Y$ given by the
equation 
\[
\left|\begin{array}{cc}
x & y\\
w & t
\end{array}\right|=0,
\]
thereby forming a smooth quadric hypersurface over $B_{3}$ with distinguished
section given by $x=y=w=0$. So the restriction of (\ref{eq:ratmp})
to $Y$ is defined except along the section where it spreads the exceptional
locus over the linear locus $\left\{ x=0\right\} $ in $P$. The result
is a birational morphism 
\[
\hat{Y}\rightarrow P
\]
that maps the exceptional locus over the section isomorphically onto
$\left\{ x=0\right\} \subseteq P$ and blows down the linear loci
$\left\{ x=y=0\right\} \subseteq Y$ and $\left\{ x=w=0\right\} \subseteq Y$. 
\par\end{flushleft}

\begin{flushleft}
Ignoring what happens over $\left\{ w=0\right\} $, that is, setting
$w=1$ and using $t$ as the affine fiber we obtain the equation 
\[
t^{2}x^{2}=x^{3}+a_{5}x^{2}t+a_{4}zx^{2}+a_{3}z^{2}xt+a_{2}z^{3}x+a_{0}z^{5}
\]
as the defining equation for the hypersurface $\hat{W}_{4}-\left\{ w=0\right\} \subseteq\hat{Y}-\left\{ w=0\right\} $.
Using 
\[
\begin{array}{c}
x=wt^{2}\\
y=wt^{3}
\end{array}
\]
the divisor given by the intersection with $\hat{W}_{4}$ then has
equation 
\begin{equation}
a_{5}t^{5}+a_{4}t^{4}z+a_{3}t^{3}z^{2}+a_{2}t^{2}z^{3}+a_{0}z^{5}=0.\label{eq:speceq2-1}
\end{equation}
This is the equation on the affine set $w=1$ of the intersection
of the proper transform of $W_{4}$ with the locus given by 
\begin{equation}
\left\{ y^{2}=x^{3}\right\} \label{eq:spec}
\end{equation}
and is called the \textit{spectral divisor}. The spectral divisor,
in particular, contains the singular locus of $W_{4}$. The condition
\eqref{eq:cond} implies that homogeneous form in \eqref{eq:speceq2-1}
is divisible by $z-t$, that is, the spectral divisor admits a $\left(4+1\right)$
factorization. The involution $\tilde{\beta}_{4}/\beta_{3}$ takes
\[
\left(z-t\right)\mapsto\left(t-z\right)
\]
and leaves \eqref{eq:speceq2-1} invariant. 
\par\end{flushleft}

\begin{flushleft}
Said otherwise, since $W_{4}$ is smooth except over $\left\{ z=0\right\} \subseteq B_{3}$,
the proper transform $\hat{W}_{4}\subseteq\hat{Y}$ of $W_{4}$ blows
up the codimension-$2$ subvariety 
\[
\left\{ z=t=0\right\} 
\]
by forming the incidence fivefold 
\[
\tilde{Y}=\left\{ \left|\begin{array}{cc}
z & t\\
\tilde{z} & \tilde{t}
\end{array}\right|=0\right\} \subseteq Y\times\mathbb{P}_{\left[\tilde{t},\tilde{z}\right]}
\]
with exceptional locus $\mathbb{P}_{\left[\tilde{t},\tilde{z}\right]}\times S_{\mathrm{GUT}}$.
The proper transform $\hat{W}_{4}$ of $W_{4}$ intersects this exceptional
locus in the hypersurface 
\[
\mathcal{D}\subseteq\mathbb{P}_{\left[\tilde{t},\tilde{z}\right]}\times B_{3}.
\]
given by the equation 
\begin{equation}
0=a_{5}\tilde{t}^{5}+a_{4}\tilde{z}\tilde{t}^{4}+a_{3}\tilde{z}^{2}\tilde{t}^{3}+a_{2}\tilde{z}^{3}\tilde{t}^{2}+a_{0}\tilde{z}^{5}.\label{eq:speceq-2}
\end{equation}
(Compare \eqref{eq:speceq2-1} with \eqref{eq:speceq-2}.) 
\par\end{flushleft}

\begin{flushleft}
It is immediate to check that the sections $\zeta$ and $\tau$ both
lie in $W_{4}/B_{3}$ and both lie in 
\begin{equation}
SP:=\left\{ \left|\begin{array}{cc}
x^{3} & 1\\
wy^{2} & 1
\end{array}\right|=0\right\} \subseteq P.\label{eq:seccond}
\end{equation}
For fiber coordinate $t$ for $SP$ with 
\par\end{flushleft}

\[
\begin{array}{c}
wt^{2}:=x\\
wt^{3}:=y
\end{array}
\]
the value $t=z$ gives the section $\tau=\left\{ \left[w,x,y\right]=\left[1,z^{2},z^{3}\right]\right\} $
of $W_{4}\cap SP$. The cohomology class of 
\[
SP\cap W_{4}
\]
is given by 
\[
c_{1}\left(\mathcal{O}_{P}\left(3\right)\right)^{2}.
\]

One sees easily that $W_{4}$ and \eqref{eq:spec} have contact of
order $4$ along $\left(\zeta\right)$ and order $1$ along $\left(\tau\right)$.
Given $b_{3}\in\left(B_{3}-S_{\mathrm{GUT}}\right)$, we denote by
$SP_{\left(4\right)}\left(b_{3}\right)$ the other four points in
which $\left.W_{4}\right|_{\pi^{-1}\left(b_{3}\right)}$ intersects
$SP$. We denote the closure in $P$ of the locus $\cup_{b_{3}\in\left(B_{3}-S_{\mathrm{GUT}}\right)}SP_{\left(4\right)}\left(b_{3}\right)$
as 
\begin{equation}
SP_{\left(4\right)}\subseteq SP\cap W_{4}\label{eq:specvar}
\end{equation}
that since $0=a_{0}+a_{2}+a_{3}+a_{4}+a_{5}$ has equation 
\begin{equation}
\begin{array}{c}
0=a_{5}\tilde{t}^{5}+a_{4}\tilde{t}^{4}\tilde{z}+a_{3}\tilde{t}^{3}\tilde{z}^{2}+a_{2}\tilde{t}^{2}\tilde{z}^{3}+a_{0}\tilde{z}^{5}=\\
\left(a_{5}\tilde{t}^{4}+a_{54}\tilde{t}^{3}\tilde{z}-a_{20}\tilde{t}^{2}\tilde{z}^{2}-a_{0}\tilde{z}^{3}\left(\tilde{t}+\tilde{z}\right)\right)\left(\tilde{t}-\tilde{z}\right)=\\
\left(\left(\tilde{t}+\tilde{z}\right)\left(a_{5}\tilde{t}^{3}+a_{4}\tilde{t}^{2}\tilde{z}-a_{0}\tilde{z}^{3}\right)-a_{420}\tilde{t}^{2}\tilde{z}^{2}\right)\left(\tilde{t}-\tilde{z}\right)
\end{array}\label{eq:preHiggs}
\end{equation}
where $a_{jk}=a_{j}+a_{k}$, etc.\footnote{Often this equation is written in terms of the variable $s=\tilde{t}/\tilde{z}$
so that $\left(\tilde{t}-\tilde{z}\right)=0$ becomes $s-1=0$, thereby
eliminating $\mathbf{10}{}_{\left\{ -4\right\} }$ representations,
as in (70) of \cite{Blumenhagen}.} 
\begin{flushleft}
Thus the spectral divisor 
\begin{equation}
\mathcal{C}_{Higgs}=\mathcal{C}_{Higgs}^{\left(4\right)}+\mathcal{C}_{Higgs}^{\left(1\right)}\label{eq:Higgs decomp}
\end{equation}
is the image of $\mathcal{D}=\mathcal{D}^{\left(4\right)}+\mathcal{D}^{\left(1\right)}$
in $\hat{W}_{4}$.\footnote{By \eqref{eq:sGUT} below, the $\mathcal{D}^{\left(1\right)}$-component
projects to the invariant global section $\upsilon$ of $W_{4}/B_{3}$
given in Lemma \ref{lem:The-sections-,Mordell}, therefore yielding
a global $U\left(1\right)_{X}$. } Thus the involution $\tilde{\beta}_{4}/\beta_{3}$ preserves \eqref{eq:Higgs decomp}.
Using that $a_{54320}=0$, the Higgs curve, that will be important
throught this paper, is defined as the image in $S_{\mathrm{GUT}}$
of the common solutions to the two equations 
\begin{equation}
\begin{array}{c}
a_{5}\tilde{t}^{4}-a_{20}\tilde{t}^{2}\tilde{z}^{2}-a_{0}\tilde{z}^{4}=0\\
a_{54}\tilde{t}^{2}-a_{0}\tilde{z}^{2}=0.
\end{array}\label{eq:Higgspush}
\end{equation}
Writing \eqref{eq:Higgspush} as two equations in the variable $\tilde{t}^{2}/\tilde{z}^{2}$,
the solution set doubly covers the surface in $B_{3}$ defined by
the resultant equation obtained by substituting 
\[
\frac{\tilde{t}^{2}}{\tilde{z}^{2}}=\frac{a_{0}}{a_{54}}
\]
in the first equation to obtain 
\[
a_{5}\left(\frac{a_{0}}{a_{54}}\right)^{2}-a_{20}\left(\frac{a_{0}}{a_{54}}\right)-a_{0}=0
\]
that, again using $a_{54320}=0$, reduces to 
\begin{equation}
a_{0}\text{·}\left|\begin{array}{cc}
a_{4} & -a_{5}\\
a_{3}+a_{0} & a_{3}
\end{array}\right|=0\label{eq:receive}
\end{equation}
with branch locus defined by the restriction of the divisor class
$N$. 
\par\end{flushleft}

\subsection{Adjusting the $E_{8}$-evolution}

Because the image $\left(\tau\right)$ of our second section $\tau$
actually coincides with $\mathcal{C}_{Higgs}^{\left(1\right)}$ in
our $F$-theory model $SU\left(5\right)_{gauge}\times SU\left(5\right)_{Higgs}$
will be actually replaced by

\begin{equation}
SU\left(5\right)_{gauge}\times\left(SU\left(4\right)_{Higgs}\times U\left(1\right)_{X}\right)\label{eq:secimp}
\end{equation}
with maximal torus comprising a maximal torus in $E_{8}$ and so a
vector space isomorphism of Cartan subalgebras 
\[
\mathcal{\mathfrak{h}}_{su\left(5\right)}\oplus\mathcal{\mathfrak{h}}_{su\left(4\right)}\oplus\mathcal{\mathfrak{h}}_{u\left(1\right)}\overset{\cong}{\longrightarrow}\mathcal{\mathfrak{h}}_{e_{8}}
\]
and associated commutative diagram 
\[
\begin{array}{ccc}
SU\left(5\right)\times SU\left(4\right)\times U\left(1\right) & \overset{Ad_{SU\left(5\right)}\times Ad_{SU\left(4\right)}\times Ad_{U\left(1\right)}}{\longrightarrow} & GL\left(su\left(5\right)\times su\left(4\right)\times u\left(1\right)\right)\\
\downarrow &  & \downarrow\\
E_{8} & \overset{Ad_{E_{8}}}{\longrightarrow} & GL\left(e_{8}\right).
\end{array}
\]

Therefore the restriction of $E_{8}\overset{Ad_{E_{8}}}{\longrightarrow}Aut\left(e_{8}\right)$
to $\left(SU\left(5\right)\right)_{gauge}\times\left(SU\left(4\right)_{Higgs}\times U\left(1\right)_{X}\right)$
decomposes as in (88) of \cite{Blumenhagen} into 
\begin{equation}
\begin{array}{c}
\left(\mathbf{15},\mathbf{1}\right)_{0}\oplus\\
\left(\mathbf{1},\mathbf{1}\right)_{0}\oplus\left(\mathbf{1},\mathbf{10}\right)_{-4}\oplus\left(\mathbf{1},\mathbf{\overline{10}}\right)_{4}\oplus\left(\mathbf{1},\mathbf{24}\right)_{0}\\
\oplus\left(\mathbf{4},\mathbf{1}\right)_{5}\oplus\left(\mathbf{4},\mathbf{\overline{5}}\right)_{-3}\oplus\left(\mathbf{4},\mathbf{10}\right)_{1}\\
\oplus\left(\mathbf{\overline{4}},\mathbf{1}\right)_{-5}\oplus\left(\mathbf{\bar{4}},5\right)_{3}\oplus\left(\mathbf{4},\overline{\mathbf{10}}\right)_{-1}\\
\oplus\left(\mathbf{6},\mathbf{5}\right)_{-2}\oplus\left(\mathbf{6},\mathbf{\overline{5}}\right)_{2}
\end{array}\label{eq:Ad}
\end{equation}
Here $\mathbf{5}$ denotes the standard matrix representation 
\[
SU\left(5\right)\rightarrow GL\left(\mathbb{C}^{5}\right),
\]
$\mathbf{10}$ denotes the induced representation 
\[
SU\left(5\right)\rightarrow GL\left(\wedge^{2}\mathbb{C}^{5}\right)
\]
and the 'bar' indicates the conjugate representation obtained by composing
with the inverse map on $SU\left(5\right)$. Analogously for $SU\left(4\right)$.

\section{Semi-stable degeneration and $V_{3}/B_{2}$\label{sec:Semi-stable-degeneration-and}}

To link a crepant resolution $\tilde{W}_{4}/B_{3}$ to its Heterotic
dual $V_{3}/B_{2}$, recall that $B_{3}=B_{2}\times\mathbb{P}_{\left[u_{0},v_{0}\right]}$
so that one can realize $\tilde{W}_{4}/B_{3}$ as $W_{4,1}/B_{3,1}$
in a family of elliptic Calabi-Yau $4$-folds $W_{4,\delta}/B_{3,\delta}$
constructed over 
\[
B_{3,\delta}=B_{2}\times\left\{ a'b'=\delta\text{·}a''b''\right\} \subseteq B_{2}\times\mathbb{P}_{\left[a',a''\right]}\times\mathbb{P}_{\left[b',b''\right]}
\]
as follows. Using the affine parameters
\begin{equation}
\begin{array}{c}
\left[a',a''\right]=\left[a,1\right]\\
\left[b',b''\right]=\left[b,1\right]
\end{array}\label{eq:affcoords}
\end{equation}
 we define the affine family
\begin{equation}
B_{2}\times\left\{ ab=\delta:0\text{\ensuremath{\le}}\delta\text{\ensuremath{\le}}1\right\} \subseteq B_{2}\times\mathbb{C}^{2}\label{eq:deg1}
\end{equation}
and identify $\mathbb{P}_{\left[u_{0},v_{0}\right]}$ with the closure
$\left\{ a'b'=a''b''\right\} $ of $\left\{ ab=1\right\} $ by the
rule
\begin{equation}
\frac{u_{0}-v_{0}}{u_{0}+v_{0}}=a=b^{-1}.\label{deg2}
\end{equation}
That is
\begin{equation}
\begin{array}{c}
\left|\begin{array}{cc}
u_{0}-v_{0} & u_{0}+v_{0}\\
a' & a''
\end{array}\right|=0\\
\\
\left|\begin{array}{cc}
u_{0}-v_{0} & u_{0}+v_{0}\\
b'' & b'
\end{array}\right|=0
\end{array}\label{eq:flow}
\end{equation}
\begin{itemize}
\item In this way we identify
\[
\tilde{W}_{4}/B_{3}/B_{2}
\]
as an elliptic fibration over the closure $B_{2}\times\left\{ a'b'=a''b''\right\} $
of $B_{2}\times\left\{ ab=1\right\} $. Also as in \cite{Clemens-2}
as coefficients in the Tate form \eqref{eq:Tate} for $W_{4}/B_{3}$
we require that
\begin{equation}
a_{2},\,a_{3},\,a_{4},\,a_{5},\,z,\,t=\frac{y}{x}\in H^{0}\left(K_{B_{3}}^{-1}\right)^{\left[-1\right]},\label{eq:sections}
\end{equation}
(that is, anti-invariant under the involution $\beta_{3}$), in other
words they must be linear combinations of the forms listed in Table
2 in \cite{Clemens-2}. Furthermore also as described via Table 2
in \cite{Clemens-2} the four-dimensional family of forms spanned
by $a_{2},a_{3},a_{4},a_{5}\in H^{0}\left(K_{B_{3}}^{-1}\right)^{\left[-1\right]}$
should deform as
\begin{equation}
a_{j,\delta}=\delta\text{·}a_{j}+a_{j,0}\label{eq:contract}
\end{equation}
where the $a_{j,0}$ lie in the three-dimensional $\left(-i\right)$-eigenspace
of the $\mathbb{Z}_{4}$ \textbf{R}-symmetry $T_{u,v}$, again as
given in Table 2 in \cite{Clemens-2}. Finally
\[
z_{\delta}:=\delta z+z_{0}
\]
where $z_{0}:=q\text{·}\left(u_{0}^{2}-v_{0}^{2}\right)$ is skew-invariant
under $\beta_{3}$ so that $q\in H^{0}\left(K_{B_{2}}^{-1}\right)$
must be invariant under the involution
\[
B_{2}\overset{\beta_{2}}{\longrightarrow}B_{2}.
\]
As we again showed in \cite{Clemens-2}, 
\[
B_{2}=D_{2}
\]
is a particular degree-$2$ del Pezzo surface on which the involution
$\beta_{2}$ on $B_{2}$ acts with four fixpoints.
\end{itemize}
Thus $B_{2,\delta}=B_{2}$ for all $\delta$ and 
\[
\begin{array}{c}
B_{3,1}=B_{3}\\
B_{3,0}=\left(B_{2}\times\mathbb{P}_{\left[a',a''\right]}\right)\cup\left(B_{2}\times\mathbb{P}_{\left[b',b''\right]}\right)
\end{array}
\]
with
\[
\left(B_{2}\times\mathbb{P}_{\left[a',a''\right]}\right)\cap\left(B_{2}\times\mathbb{P}_{\left[b',b''\right]}\right)=\left(B_{2}\times\left\{ \left[0,1\right]\right\} \right)\cup\left(B_{2}\times\left\{ \left[0,1\right]\right\} \right)
\]
and applying the relations \eqref{eq:flow}. $a_{2},\,a_{3},\,a_{4},\,a_{5}$
are required to have no common zeros on $B_{3,\delta}$ for $\delta$
on a small complex disk around $\delta=0$. 

The Tate form \eqref{eq:Tate} and the action $\beta_{3}$ on $B_{3}$
as described in the tables in Section 4 of \cite{Clemens-2} then
determine equivariant involutions $\tilde{\beta}_{4}/\beta_{3}$ on
$W_{4}/B_{3}$. For generic choice of $z,$ $S_{\mathrm{GUT}}$ will
not contain any of the eight fixpoints of $\beta_{3}$ acting on $B_{3}$
so $\beta_{3}$ will act freely on $S_{\mathrm{GUT}}$ yielding a
smooth Enriques surface as quotient. The sections $a_{j,\delta},z_{\delta}\in H^{0}\left(K_{B_{3}}^{-1}\right)$
in the Tate form will be allowed to vary under the contraction \eqref{eq:contract}
and in particular the discriminant component $S_{\mathrm{GUT},\delta}\subseteq B_{3,\delta}$
varies as defined by 
\[
z_{\delta}:=\delta\text{·}z+\left(1-\delta\right)q\text{·}\left(u_{0}^{2}-v_{0}^{2}\right)=0
\]
with $\beta_{2}$-invariant $q\in H^{0}\left(K_{B_{2}}^{-1}\right)$. 

\subsection{Degeneration of a single $K3$-surface\label{subsec:Degeneration-of-a}}

Over each point $\left(a,b\right)\in\mathbb{C}^{2}$ associate the
Weierstrass form 
\begin{equation}
y^{2}=x^{3}+g{}_{2}\left(a,b\right)x+g{}_{3}\left(a,b\right)\label{04'}
\end{equation}
where $g_{2}$ is homogeneous of total degree $4$ and $g_{3}$ homogeneous
of total degree $6$. Restrict the Weierstrass form to the locus 
\begin{equation}
\Gamma_{\delta}:=\left\{ \left(a,b\right):a\text{·}b=\delta\right\} \label{05}
\end{equation}
so that for $\delta\neq0$, the discriminant has degree $24$. As
shown below compactifying at infinity yields a $K3$-surface elliptically
fibered over the closure of $\Gamma_{\delta}$ and setting $\delta=0$
yields the union of two $dP_{9}$'s meeting over $\left\{ a=b=0\right\} $.
For fixed value $b_{2}\in B_{2}$ the fiber $\pi_{B_{2}}^{-1}\left(b_{2}\right)$
has homogeneous linear coordinates $\left[u_{0},v_{0}\right]$ and
$\pi_{B_{3}}^{-1}\left(\pi_{B_{2}}^{-1}\left(b_{2}\right)\right)$
is a $K3$-surface in $W_{4}$. Since $g_{2}$ and $g_{3}$ in \eqref{eq:Wform}
must be homogeneous forms of degree 4 and $6$ respectively in $\left[u_{0},v_{0}\right]$,
dividing \eqref{eq:Wform} by $\left(u_{0}+v_{0}\right)^{6}$ gives
the affine equation 
\begin{equation}
y^{2}=x^{3}+g_{2}\left(a\right)x_{0}+g_{3}\left(a\right)\label{D}
\end{equation}
where 
\[
\begin{array}{c}
a:=\frac{u_{0}-v_{0}}{u_{0}+v_{0}}\\
x:=\frac{x_{nf}}{\left(u_{0}+v_{0}\right)^{2}}\\
y:=\frac{y_{nf}}{\left(u_{0}+v_{0}\right)^{3}}.
\end{array}
\]
On the $\mathbb{P}^{1}$-fiber over $b_{2}\in B_{2}$ we obtain the
Weierstrass form 
\begin{equation}
y^{2}=x^{3}+g_{2}\left(b_{2},a\right)x+g_{3}\left(b_{2},a\right)\label{ready}
\end{equation}
where $g_{2}\left(b_{2},a\right)$ is a rational function of $a$
with denominator of degree $4$ and $g_{3}\left(b_{2},a\right)$ can
be expressed as a rational function of $a$ with denominator of degree
$6$.

Thus one can write a global decomposition of \eqref{ready} with 
\[
\begin{array}{c}
g_{2}\left(b_{2},a\right)=\sum_{j=1}^{4}g_{2,j}^{-}\left(b_{2}\right)a^{-j}+g_{2,0}\left(b_{2}\right)+\sum_{j=1}^{4}g_{2,j}^{+}a^{j}\\
g_{3}\left(b_{2},a\right)=\sum_{j=1}^{6}g_{3,j}^{-}\left(b_{2}\right)a^{-j}+g_{3,0}\left(b_{2}\right)+\sum_{j=1}^{6}g_{3,j}^{+}a^{j}.
\end{array}
\]
Then, letting $b=a^{-1}$, we can equivalently write 
\[
\begin{array}{c}
g_{2}\left(b_{2};a,b\right)=\sum_{j=1}^{4}g_{2,j}^{-}\left(b_{2}\right)b^{j}+g_{2,0}\left(b_{2}\right)+\sum_{j=1}^{4}g_{2,j}^{+}a^{j}\\
g_{3}\left(b_{2};a,b\right)=\sum_{j=1}^{6}g_{3,j}^{-}\left(b_{2}\right)b^{j}+g_{3,0}\left(b_{2}\right)+\sum_{j=1}^{6}g_{3,j}^{+}a^{j}.
\end{array}
\]
where the $g_{2}$ and $g_{3}$ are the functions on the curve $\left\{ ab=1\right\} $
over the point $b_{2}$ in $B_{2}$.

Then for any $b_{2}\in B_{2}$ and any $\left(a,b\right)\in\mathbb{C}^{2}$
we can write the Weierstrass form 
\[
\begin{array}{c}
y^{2}=\frac{1}{2}x^{3}+\left(\frac{g_{2,0}\left(b_{2}\right)}{2}+\sum_{j=1}^{4}g_{2,j}^{+}a^{j}\right)x+\left(\frac{g_{3,0}\left(b_{2}\right)}{2}+\sum_{j=1}^{6}g_{3,j}^{+}b^{j}\right)\\
+\frac{1}{2}x^{3}+\left(\frac{g_{2,0}\left(b_{2}\right)}{2})\sum_{j=1}^{4}g_{2,j}^{-}\left(b_{2}\right)b^{j}\right)x+\left(\frac{g_{3,0}\left(b_{2}\right)}{2}+\sum_{j=1}^{6}g_{3,j}^{-}\left(b_{2}\right)b^{j}\right).
\end{array}
\]

\subsection{$dP_{9}$-bundles over $B_{2}$ giving a singular Calabi-Yau fourfold\label{subsec:-bundles-over-}}

As in Section \ref{subsec:Degeneration-of-a} we consider $\left(\mathbb{P}^{1}\right)^{2}$
with coordinates $\left(\left[a',a''\right],\left[b',b''\right]\right)$
as a fiber of the total space 
\[
\mathbb{P}\left(K_{B_{2}}^{-1}\oplus K_{B_{2}}^{-1}\right)\times_{B_{2}}\mathbb{P}\left(K_{B_{2}}^{-1}\oplus K_{B_{2}}^{-1}\right).
\]
The union of our two $dP_{9}$-bundles is then given as the subspace
of defined by \eqref{eq:goodDP} and the equation 
\[
a'b'=0.
\]
Returning to our fibration 
\[
\pi_{B_{2}}:B_{3}\rightarrow B_{2},
\]
for each fiber $\pi_{B_{3}}^{-1}\left(b_{3}\right)$ of $\tilde{W}_{4}/B_{3}$
we will associate $\mathit{two}$ copies of the Weierstrass equation,
namely the one distinguished by designating $\tilde{\zeta}\left(b_{3}\right)$
as the identity element of the group structure and the other distinguished
by designating $\tilde{\tau}\left(b_{3}\right)$ as the identity element
of the group structure. The two fibers $\pi_{B_{3}}^{-1}\left(b_{3}\right)$
and $\pi_{B_{3}}^{-1}\left(\beta\left(b_{3}\right)\right)$ are identified
under the isomorphism induced by the involution $\tilde{\beta}_{4}$
induced by the involution $\beta_{3}$ on $B_{3}$. However, as will
become clear in Section \ref{sec:Geometric-Model-Double-Cover} the
line bundle 
\begin{equation}
\mathcal{O}_{\pi_{B_{3}}^{-1}\left(b_{3}\right)}\left(\tilde{\zeta}\left(b_{3}\right)-\tilde{\tau}\left(b_{3}\right)\right)\label{eq:lb}
\end{equation}
is not trivial over any $b_{3}\in S_{\mathrm{GUT}}$. As we have seen
in \eqref{eq:Wnf}, the identification acts on sheaves $\mathcal{F}$
on the elliptic curve in Weierstrass form 
\begin{equation}
wy_{nf}^{2}=x_{nf}^{3}+g_{2}w^{2}x_{nf}+w^{3}\label{eq:Wform}
\end{equation}
where $g_{2}\left(b_{3}\right)=\left(Bz^{2}-\frac{A^{2}}{3}\right)$
and $g_{3}\left(b_{3}\right)=\left(Cz^{4}-\frac{AB}{3}z^{2}+\frac{2A^{3}}{27}\right)$.
The action on $\mathcal{F}$ is given on $\pi_{B_{3}}^{-1}\left(b_{3}\right)$
by 
\begin{equation}
\mathcal{F\mapsto F}\otimes\mathcal{O}_{\pi_{B_{3}}^{-1}\left(b_{3}\right)}\left(\tilde{\zeta}\left(b_{3}\right)-\tilde{\tau}\left(b_{3}\right)\right)\label{eq:paste}
\end{equation}
so that for each $b_{2}\in B_{2}$ the $K3$-surface $\left(\pi_{B_{2}}\circ\pi_{B_{3}}\right)^{-1}\left(b_{2}\right)$
can be obtained as the smoothing the union of two $dP_{9}$'s described
in Subsubsection \ref{subsec:Degeneration-of-a}. That is, we realize
$\tilde{W}_{4}/B_{2}$ as the smoothing of two $dP_{9}$-bundles 
\begin{equation}
dP_{a}\cup dP_{b}\label{eq:DP}
\end{equation}
over $B_{2}$.

Simultaneously, via 
\[
S_{\mathrm{GUT},\delta}:=\left\{ \delta z+\left(1-\delta\right)z_{0}=0\right\} \subseteq B_{2},
\]
we move $S_{\mathrm{GUT}}$ to the reducible quadric $z_{0}$ given
by 
\begin{equation}
z_{0}=\left(u_{0}^{2}-v_{0}^{2}\right)\cdot q\left(u_{1},v_{1},u_{2},v_{2}\right)\label{eq:ssdq}
\end{equation}
where $q\left(u_{1},v_{1},u_{2},v_{2}\right)$ is invariant under
the action of the involution $\beta_{2}$ on $B_{2}$. So at $\delta=0$
$S_{\mathrm{GUT}}$ splits into two `horizontal' components given
by $u_{0}-v_{0}=0$ and $u_{0}+v_{0}=0$ and a `vertical' component
given by $\pi_{B_{2}}^{-1}\left(\left\{ q\left(u_{1},v_{1},u_{2},v_{2}\right)=0\right\} \right)$.
Notice that, since $q$ is skew-symmetric, it vanishes on the four
fixpoints of the action of $\beta_{2}$ on $B_{2}$. Thus the intersection
\[
S_{\mathrm{GUT},\delta}\cap\pi_{B_{2}}^{-1}\left(fixpoint\,set\,of\,\beta_{2}\right)
\]
is independent of $\delta$. 

From the Weierstrass forms just above, we read off that $\tilde{W}_{4}$
is the smoothing over $ab=1$ of the union of the two $dP_{9}$-bundles
given over each point $b_{2}\in B_{2}$ by 
\begin{equation}
\frac{1}{2}y^{2}=\frac{1}{2}x^{3}+\left(\frac{g_{2,0}\left(b_{2}\right)}{2}+\sum_{j=1}^{8}g_{2,j}^{+}a^{j}\right)x+\left(\frac{g_{3,0}\left(b_{2}\right)}{2}+\sum_{j=1}^{12}g_{3,j}^{+}a^{j}\right)\label{p1}
\end{equation}
and 
\begin{equation}
\frac{1}{2}y^{2}=\frac{1}{2}x^{3}+\left(\frac{g_{2,0}\left(b_{2}\right)}{2}\sum_{j=1}^{8}g_{2,j}^{-}\left(b_{2}\right)b^{j}\right)x_{0}+\left(\frac{g_{3,0}\left(b_{2}\right)}{2}+\sum_{j=1}^{12}g_{3,j}^{-}\left(b_{2}\right)b^{j}\right)\label{p2}
\end{equation}
These two $dP_{9}$'s over $b_{2}$ contain the common fiber of $V_{3}/B_{2}$
whose Weierstass form is 
\[
wy^{2}=x^{3}+w^{2}g_{2,0}\left(b_{2}\right)x+w^{3}g_{3,0}\left(b_{2}\right).
\]
The spectral data on the Heterotic side, namely the two $E_{8}$-bundles
on the fiber of $V_{3}/B_{2}$ over $b_{2}$, are given via the Friedman-Morgan-Witten
classification \cite{Friedman:1997yq} by the two $dP_{9}$-bundles.
Namely in Section 4.2 of \cite{Friedman:1997yq} Friedman-Morgan-Witten
give a classifying space for imbeddings of an elliptic fiber $E_{b_{2}}$
of $V_{3}/B_{2}$ into a rational elliptic surface $dP_{9}\left(b_{2}\right)$,
each such corresponding canonically by a theorem of E. Looijenga \cite{Looijenga}
to an isomorphism class of flat $E_{8}^{\mathbb{C}}$-bundles $F$
over $E_{b_{2}}.$ Considered as fibrations over $B_{2}$, fibers
are so-called $dP_{9}$-surfaces. Setting $\left[s,t\right]=\left[a',a''\right]$,
respectively $\left[s,t\right]=\left[b',b''\right]$ for $s,t$ as
in \cite{Friedman:1997yq}, fibers are given uniquely in $\mathbb{P}_{1,1,2,3}^{3}$
by an equation 
\begin{equation}
y^{2}=4x^{3}-\left(g_{2}t^{4}-\beta_{1}st^{3}-\ldots-\beta_{4}s^{4}\right)x-\left(g_{3}t^{6}-\alpha_{2}s^{2}t^{4}-\ldots-\alpha_{6}s^{6}\right)\label{eq:LFMW}
\end{equation}
of weighted homogeneous degree $6$ in the variables $x,y,s,t$ with
respective weights $2,3,1,1$ over the weighted projective space with
coordinates 
\[
\left[\alpha_{2},\ldots,\alpha_{6},\beta_{1},\ldots,\beta_{4}\right]\in\mathbb{P}_{\left[2,3,4,5,6,1,2,3,4\right]}^{8}
\]
weighted as per their respective indices. 

So to realize the semi-stable degeneration geometrically, the crepant
resolution $\tilde{W}_{4}/B_{3}$ of $W_{4}/B_{3}$ is given over
the locus 
\[
\Gamma_{1}=\left\{ \left(a,b\right):ab=1\right\} \subseteq B_{3}\times\mathbb{C}^{2}.
\]
The deformation in this Section is given by restricting the Weierstrass
form to the locus 
\[
\Gamma_{\delta}=\left\{ \left(a,b\right):ab=\delta\right\} 
\]
as $\delta$ goes to $0$. Furthermore $S_{\mathrm{GUT}}$ deforms
with $\delta$ via the formula 
\begin{equation}
z_{\delta}=\delta\text{·}z+z_{0}.\label{eq:smooth}
\end{equation}
We thereby obtain the family$\left\{ \tilde{W}_{4,\delta}\right\} $
of fourfolds over the affine line $\mathbb{C}_{\delta}^{\ast}$ as
in Subsection \ref{subsec:Degeneration-of-a}.

Then the Heterotic model $\left(V_{3},F_{a},F_{b}\right)$ over $B_{2}$
canonically corresponds to a normal crossing Calabi-Yau $4$-fourfold
with two components $dP_{a}$ and $dP_{b}$ obtained by making the
construction described just above equivariantly over $B_{2}$.

Thus the family $W_{4,\delta}/B_{3,\delta}$ defined by the Tate form
on $B_{3,\delta}$ degenerates as $\delta$ approaches zero to a reducible
Calabi-Yau $4$-fold 
\[
W_{4,0}=dP_{a}\cup dP_{b}
\]
over 
\[
\left(B_{2}\times\mathbb{P}_{\left[a,1\right]}\right)\cup\left(B_{2}\times\mathbb{P}_{\left[b,1\right]}\right)
\]
with 
\[
V_{3}=dP_{a}\cap dP_{b}
\]
giving the Heterotic model $V_{3}/B_{2}$ defined by $a=b=0$. Finally
there are common distinguished sections on every $W_{4,\delta}/B_{3,\delta}$.
defined equivariantly by 
\[
\begin{array}{c}
\zeta_{\delta}=\left\{ w=x=0\right\} \\
\tau_{\delta}=\left\{ \left[w,x,y\right]=\left[w,wz_{\delta}^{2},wz_{\delta}^{3}\right]\right\} .
\end{array}.
\]

The Calabi-Yau threefold 
\[
V_{3}=dP_{a}\cap dP_{b}
\]
is elliptically fibered over $B_{2}\times\left\{ a=b=0\right\} $,
the general fiber of $dP_{a}/B_{2}$ has an $I_{5}$-fiber at $a=\infty$
and the general fiber of $dP_{b}/B_{2}$ has an $A_{4}$-fiber at
$b=\infty$.

\subsection{The involution \label{subsec:The-involution-}}

Then by construction the involution $\tilde{\beta}_{4}/\beta_{3}$
induces involutions $\tilde{\beta}_{4,\delta}/\beta_{3,\delta}$ on
$W_{4,\delta}/B_{3,\delta}$ for all $\delta\text{\ensuremath{\ge}}0$.
Furthermore the above asssumptions force 
\[
\begin{array}{ccc}
W_{4,0}/B_{3,0} & \overset{\tilde{\beta}_{4,0}}{\longrightarrow} & W_{4,0}/B_{3,0}\\
\downarrow &  & \downarrow\\
B_{2} & \overset{\beta_{2}}{\longrightarrow} & B_{2}
\end{array}
\]
to be given by the elliptic fibrations with involutions 
\begin{equation}
\begin{array}{c}
dP_{a}/\left(B_{2}\times\mathbb{P}_{\left[a,1\right]}\right)\overset{\tilde{\beta}_{a}}{\longrightarrow}dP_{9,a}/\left(B_{2}\times\mathbb{P}_{\left[a,1\right]}\right)\\
\left(\left[w,x,y\right],\left(b_{2},\left[a,1\right]\right)\right)\mapsto\left(\left[w,x,-y\right],\left(\beta\left(b_{2}\right),\left[-a,1\right]\right)\right)
\end{array}\label{eq:stacka}
\end{equation}
and 
\begin{equation}
\begin{array}{c}
dP_{9,b}/\left(B_{2}\times\mathbb{P}_{\left[b,1\right]}\right)\overset{\tilde{\beta}_{b}}{\longrightarrow}dP_{9,b}/\left(B_{2}\times\mathbb{P}_{\left[b,1\right]}\right)\\
\left(\left[w,x,y\right],\left(b_{2},\left[b,1\right]\right)\right)\mapsto\left(\left[w,x,-y\right],\left(\beta\left(b_{2}\right),\left[-b,1\right]\right)\right).
\end{array}\label{eq:stackb}
\end{equation}
The action of these involutions over the fixpoints of the action of
$\beta_{2}$ on $B_{2}$ is treated in detail in §2.4 of \cite{Clemens-1}. 

\section{\label{sec:NC}Passing from Heterotic theory to F-theory }

Essentially one passes from the Heterotic model to the $F$-theory
model by reading the Subsections of the previous Chapter in reverse
order and from bottom to top. We can paste one of these $dP_{9}$
along $E$ at $s=0$ and the other along $E$ at $s=0$ to obtain
a normal crossing elliptic $K3$-surface with unobstructed deformation
space thereby joining $dP_{a}^{\vee}$ to $dP_{b}^{\vee}$ to form
a normal crossing Calabi-Yau fourfold. Having that, a theorem of Kawamata-Namikawa
\cite{Kawamata} guarantees that the normal crossing elliptic Calabi-Yau
fourfold has an unobstructed deformation theory. $B_{3}=\mathbb{P}_{\left[u_{0},v_{0}\right]}\times B_{2}$
so that our only choices are the section of $K_{B_{2}}^{-1}$ in the
definition of $z_{0}$ and the smooth section $z$ of $K_{B_{3}}^{-1}$
in \eqref{eq:smooth}.

\subsection{Initial data\label{subsec:The-base}}

In Heterotic theory, we begin with the smooth degree-$2$ del Pezzo
del Pezzo $\left(B_{2},\beta_{2}\right)$ with involution $\beta_{2}$
with $4$ fixpoints constructed in \cite{Clemens-2}. Letting $B_{2}^{\vee}$
denote the quotient, we are given a smooth, elliptically fibered Calabi-Yau
threefold 
\[
V_{3}^{\vee}\rightarrow B_{2}^{\vee}
\]
with Weierstrass form 
\begin{equation}
y^{2}=x^{3}+g_{2}\left(b_{2}\right)x+g_{3}\left(b_{2}\right)\label{W}
\end{equation}
for each $b_{2}\in B_{2}^{\vee}$. This Weierstrass equation lives
at each point $b_{2}$ of a $\mathbb{P}_{\left[w,x,y\right]}$-bundle
over $B_{2}$ with 
\begin{align*}
x & \in H^{0}\left(K_{B_{2}}^{-2}\right)\\
y & \in H^{0}\left(K_{B_{2}}^{-3}\right)\\
g_{2} & \in H^{0}\left(K_{B_{2}}^{-4}\right)\\
g_{3} & \in H^{0}\left(K_{B_{2}}^{-6}\right)
\end{align*}
in order that $V_{3}^{\vee}$ have trivial canonical bundle. has fundamental
group $\mathbb{Z}_{2}$.

Furthermore $V_{3}^{\vee}$ supports two principal $E_{8}$-bundles
\begin{equation}
F_{a}^{\vee}\oplus F_{b}^{\vee}\label{eq:bundles}
\end{equation}
with Yang-Mills connections. Pulling back via
\[
\begin{array}{ccc}
V_{3}=B_{2}\times_{B_{2}^{\vee}}V_{3}^{\vee} & \rightarrow & B_{2}\\
\downarrow &  & \downarrow\\
V_{3}^{\vee} & \rightarrow & B_{2}^{\vee}
\end{array}
\]
we have two $\beta_{2}$-invariant $E_{8}$-bundles, $F_{a}$ and
$F_{b}$ each with a Yang-Mills connection with fiber over $b_{2}\in B_{2}$. 

The smooth, torus-fibered Calabi-Yau threefold $V_{3}^{\vee}/B_{2}^{\vee}$
with fundamental group $\mathbb{Z}_{2}$ is also endowed with two
disjoint sections 
\[
\zeta,\tau:B_{2}^{\vee}\rightarrow V_{3}^{\vee}.
\]
The pull-back of the two sections under 
\[
\pi_{B_{2}^{\vee}}:B_{2}\rightarrow B_{2}^{\vee}
\]
becomes the union of two disjoint sections 
\[
\zeta,\tau:B_{2}\rightarrow V_{3}:=B_{2}\times_{B_{2}^{\vee}}V_{3}^{\vee}.
\]

Also, following §2.4 of \cite{Clemens-1}, the fiber of the smooth
threefold $V_{3}^{\vee}/B_{2}^{\vee}$ over the four orbifold points
of $B_{2}^{\vee}$ must be acquired with multiplicity two, being doubly
covered by the elliptic fibers $E_{b_{2}}$ of $V_{3}/B_{2}$ over
the $\beta_{2}$-fixpoints $b_{2}\in B_{2}$. The covering must be
unbranched via translation by a distinguished half-period. Given the
prior conditions imposed on our model that half-period must be $\mathcal{O}_{E_{b_{2}}}\left(\zeta\left(b_{2}\right)-\tau\left(b_{2}\right)\right)$.
The role of the logarithmic transform over a neighborhood of each
$\beta_{2}$-fixpoint, especially how it induces the action 
\[
\left(b_{2},\left(x,y\right)\right)\mapsto\left(\beta_{2}\left(b_{2}\right),\left(x,y\right)\right)
\]
of $\tilde{\beta}_{3}/\beta_{2}$ on (Weierstrass) fibers of $V_{3}/B_{2}$,
as well as the transition from the $\beta_{2}$-pull-backs 
\begin{equation}
F_{a}\oplus F_{b}\label{eq:bundle2}
\end{equation}
is explained in \cite{Clemens-1}. 

The Yang-Mills connections on the bundles \eqref{eq:bundle2} restrict
to sums of eight flat line bundles on each fiber of $V_{3}/B_{2}$
and are completely determined by that family of restrictions.

\subsection{Building a normal-crossing $K3$ from an elliptic curve with two
flat $E_{8}$-bundles }

As in \eqref{eq:LFMW} the two flat $E_{8}$-bundles on the elliptic
fiber $E_{b_{2}}$ of $V_{3}/B_{2}$ over $b_{2}$ determine two $dP_{9}$-bundles
\[
dP_{a}\left(b_{2}\right)\cup dP_{b}\left(b_{2}\right),
\]
with $E_{b_{2}}$ given in each by 
\[
s=0.
\]

That is, over $B_{2}$ we have the union 
\[
dP_{a}\cup dP_{b}
\]
of two $dP_{9}$-bundles with
\[
dP_{a}\cap dP_{b}=V_{3}.
\]
Since the canonical bundle of this normal-crossing variety is trivial,
the theorem of Kawamata-Namikawa cited above establishes that its
deformation space is unobstructed. In fact we employ the very specific
smoothing determined by \eqref{eq:deg1} and \eqref{eq:contract}.

Finally the homogeneous coordinates $\left[u_{0},v_{0}\right]$ in
\cite{Clemens-2} and the curve \eqref{deg2}, the equations 
\[
\begin{array}{c}
\left|\begin{array}{cc}
u_{0}-v_{0} & u_{0}+v_{0}\\
a' & a''
\end{array}\right|=0\\
\\
\left|\begin{array}{cc}
u_{0}+v_{0} & u_{0}-v_{0}\\
b' & b''
\end{array}\right|=0
\end{array}
\]
allow us as in \eqref{eq:deg1} to form the family of hypersurfaces
\[
B_{3,\delta}\subseteq B_{2}\times\left(\mathbb{P}_{\left[a',a''\right]}\times\mathbb{P}_{\left[b',b''\right]}^{1}\right)\rightarrow B_{2}
\]
with fibers given by 
\begin{equation}
\left\{ a'\text{·}b'=\delta\text{·}a''\text{·}b''\right\} \label{eq:bend}
\end{equation}
with 
\[
K_{B_{3,\delta}}^{-1}=\left.K_{B_{2}}^{-1}\boxtimes\left(\mathcal{O}_{\mathbb{P}^{1}}\left(1\right)\boxtimes\mathcal{O}_{\mathbb{P}^{1}}\left(1\right)\right)\right|_{\left\{ a'\text{·}b'=\delta\text{·}a''\text{·}b''\right\} }=:N_{\delta}.
\]
The forms \eqref{eq:contract} then yield our $F$-theory model $W_{4}/B_{3}$
at $\delta=1$.

\subsection{The action of the involution on the Heterotic model}

In Section 4.2 of \cite{Friedman:1997yq} Friedman-Morgan-Witten gives
a classifying space for imbeddings of an elliptic curve $E$ in Weierstrass
form into $dP_{9}$'s, each corresponding canonically to an isomorphism
class of flat $E_{8}$-bundles 
\[
F\rightarrow E.
\]
The flat $E_{8}$-bundle $F$ is given as the sum of eight flat line
bundles, each given by the divisor of the form
\[
p-e
\]
where $e=\left(\tilde{\zeta}\right)\cap E$ is the identity element
of $E$ considered as an abelian group and $p$ is a geometrically
given point on the torus $E$. Again in \cite{Clemens-1} we follow
Friedman-Morgan-Witten and consider the family of $dP_{9}$-hypersurfaces
\begin{equation}
y^{2}=4x^{3}-\left(g_{2}t^{4}-\beta_{1}st^{3}-\ldots-\beta_{4}s^{4}\right)x-\left(g_{3}t^{6}-\alpha_{2}s^{2}t^{4}-\ldots-\alpha_{6}s^{6}\right)\label{44}
\end{equation}
in $\mathbb{P}_{1,1,2,3}^{3}$ where the $\alpha_{j}$ and $\beta_{j}$
are homogeneous forms of weight $j$ in a $\mathbb{P}_{1,2,2,3,3,4,4,5,6}^{8}$.
Fixing the values of $\alpha_{j}$ and $\beta_{j}$ we think of the
solution set of \eqref{44} as a rational hypersurface in $\mathbb{P}_{1,1,2,3}^{3}$
with distinguished pencil 
\begin{equation}
\gamma s+\delta t=0.\label{444}
\end{equation}
The given elliptic curve $E$ sits in each $dP_{9}$ in \eqref{44}
as the solution set to the equation 
\[
s=0.
\]

If we change the basepoint of the elliptic curve $E_{b_{3}}$ from
$\tilde{\zeta}\left(b_{3}\right)$ to $\tilde{\tau}\left(b_{3}\right)$,
then the trivial bundle goes to itself so that translation by $\left(\tilde{\tau}\left(b_{3}\right)-\tilde{\zeta}\left(b_{3}\right)\right)$
acts trivially on $\mathrm{Pic}^{0}\left(E_{b_{3}}\right)$ and therefore
acts trivially on flat $E_{8}$-bundles on $E_{b_{3}}$. However changing
the identity element on the torus $E$ from $\left(\tilde{\zeta}\right)\cap E$
to $\left(\tilde{\tau}\right)\cap E$ and then applying the Friedman-Morgan-Witten
dictionary gives a different sum of eight flat line bundles, namely
those obtained by divisors
\[
p-e+\left(\left(\tilde{\zeta}\right)\cap E\right)-\left(\left(\tilde{\tau}\right)\cap E\right)
\]
that is, all eight flat line bundles are tensored with the non-trivial
flat line bundle
\[
\mathcal{O}_{\tilde{W}_{4}}\left(\left(\tilde{\zeta}\right)-\left(\tilde{\tau}\right)\right).
\]
\begin{lem}
i) The induced action of the involution $\tilde{\beta}_{4,0}$ on
the union of the two $dP_{9}$-bundles at $\delta=z_{0}=0$ takes
each of the two to itself.

ii) The action of the involution $\tilde{\beta}_{3}$ on the intersection
$V_{3}$ of the two $dP_{9}$-bundles is given by the map 
\[
\left(b_{2};x,y\right)\mapsto\left(\beta_{2}\left(b_{2}\right);x,y\right)+\left(\tilde{\zeta}\left(\beta_{2}\left(b_{2}\right)\right)-\tilde{\tau}\left(\beta_{2}\left(b_{2}\right)\right)\right)
\]
where addition is with respect to the addition law on the elliptic
curve and the (identical) Weierstass forms 
\begin{equation}
\begin{array}{c}
y^{2}=x^{3}+g_{2}\left(b_{2}\right)x+g_{3}\left(b_{2}\right)\\
y^{2}=x^{3}+g_{2}\left(\beta\left(b_{2}\right)\right)x+g_{3}\left(\beta\left(b_{2}\right)\right).
\end{array}\label{hetcurve-1}
\end{equation}

iii) Over fixpoints $b_{2}$ of the action of $\beta_{2}$ on $B_{2}$,
$\tau\left(b_{2}\right)-\zeta\left(b_{2}\right)$ is a non-trivial
half-period on the fiber $E_{b_{2}}$ of $V_{3}/B_{2}$ but the action
of $\tilde{\beta}_{3}$ on the intersection cycle 
\[
4\text{·}\zeta\left(b_{2}\right)+5\text{·}\tau\left(b_{2}\right)
\]
of length nine determining the $E_{8}$-bundles on the fiber is given
by tensoring with the canonical section of the line bundle associated
to the divisor $\zeta\left(b_{2}\right)-\tau\left(b_{2}\right).$ 
\end{lem}

\begin{proof}
i) This assertion is immediate from the definition of $\tilde{\beta}_{4}$.

ii) First of all the involution is given by 
\begin{equation}
\begin{array}{c}
\left(x,y\right)\mapsto\left(x,-y\right)\\
a\mapsto-a
\end{array}\label{eq:signs-1}
\end{equation}
but the fiber of 
\[
V_{3}/B_{2}
\]
is given in (\ref{p1}) and (\ref{p2}) as 
\[
y^{2}=x^{3}+g_{2}^{0}x+g_{3}^{0}
\]
so that it does not reflect the sign change of $a$ and we showed
in \cite{Clemens-1} how the sign change of $y$ is absorbed in the
process of taking residue. Secondly the equations of the two $K3$-surfaces
fibers 
\[
\tilde{W}_{4}\times_{B_{2}}\left\{ b_{2}\right\} \cup\tilde{W}_{4}\times_{B_{2}}\left\{ \beta\left(b_{2}\right)\right\} 
\]
are invariant under $\tilde{\beta}_{4}$ and so are identical. Furthermore
by \eqref{eq:signs-1} the globally defined relative holomorphic $2$-form
on $\tilde{W}_{4}/B_{2}$ is invariant under the action of $\tilde{\beta}_{4}$.
Thus, for any point $b_{3}\in B_{3}$ lying over $b_{2}\in B_{2}$,
the relative one-form on $\tilde{W}_{4}/B_{2}$ obtained from the
contraction of the relative holomorphic $2$-form on $\tilde{W}_{4}/B_{2}$
against a section of the relative tangent bundle of $B_{3}/B_{2}$
is invariant under the action of $\tilde{\beta}_{4}$. So the induced
relative one-form on $V_{3}/B_{2}$ is invariant under the action
of $\tilde{\beta}_{3}$. So as we showed in \cite{Clemens-1} $\tilde{\beta}_{3}$
must act as a translation of distinguished basepoint with respect
to the common Weierstrass form \eqref{hetcurve-1}. Since 
\[
\begin{array}{c}
\tilde{\zeta}\left(b_{2}\right)\mapsto\tilde{\tau}\left(b_{2}\right)\\
\tilde{\zeta}\left(\beta\left(b_{2}\right)\right)\mapsto\tilde{\tau}\left(\beta\left(b_{2}\right)\right)
\end{array}
\]
respectively, ii) is proved.

iii) Choosing $a_{5}\text{\ensuremath{\neq}}0$ at the fixpoints $b_{3}$
of the action of $\tilde{\beta}$ on $B_{3}$ separates $\zeta\left(b_{3}\right)$
from $\tau\left(b_{3}\right)$. The assertion is forced by the fact
that, in the canonical identification of an elliptic curve $E$ with
$\mathrm{Pic}^{0}\left(E\right)$, the trivial bundle is associated
to $\zeta\left(b_{3}\right)$ for the Weierstrass form based at $\zeta\left(b_{3}\right)$
and is associated to $\tau\left(b_{2}\right)$ for the Weierstrass
form based at $\tau\left(b_{3}\right)$. 
\end{proof}
Said otherwise the involution on $dP_{9}$-structures \eqref{44}
induced by 
\begin{equation}
\begin{array}{c}
y\mapsto-y\\
\left(s,t\right)\mapsto\left(-s,t\right)
\end{array}\label{eq:fac1}
\end{equation}
induces the involution 
\[
\begin{array}{c}
\mathrm{Pic}^{0}\left(E\right)\rightarrow\mathrm{Pic}^{0}\left(E\right)\\
L\mapsto L^{-1}
\end{array}
\]
so that there is a unique translation 
\begin{equation}
\begin{array}{c}
E\rightarrow E\\
x\mapsto x+\left(e'-e\right)
\end{array}\label{eq:fac2}
\end{equation}
such that the composition takes the trivial bundle to itself under
the change of basepoint of the torus $E$ from $e$ to $e'$. Using
the decomposition 
\[
\mathrm{Pic}^{0}\left(E\right)^{\oplus8}
\]
of the moduli space of semi-stable $E_{8}$-bundles on $E$, the functorial
diagram 
\begin{equation}
\begin{array}{ccc}
F_{a} & \overset{\tilde{\beta}_{3}^{\ast}}{\longrightarrow} & F_{a}\\
\downarrow &  & \downarrow\\
\mathrm{Pic}^{0}\left(V_{3}/B_{2}\right)^{\oplus8} & \overset{\left(\otimes\mathcal{O}_{\tilde{W}_{0}}\left(\tilde{\zeta}-\tilde{\tau}\right)\right)\circ\left(\right)^{-1}}{\longrightarrow} & \mathrm{Pic}^{0}\left(V_{3}/B_{2}\right)^{\oplus8}
\end{array}\label{eq:goodDP}
\end{equation}
is such that the bottom horizonal map is the identity map. This is
what allows us to induce a semi-stable $E_{8}$-bundle on the quotient
\[
B_{2}^{\vee}=\frac{B_{2}}{\left\{ \beta_{2}\right\} }.
\]
See \cite{Clemens-1} for a more detailed analysis. 
\begin{rem}
\label{rem:If-we-change}The classifying space for flat $E_{8}$-structures
on the elliptic curve $E$ is acted on by the semi-direct product
\[
\mathrm{Pic}^{0}\left(E\right)\rtimes\mathbb{Z}_{2}
\]
where $\mathbb{Z}_{2}$ is generated by 
\[
\left(\left[s,t\right],\left(x,y\right)\right)\mapsto\left(\left[-s,t\right],\left(x,-y\right)\right).
\]
\end{rem}

\subsection{$SU\left(5\right)_{gauge}$-roots and the semi-stable limit\label{subsec:The-chiral-spectrum}}

The action of $\tilde{\beta}_{3}$ on 
\[
V_{3}/B_{2}=dP_{a}\cap dP_{b}
\]
is explained in \cite{Clemens-1} through the lens of its local action
over a neighborhood of $\beta_{2}$-fixpoints, as is the compatibility
of this action with the action induced by $\tilde{\beta}_{4}$ on
the semi-stable degeneration 
\[
\tilde{W}_{4}/B_{2}\Rightarrow dP_{a}\cup dP_{b}.
\]
One checks directly that this action is compatible with the semi-stable
degeneration of $S_{\mathrm{GUT}}$ to 
\[
\left\{ z_{0}=q\text{·}\left(u_{0}^{2}-v_{0}^{2}\right)=0\right\} .
\]
$z_{0}$ deposits an $I_{5}$-fiber at the point $\left\{ a''=0\right\} $
on each fiber of $dP_{a}/B_{2}$ and at the point $\left\{ b''=0\right\} $
on each fiber of $dP_{b}/B_{2}$.By construction, $\tilde{\zeta}$
and $\tilde{\tau}$ can only meet over $\left\{ z_{\delta}=0\right\} $.
If 
\[
b_{2}\in\left\{ u_{1}^{2}-v_{1}^{2}=u_{2}^{2}-v_{2}^{2}=0\right\} \subseteq B_{2}
\]
is any of the $4$ fixpoints of the involution $\beta_{2}$, then
when $\delta=0$ the fixpoints of the action of $\beta_{3}$ on the
fiber $\mathbb{P}_{\left[u_{0},v_{0}\right]}\times\left\{ b_{2}\right\} $
occur at $a=0,\infty$ since $\beta$ takes $a$ to $-a$. We choose
the form $q\in H^{0}\left(K_{B_{2}}^{-1}\right)$ in the definition
of $z_{0}$ so as not to vanish on any of the fixpoints of the action
of $\beta_{2}$ on $B_{2}$. 

\section{Geometric Model-Double Cover Form\label{sec:Geometric-Model-Double-Cover}}

\subsection{Putting the sections $\zeta$ and $\tau$ of $W_{4}/B_{3}$ on equal
footing }

As mentioned earlier, the $\mathbb{Z}_{2}$-action will ultimately
incorporate a translation in the fiber direction that will interchange
the two distinguished sections $\zeta$ and $\tau$, and so will interchange
the two distinguished divisors, $\left(\zeta\right)$ and $\left(\tau\right)$,
to which they correspond. In particular, the Weierstrass forms for
the elliptic fibration determined by either of the two sections must
be 'on equal footing with the other one' throughout. So we will have
to begin from another birational model for $W_{4}/B_{3}$ that achieves
the desired `equal footing.' On each fiber of $W_{4}/B_{3}$, the
intersection with the sections $\left(\zeta\right)+\left(\tau\right)$
distinguishes a $g_{2}^{1}$ (that is, a linear series of projective
dimension one and degree two). Taken together the $g_{2}^{1}$ 's
yield a $\mathbb{P}^{1}$-bundle 
\[
Q:=\mathbb{P}\left(\left(\pi_{B_{3}}\right)_{\ast}\mathcal{O}_{W_{4}}\left(\left(\zeta\right)+\left(\tau\right)\right)\right)
\]
over $B_{3}$.

Since we are eventually going to change the elliptic group structure
on the torus fibers of $W_{4}/B_{3}$ from the one given by $\left(\zeta\right)$
that does not pass through the set of singular points of $W_{4}$
to the one given by $\left(\tau\right)$ that actually contains the
set of singular points of $W_{4}$, we first want to define a crepant
partial resolution of our $F$-theory model $W_{4}/B_{3}$ that becomes
a branched double cover of $Q/B_{3}$ and has the property that the
proper transform of $\left(\tau\right)$ misses the singular set of
the resulting fourfold entirely. We achieve this by projecting each
torus fiber from the third point $\upsilon\left(b_{3}\right)$ of
intersection of the line between $\zeta\left(b_{3}\right)$ and $\tau\left(b_{3}\right)$
with the fiber of $W_{4}/B_{3}$ over $b_{3}$.

\subsubsection{Line between the sections}

Over $b_{3}\in\left(B_{3}-S_{\mathrm{GUT}}\right)$, the line through
$\zeta\left(b_{3}\right)$ and $\tau\left(b_{3}\right)$ in $\left\{ b_{3}\right\} \times_{B_{3}}P$
is given by 
\[
x-z^{2}w=0.
\]
We have denoted the third point of intersection of this line with
$\left\{ b_{3}\right\} \times_{B_{3}}W_{4}$ as 
\[
\upsilon\left(b_{3}\right)=\left[w,z^{2}w,\left(-z-a_{420}\right)z^{2}w\right].
\]

Thus we can modify the defining equation 
\[
\left|\begin{array}{cc}
x^{3}+a_{4}zwx^{2}+a_{2}z^{3}w^{2}x+a_{0}z^{5}w^{3} & 1\\
wy\left(y-\left(a_{5}x+a_{3}z^{2}w\right)\right) & 1
\end{array}\right|=0
\]
of $W_{4}$ by projecting $\left\{ b_{3}\right\} \times_{B_{3}}W_{4}$
from the point $\upsilon\left(b_{3}\right)$. That is we write 
\begin{equation}
\begin{array}{c}
x-z^{2}w=X\\
y+\left(z+a_{420}\right)z^{2}w=Y
\end{array}\label{eq:relation}
\end{equation}
in order that the section $\left(\upsilon\right)$ correspond to the
point $\left[w,X,Y\right]=\left[1,0,0\right]$. The change of fiber
coordinates in $P/B_{3}$ is given by 
\begin{equation}
\left(\begin{array}{c}
w\\
X\\
Y
\end{array}\right)=\left(\begin{array}{ccc}
1 & 0 & 0\\
-z^{2} & 1 & 0\\
\left(z+a_{420}\right)z^{2} & 0 & 1
\end{array}\right)\left(\begin{array}{c}
w\\
x\\
y
\end{array}\right).\label{eq:matrixrel}
\end{equation}
(Notice that this is the identity matrix over a first-order neighborhood
of $S_{\mathrm{GUT}}$.) Recalling again that $-a_{53}=a_{420}$,
we obtain 
\begin{equation}
\left|\begin{array}{cc}
\left(X^{3}+3z^{2}wX^{2}+3z^{4}w^{2}X\right)+a_{4}zw\left(X^{2}+2z^{2}wX\right)+a_{2}z^{3}w^{2}X & 1\\
wY\left(Y-z^{3}w-a_{5}X\right)-\left(z+a_{420}\right)z^{2}w^{2}\left(Y-a_{5}X\right) & 1
\end{array}\right|=0\label{eq:W4}
\end{equation}
as the equation for $W_{4}$ in the projective coordinates $\left[w,X,Y\right]$.

\subsection{Fundamental projection and modified Weierstrass form\label{sec:Fundamental-projection-and}}

By considering \eqref{eq:W4} as a quadratic equation in $w$, the
Tate form then yields a birationally equivalent double-cover 
\[
\overline{W}_{4}/Q
\]
where the branched double cover $\overline{W}_{4}$ of $Q$ is given
by the equation 
\begin{equation}
\begin{array}{c}
w^{2}=\left(3z^{2}X^{2}+a_{4}zX^{2}-Y^{2}+a_{5}XY\right)^{2}\\
-4\left(\left(3z^{4}+\left(2a_{4}+a_{2}-a_{5}\right)z^{3}-a_{5}a_{420}z^{2}\right)X^{4}+\left(2z^{3}+a_{420}z^{2}\right)X^{3}Y\right)
\end{array}\label{eq:Delta0}
\end{equation}
where the fiber coordinates of $Q/B_{3}$ are given by 
\[
\begin{array}{c}
X=x-z^{2}w\\
Y=y+\left(z+a_{420}\right)z^{2}.
\end{array}
\]

\eqref{eq:Delta0} also allows us to see that $W_{4}$ is birationally
a double cover 
\[
\pi_{Q}:\overline{W}_{4}\rightarrow Q
\]
with equation 
\begin{equation}
\begin{array}{c}
w^{2}=\left(3z^{2}X^{2}+a_{4}zX^{2}-Y^{2}+a_{5}XY\right)^{2}\\
-4\left(\left(3z^{4}+\left(2a_{4}+a_{2}-a_{5}\right)z^{3}-a_{5}a_{420}z^{2}\right)X^{4}+\left(2z^{3}+a_{420}z^{2}\right)X^{3}Y\right)
\end{array}\label{eq:neweq}
\end{equation}
where 
\[
\overline{W}_{4}\subseteq\bar{P}:=\mathbb{P}\left(\mathcal{O}_{B_{3}}\left(6N\right)\oplus\mathcal{O}_{B_{3}}\left(2N\right)\oplus\mathcal{O}_{B_{3}}\left(3N\right)\right)=\mathbb{P}_{\left[w,X,Y\right]}
\]
and 
\[
Q=\mathbb{P}_{\left[X,Y\right]}:=\mathbb{P}\left(\mathcal{O}_{B_{3}}\left(2N\right)\oplus\mathcal{O}_{B_{3}}\left(3N\right)\right)
\]
with branch locus $\Delta$ given by the zeros of the homogeneous
polynomial \eqref{eq:Delta0}. Notice that \eqref{eq:neweq} is invariant
under the transformations 
\[
\left[\left(a_{0},a_{2},a_{3},a_{4},a_{5}\right),z,X,Y,w\right]\mapsto\left[\left(-a_{0},-a_{2},-a_{3},-a_{4},-a_{5}\right),-z,\,X,-Y,\,\text{\textpm}w\right]
\]
so both are possible but only one will leave the holomorphic four-form
on $\overline{W}_{4}$ invariant, namely the one that is compatible
with the action of $\tilde{\beta}_{4}/\beta_{3}$ that transforms
the relative one-form by 
\[
\frac{dx}{y}\mapsto-\frac{dx}{y}.
\]
To see which one, we divide \eqref{eq:neweq} by $X^{4}$ and define
\[
\vartheta_{0}:=\frac{Y}{X}.
\]
We have the affine equation 
\begin{equation}
\begin{array}{c}
w_{0}^{2}=\left(3z^{2}+a_{4}z-\vartheta_{0}^{2}+a_{5}\vartheta_{0}\right)^{2}\\
-4\left(\left(3z^{4}+\left(2a_{4}+a_{2}-a_{5}\right)z^{3}-a_{5}a_{420}z^{2}\right)+\left(2z^{3}+a_{420}z^{2}\right)\vartheta_{0}\right)
\end{array}\label{eq:newer}
\end{equation}
Therefore, since $\vartheta_{0}\mapsto-\vartheta_{0}$, $\tilde{\beta}_{4}/\beta_{3}$
is only compatible with the transformation 
\[
\frac{dw_{0}}{\vartheta_{0}}\mapsto-\frac{dw_{0}}{\vartheta_{0}}
\]
on the relative one-form so we conclude that $w_{0}$ must be invariant.
Furthermore, since the sections $\zeta$ and $\tau$ are now given
by $\left\{ X=0\right\} $, so, referring to \eqref{eq:neweq}, their
equation becomes 
\[
0=w_{0}^{2}-Y^{4}=\left(w_{0}+Y^{2}\right)\left(w_{0}-Y^{2}\right)
\]
and so each of the two sections must be taken to itself under $\tilde{\beta}_{4}/\beta_{3}$.

The canonical bundle of $Q$ is 
\[
-2\left(X_{0}\right)-2N
\]
where $\left(X_{0}\right)$ denotes the divisor $\left\{ X_{0}=0\right\} $
whereas the branch locus $\Delta$ has divisor class 
\[
4\left(X_{0}\right)+4N.
\]
We therefore replace $W_{4}$ with the birationally equivalent double
cover 
\[
\pi_{Q}:\overline{W}_{4}\rightarrow Q
\]
branched over $\Delta$ and $\overline{W}_{4}$ is Calabi-Yau. Furthermore
\[
\overline{W}_{4}\cap\left\{ X=0\right\} =\left(\bar{\zeta}\right)\cup\left(\bar{\tau}\right)
\]
is the union of the proper transforms of the two original sections
of $W_{4}/B_{3}$.

\subsection{The branch locus}

The branch locus $\Delta$ is defined by the equation

\begin{equation}
\begin{array}{c}
0=\left(\left(3z^{2}+a_{4}z\right)-\vartheta_{0}\left(\vartheta_{0}-a_{5}\right)\right)^{2}\\
-4\left(\left(3z^{4}+\left(2a_{4}+a_{2}-a_{5}\right)z^{3}-a_{5}a_{420}z^{2}\right)+\left(2z^{3}+a_{420}z^{2}\right)\vartheta_{0}\right)
\end{array}\label{eq:affdelta}
\end{equation}
on the space 
\[
Q-\left\{ X=0\right\} =\left|\mathcal{O}_{B_{3}}\left(N\right)\right|,
\]
the total space of the line bundle $\mathcal{O}_{B_{3}}\left(N\right)$.
ReThen the equation for 
\[
\left(\overline{W}_{4}-\left\{ X=0\right\} \right)\subseteq\left(\bar{P}-\left\{ X=0\right\} \right)
\]
can be rewritten as 
\begin{equation}
\begin{array}{c}
w_{0}^{2}=\left(\left(3z^{2}+a_{4}z\right)-\vartheta_{0}\left(\vartheta_{0}-a_{5}\right)\right)^{2}-4a_{420}z^{2}\left(\vartheta_{0}-a_{5}\right)\\
-4z^{3}\left(3z+2a_{4}+a_{2}-\left(\vartheta_{0}+\left(\vartheta_{0}-a_{5}\right)\right)\right).
\end{array}\label{eq:waffeq}
\end{equation}
Again rescaling the $a_{j}$ and appealing to Bertini's theorem, we
will have for general choices of the $a_{j}\in\mathcal{A}$ that singularities
of $\Delta$ are supported on the locus 
\[
\left\{ \vartheta_{0}\left(\vartheta_{0}-a_{5}\right)=z=0\right\} .
\]
Thus $S_{\mathrm{GUT}}\times_{B_{3}}\overline{W}_{4}$ has two components,
each isomorphic to $S_{\mathrm{GUT}}\times_{B_{3}}Q$. They are given
by 
\[
w_{0}=\text{\textpm}\vartheta_{0}\left(\vartheta_{0}-a_{5}\right).
\]
One of these components, that we will denote as $D_{1}$, intersects
$\left(\zeta\right)$, while the other, that we will denote as $D_{4}$,
intersects $\left(\tau\right)$. We write $D_{1}=:\left\{ G_{1}=0\right\} $
and $D_{4}=:\left\{ G_{4}=0\right\} $ so that 
\begin{equation}
z=G_{1}G_{4}\label{eq:decomp-1}
\end{equation}
on $\overline{W}_{4}$.

Over $\left\{ z=0\right\} $ the two components coincide as the component
$\left\{ \vartheta_{0}=a_{5}\right\} $ of the branch locus of the
reducible double cover 
\begin{equation}
\left\{ w_{0}^{2}=\vartheta_{0}^{2}\left(\vartheta_{0}-a_{5}\right)^{2}\right\} \label{eq:BR}
\end{equation}
of 
\[
Q\times_{B_{3}}\left\{ z=0\right\} .
\]

\subsection{\label{subsec:The-standard--formulation}The standard $\mathbb{P}^{112}$-formulation}

To relate the presentation \eqref{eq:Delta0} to the more standard
$\mathbb{P}^{112}$-notation used for this type of model, let 
\[
\ddot{P}:=\mathbb{P}\left(\mathcal{O}_{B_{3}}\oplus\mathcal{O}_{B_{3}}\left(2\right)\oplus\mathcal{O}_{B_{3}}\left(4\right)\right)
\]
with coordinates $\left(w,x,y\right)$ and write 
\[
\begin{array}{c}
X\rightarrow w\\
Y\rightarrow x\\
w\rightarrow y+x^{2}
\end{array}
\]
so that 
\[
\begin{array}{c}
w^{2}=\left(3z^{2}X^{2}+a_{4}X^{2}z-Y\left(Y-a_{5}X\right)\right)^{2}\\
-4\left(\left(3z^{4}+\left(2a_{4}+a_{2}-a_{5}\right)z^{3}-a_{5}a_{420}z^{2}\right)X^{4}+\left(2z^{3}+a_{420}z^{2}\right)X^{3}Y\right)
\end{array}
\]
becomes 
\[
\begin{array}{c}
\left(y+x^{2}\right)^{2}=\left(w^{2}\left(3z^{2}+a_{4}z\right)-x\left(x-a_{5}w\right)\right)^{2}\\
-4\left(\left(3z^{4}+\left(2a_{4}+a_{2}-a_{5}\right)z^{3}-a_{5}a_{420}z^{2}\right)w^{4}+\left(2z^{3}+a_{420}z^{2}\right)xw^{3}\right)
\end{array}
\]
that can be rewritten as 
\begin{equation}
\begin{array}{c}
y^{2}+2x^{2}y=-2a_{5}x^{3}w\\
+\left(a_{5}^{2}-2a_{4}z-6z^{2}\right)x^{2}w^{2}\\
+\left(2a_{5}a_{4}z+\left(6a_{5}-4a_{420}\right)z^{2}-8z^{3}\right)xw^{3}\\
-\left(\left(-a_{4}^{2}-4a_{5}a_{420}\right)z^{2}+\left(2a_{4}+4a_{2}-4a_{5}\right)z^{3}+3z^{4}\right)w^{4}.
\end{array}\label{eq:pee112}
\end{equation}
We include Appendix \ref{subsec:-App.1} by Sakura Schäfer-Nameki
containing a brief overview of elliptically-fibered Calabi-Yau manifolds
symmetric with respect to two sections.

\section{Desingularization of $\overline{W}_{4}$\label{sec:Desingularization-of}}

\subsection{Localizing at the singularities of $\overline{W}{}_{4}$ \label{subsec:Localizing-at-the}}

We first rearrange terms of \eqref{eq:affdelta} in increasing order
of total degree in the variables $\left(z,\vartheta_{0}\right)$ so
that the branch locus $\Delta$ is given by 
\begin{equation}
\begin{array}{c}
\left(a_{5}\vartheta_{0}+a_{4}z\right)^{2}+4a_{5}a_{420}z^{2}\\
-2\left(a_{5}\vartheta_{0}^{3}+a_{4}\vartheta_{0}^{2}z\right)+\left(6a_{5}-4a_{420}\right)z^{2}\vartheta_{0}+\left(-2a_{4}-4a_{2}+4a_{5}\right)z^{3}\\
+\vartheta_{0}^{4}-6\vartheta_{0}^{2}z^{2}-8z^{3}\vartheta_{0}-3z^{4}.
\end{array}\label{eq:loc1}
\end{equation}
$\Delta$ will have ordinary nodal singularities along $\left\{ z=\vartheta=0\right\} $
except where the quadratic normal cone 
\[
\left(a_{5}\vartheta_{0}+a_{4}z\right)^{2}+4a_{5}a_{420}z^{2}=0
\]
is not of maximal rank, namely where 
\[
a_{5}a_{420}=0.
\]
If $a_{5}a_{420}=0$, the reduced quadratic cone is given by 
\[
\left[\vartheta_{0},z\right]=\left[a_{4},-a_{5}\right].
\]
on which the cubic cone evaluates as 
\[
a_{5}^{2}\left(a_{420}a_{4}+\left(2a_{4}+a_{2}+a_{5}\right)a_{5}\right)=0.
\]
So the quadratic and cubic cone both vanish identically over $\left\{ a_{5}=0\right\} $
but both only vanish over $\left\{ a_{420}=a_{53}=0\right\} $ when
additionally 
\begin{equation}
a_{5}^{3}\left(a_{4}-a_{0}+a_{5}\right)=0,\label{eq:bingo}
\end{equation}
a locus that has only finite intersection with $S_{\mathrm{GUT}}$
.

We next rearrange terms of \eqref{eq:affdelta} in increasing order
of total degree in the variables $\left(z,a_{5}-\vartheta_{0}\right)$.
The branch locus $\Delta$ is then given by 
\begin{equation}
-\begin{array}{c}
0=\left(a_{5}\left(a_{5}-\vartheta_{0}\right)+a_{4}z\right)^{2}\\
2a_{5}\left(a_{5}-\vartheta_{0}\right)^{3}-2a_{4}z\left(a_{5}-\vartheta_{0}\right)^{2}+\left(4a_{420}+6a_{5}\right)\left(a_{5}-\vartheta_{0}\right)z^{2}-\left(2a_{4}-4a_{2}-4a_{5}\right)z^{3}\\
+\left(a_{5}-\vartheta_{0}\right)^{4}-6\left(a_{5}-\vartheta_{0}\right)^{2}z^{2}+8z^{3}\left(a_{5}-\vartheta_{0}\right)-3z^{4}.
\end{array}\label{eq:loc2}
\end{equation}
The quadratic normal cone along $\left\{ z=a_{5}-\vartheta_{0}=0\right\} $
is of rank one except where $a_{5}=a_{4}=0$. It is spanned by the
vector $\left(a_{5}-\vartheta_{0},z\right)=\left(a_{4},-a_{5}\right)$.
Evaluating the cubic normal cone along this vector and recalling that
$a_{420}=-a_{53}$ gives 
\[
\left(-a_{4}a_{3}+a_{5}\left(a_{0}+a_{3}\right)\right)a_{5}^{2}=0.
\]
Thus a simple modification will resolve the singularity of $\overline{W}_{4}$
over $\left\{ a_{5}-\vartheta_{0}=z=0\right\} $ unless $a_{5}=0$
or 
\begin{equation}
\left|\begin{array}{cc}
a_{4} & -a_{5}\\
a_{0}+a_{3} & a_{3}
\end{array}\right|=0.\label{eq:bingo2}
\end{equation}
\medskip{}

\subsection{Resolving the singularities of $\overline{W}_{4}$\label{sec:Resolving-the-singularities}}

\subsubsection{First modification of $\Delta$}

We first make the modification of $\Delta$ and $\overline{W}_{4}$
at $\left\{ \vartheta=z=0\right\} $. Namely, inside 
\[
Q^{\left(1\right)}:=\left\{ \left|\begin{array}{cc}
\vartheta_{0} & z\\
\vartheta_{1} & Z_{14}
\end{array}\right|\right\} \subseteq Q\times\mathbb{P}_{\left[\vartheta_{1},Z_{14}\right]}
\]
we write 
\[
Z_{0}:=\frac{\vartheta_{0}}{\vartheta_{1}}=\frac{z}{Z_{14}}.
\]
The proper transform $\Delta^{\left(1\right)}$ of $\Delta$ in $Q^{\left(1\right)}$
is given by the equation 
\begin{equation}
\begin{array}{c}
0=\left(a_{5}\vartheta_{1}+a_{4}Z_{14}\right)^{2}+4a_{5}a_{420}Z_{14}^{2}\\
+\left(-2a_{5}\vartheta_{1}^{3}-2a_{4}\vartheta_{1}^{2}Z_{14}+\left(6a_{5}-4a_{420}\right)\vartheta_{1}Z_{14}^{2}-\left(2a_{4}+4a_{2}-4a_{5}\right)Z_{14}^{3}\right)Z_{0}\\
+\left(\vartheta_{1}^{4}-6\vartheta_{1}^{2}Z_{14}^{2}-8Z_{14}^{3}\vartheta_{1}-3Z_{14}^{4}\right)Z_{0}^{2}.
\end{array}\label{eq:DSp}
\end{equation}

Over $z=Z_{14}Z_{0}=0$ we have two possible singular loci in $\Delta^{\left(1\right)}$.
These loci are given respectively by 
\[
\left\{ Z_{14}=\left(\vartheta_{1}Z_{0}-a_{5}\right)^{2}=0\right\} 
\]
and by singular points of 
\[
\left\{ Z_{0}=\left(a_{5}\vartheta_{1}+a_{4}Z_{14}\right)^{2}+4a_{5}a_{420}Z_{14}^{2}=0\right\} 
\]
at which $Z_{0}=0$ and $Z_{14}\text{\ensuremath{\neq}}0$. In this
last case we write 
\[
Z_{0}=\left(a_{5}\frac{\vartheta_{1}}{Z_{14}}+a_{4}\right)^{2}+4a_{5}a_{420}=0
\]
to conclude that singular points can only occur where the discriminant
\[
\left\{ 4a_{5}a_{420}=0\right\} 
\]
of the quadratic equation in $\frac{\vartheta_{1}}{Z_{14}}$ is singular,
that is, where $Z_{0}=a_{5}=a_{420}=0$. Substituting in \eqref{eq:DSp}
we conclude that, in addition, $a_{4}=0.$ Since $\frac{\partial}{\partial Z_{0}}$
applied to \eqref{eq:DSp} and evaluated at such singular points would
also have to vanish, we would also have $a_{2}=0$ and so $a_{0}=0$
contradicting the assumption of a generic allowable selection of the
$a_{j}$ and $z$ in the linear system $\left|\mathcal{A}\right|$.
\begin{lem}
\label{lem:Singularities-of-the}Singularities of the branch locus
$\Delta^{\left(1\right)}$ lie on the locus 
\[
\left\{ Z_{14}=\left(\vartheta_{1}Z_{0}-a_{5}\right)^{2}=0\right\} .
\]
\end{lem}

The assumption that the curve $\left\{ a_{5}=z=0\right\} \subseteq B_{3}$
is smooth can be weakened to allow nodal singularities. Since $a_{4}=0$
there, the potential nodal singularities in $\Delta{}^{\left(1\right)}$
thereby created are in fact already resolved by the fact that the
entire fiber of $Z_{0}$ over any point where $a_{5}=a_{4}=0$ already
lies in $\Delta{}^{\left(1\right)}$. In fact, in §6.3 of \cite{Clemens-2}
we impose a reducibility assumption on $\left\{ a_{5}=z=0\right\} \subseteq B_{3}$
. 

The proper transform $\Delta^{\left(1\right)}$ of $\Delta$ in $Q^{\left(1\right)}$
is the total transform minus twice the exceptional divisor, so that
the proper transform is again twice the anticanonical divisor of $Q^{\left(1\right)}$.
So the branched double cover $W_{4}^{\left(1\right)}$ given by the
equation 
\begin{equation}
\begin{array}{c}
w_{1}^{2}=\left(a_{5}\vartheta_{1}+a_{4}Z_{14}\right)^{2}+4a_{5}a_{420}Z_{14}^{2}\\
+\left(-2a_{5}\vartheta_{1}^{3}-2a_{4}\vartheta_{1}^{2}Z_{14}+\left(6a_{5}+4a_{420}\right)\vartheta_{1}Z_{14}^{2}-\left(2a_{4}+4a_{2}-4a_{5}\right)Z_{14}^{3}\right)Z_{0}\\
+\left(\vartheta_{1}^{4}-6\vartheta_{1}^{2}Z_{14}^{2}-8Z_{14}^{3}\vartheta_{1}-3Z_{14}^{4}\right)Z_{0}^{2}.
\end{array}\label{eq:d}
\end{equation}
is again Calabi-Yau. We retain the notation $D_{1}$ and $D_{4}$
for the proper transforms in $W_{4}^{\left(1\right)}$ of those respective
divisors in $\overline{W}_{4}$ , namely the two components of $W_{4}^{\left(1\right)}\times_{Q^{\left(1\right)}}\left\{ Z_{14}=0\right\} $.
We denote the irreducible branched double cover $W_{4}^{\left(1\right)}\times_{\Delta^{\left(1\right)}}\left\{ Z_{0}=0\right\} $
as $D_{0}$. Over $\left\{ z=a_{5}a_{420}=0\right\} $, $D_{0}$ splits
into two components, $\left\{ w_{1}=\text{\textpm}\left(a_{5}\vartheta_{1}+a_{4}Z_{14}\right)\right\} $,
and each component projects isomorphically to $\left\{ Z_{0}=0\right\} \subseteq Q^{\left(1\right)}.$

We note at this point that the section $\left\{ \vartheta_{0}=0\right\} $
is the projection of the section $\left(\upsilon\right)$ into $Q$.
It is this section that will be invariant under the final $\mathbb{Z}_{2}$-action
whereas the sections $\left(\zeta\right)$ and $\left(\tau\right)$
will be interchanged. Therefore the image of $\left(\upsilon\right)$
will be the distingushed section of the Calabi-Yau quotient $W_{4}^{\vee}/B_{3}^{\vee}$.
The proper transform of $\left(\upsilon\right)$ in $W_{4}^{\left(1\right)}/B_{3}$
passes though $D_{0}$ over $S\mathrm{_{GUT}}$, thereby justifying
the notation of $D_{0}$ as the 'inherited' component of the quotient
$W_{4}^{\vee}/B_{3}^{\vee}$.

\subsubsection{Second modification of $\Delta$ \label{subsec:First-modification-of}}

Recalling that $z=Z_{0}Z_{14}$ we must next attend to the singularities
of $\Delta^{\left(1\right)}$ lying in $\left\{ Z_{14}=0\right\} $.
As we have seen in the last Subsection, these lie on the locus $\left\{ a_{5}-\vartheta_{1}Z_{0}=Z_{14}=0\right\} $
that only intersects $D_{0}$ over $\left\{ a_{5}=Z_{0}=0\right\} $.

First, referring to \eqref{eq:DSp}, we rewrite the equation for $\Delta^{\left(1\right)}$
in terms of the variables $\left(\left(a_{5}-\vartheta_{1}Z_{0}\right),Z_{14}\right)$
as

\begin{equation}
\begin{array}{c}
0=\left(\vartheta_{1}\left(a_{5}-\vartheta_{1}Z_{0}\right)+a_{4}Z_{14}\right)^{2}\\
\left(6a_{5}+4a_{420}\right)\left(a_{5}-\vartheta_{1}Z_{0}\right)Z_{14}^{2}-\left(4a_{5}+2a_{4}+4a_{2}\right)Z_{14}^{3}Z_{0}\\
-6\left(a_{5}-\vartheta_{1}Z_{0}\right)^{2}Z_{14}^{2}+8\left(\vartheta_{1}Z_{0}-a_{5}\right)Z_{14}^{3}Z_{0}-3Z_{14}^{4}Z_{0}^{2}.
\end{array}\label{liftgut}
\end{equation}
Next define 
\[
Q^{\left(2\right)}:=\left\{ \left|\begin{array}{cc}
a_{5}-\vartheta_{1}Z_{0} & Z_{14}\\
\vartheta_{2} & \tilde{Z}_{14}
\end{array}\right|=0\right\} \subseteq Q^{\left(1\right)}\times\mathbb{P}_{\left[\vartheta_{2},\tilde{Z}_{14}\right]}
\]
and 
\[
Z_{23}:=\frac{a_{5}-\vartheta_{1}Z_{0}}{\vartheta_{2}}=\frac{Z_{14}}{\tilde{Z}_{14}}
\]
for which we have 
\[
\pi^{\left(2\right)}:Q^{\left(2\right)}\rightarrow B_{3}.
\]
The equation for the proper transform $\Delta^{\left(2\right)}$ of
$\Delta^{\left(1\right)}$ becomes 
\begin{equation}
\begin{array}{c}
0=\left(\vartheta_{1}\vartheta_{2}+a_{4}\tilde{Z}_{14}\right)^{2}\\
+\left(\left(6a_{5}+4a_{420}\right)\vartheta_{2}-\left(4a_{5}+2a_{4}+4a_{2}\right)Z_{0}\tilde{Z}_{14}\right)\tilde{Z}_{14}^{2}Z_{23}\\
-\left(6\vartheta_{2}^{2}-8\vartheta_{2}Z_{0}\tilde{Z}_{14}+3\tilde{Z}_{14}^{2}Z_{0}^{2}\right)\tilde{Z}_{14}^{2}Z_{23}^{2}.
\end{array}\label{eq:delta2}
\end{equation}

The proper transform $\Delta^{\left(2\right)}$ in $Q^{\left(2\right)}$
is the total transform minus twice the exceptional divisor, so that
the proper transform is again twice the anticanonical divisor of $Q^{\left(2\right)}$.
Therefore we have a Calabi-Yau fourfold $W_{4}^{\left(2\right)}$
given by the completion over $\left\{ X=0\right\} $ of the solution
set of 
\begin{equation}
\begin{array}{c}
w_{2}^{2}=\left(\vartheta_{1}\vartheta_{2}+a_{4}\tilde{Z}_{14}\right)^{2}\\
+\left(\left(6a_{5}+4a_{420}\right)\vartheta_{2}-\left(4a_{5}+2a_{4}+4a_{2}\right)Z_{0}\tilde{Z}_{14}\right)\tilde{Z}_{14}^{2}Z_{23}\\
-\left(6\vartheta_{2}^{2}-8\vartheta_{2}Z_{0}\tilde{Z}_{14}+3\tilde{Z}_{14}^{2}Z_{0}^{2}\right)\tilde{Z}_{14}^{2}Z_{23}^{2}.
\end{array}\label{eq:W0}
\end{equation}

Over the exceptional locus, given by $\left\{ Z_{23}=0\right\} \subseteq Q^{\left(2\right)}$,
the equation for $\Delta^{\left(2\right)}$ is a perfect square so
that $W_{4}^{\left(2\right)}\times_{B_{3}}S_{\mathrm{GUT}}$ splits
into five components. Two of these are new components lying over $\left\{ Z_{23}=0\right\} \subseteq Q^{\left(2\right)}$
that we call $D_{2}$ and $D_{3}$. In addition we have lifted components,
the two components that we continue to call $D_{1}$ and $D_{4}$
lying over $\left\{ \tilde{Z}_{14}=0\right\} \subseteq Q^{\left(2\right)}$,
and finally the lifted component over $\left\{ Z_{0}=0\right\} $
that we continue to call $D_{0}$. We number things so that, over
a general point of $S_{\mathrm{GUT}}$, $D_{2}$ intersects $D_{1}$
and $D_{3}$ intersects $D_{4}$. The incidence of the five components
over a general point of $S_{\mathrm{GUT}}$ is that of the extended
Dynkin diagram $\tilde{A}_{4}$.

By Lemma \ref{lem:Singularities-of-the} singular points of $\Delta^{\left(2\right)}$
can only occur where $Z_{23}=0$, so $\tilde{Z}_{14}\text{\ensuremath{\neq}}0$
there and these points can be singular points of $\Delta^{\left(2\right)}$
only if $\frac{\vartheta_{1}\vartheta_{2}}{\tilde{Z}_{14}}=-a_{4}$
and 
\[
\frac{\vartheta_{2}}{\tilde{Z}_{14}}=\frac{\left(2a_{4}+4a_{2}+4a_{5}\right)Z_{0}}{6a_{5}+4a_{420}}.
\]
Multiplying both sides of this last equation by $\vartheta_{1}$ and
substituting, then recalling that at these points $a_{5}-\vartheta_{1}Z_{0}=0$
and that $a_{54320}=0$, we obtain 
\[
a_{5}a_{0}+a_{3}\left(a_{5}+a_{4}\right)=0
\]
that can be rewritten as the relation 
\[
Z_{0}\text{·}\left|\begin{array}{cc}
a_{4} & -a_{5}\\
a_{0}+a_{3} & a_{3}
\end{array}\right|=0
\]
given in \eqref{eq:bingo2}. 
\begin{lem}
\label{lem:Singular2}Singularities of the branch locus $\Delta^{\left(2\right)}$
only occur over 
\[
z=\left|\begin{array}{cc}
a_{4} & -a_{5}\\
a_{0}+a_{3} & a_{3}
\end{array}\right|=0
\]
and then only where 
\[
\frac{\vartheta_{2}}{\tilde{Z}_{14}}=\frac{\left(a_{4}+2a_{2}+2a_{5}\right)Z_{0}}{a_{5}-2a_{3}}.
\]
\end{lem}

Over a general point of $\left\{ z=a_{5}=0\right\} $, $D_{0}$ splits
into the two components $D_{01}+D_{04}$ that come to coincide with
the specialization of $D_{2}+D_{3}$ over $\left\{ z=a_{5}=0\right\} $
to form the two components of multiplicity two in the Dynkin diagram
for $D_{5}$ as is shown by the relation

\[
\vartheta_{1}Z_{0}+\vartheta_{2}Z_{23}=a_{5}.
\]

Over a general point of $\left\{ z=a_{420}=a_{53}=0\right\} $ again
$D_{0}$ splits into two components, augmenting the extended Dynkin
diagram $\tilde{A}_{4}$ to the extended Dynkin diagram $\tilde{A}_{5}$.

\subsubsection{Singularities of higher codimension\label{subsec:Singularities-of-higher}}

We have seen in Lemma \ref{lem:Singular2} that, at singular points
of $\Delta^{\left(2\right)}$, $Z_{23}=0$ and $\tilde{Z}_{14}\text{\ensuremath{\neq}}0$
and they lie over 
\begin{equation}
\left\{ z=\left|\begin{array}{cc}
a_{4} & -a_{5}\\
a_{0}+a_{3} & a_{3}
\end{array}\right|=0\right\} \subseteq B_{3}.\label{eq:keyformula}
\end{equation}
In fact, if one rewrites \eqref{eq:W0} in the form 
\begin{equation}
\left|\begin{array}{cc}
w_{2}-\left(\vartheta_{1}\vartheta_{2}+a_{4}\tilde{Z}_{14}\right) & \tilde{Z}_{14}^{2}Z_{23}\\
-A & w_{2}+\left(\vartheta_{1}\vartheta_{2}+a_{4}\tilde{Z}_{14}\right)
\end{array}\right|=0\label{eq:matrix1}
\end{equation}
where 
\[
\begin{array}{c}
-A=\left(\left(6a_{5}+4a_{420}\right)\vartheta_{2}-\left(4a_{5}+2a_{4}+4a_{2}\right)Z_{0}\tilde{Z}_{14}\right)\\
-\left(6\vartheta_{2}^{2}-8\vartheta_{2}Z_{0}\tilde{Z}_{14}+3\tilde{Z}_{14}^{2}Z_{0}^{2}\right)Z_{23}.
\end{array}
\]
over $\left\{ \tilde{Z}_{14}=0\right\} $ the equation of $W_{4}^{\left(2\right)}$
becomes 
\[
\left(w_{2}-\vartheta_{1}\vartheta_{2}\right)\left(w_{2}+\vartheta_{1}\vartheta_{2}\right)=0.
\]
Recalling that $\tilde{Z}_{14}$ and $\vartheta_{1}\vartheta_{2}$
cannot vanish simultaneously, one sees that 
\begin{equation}
C_{\mathbf{\bar{5}}}^{\left(44\right)}\cap\left\{ \tilde{Z}_{14}=0\right\} =\textrm{Ø}\label{eq:Ccurve}
\end{equation}
where $C_{\mathbf{\bar{5}}}^{\left(44\right)}$ is given by the vanishing
\begin{equation}
\begin{array}{c}
w_{2}=\vartheta_{1}\vartheta_{2}+a_{4}\tilde{Z}_{14}=Z_{23}=\\
\left(a_{5}-2a_{3}\right)\vartheta_{2}-\left(a_{4}+2a_{2}+2a_{5}\right)Z_{0}\tilde{Z}_{14}=0
\end{array}\label{eq:Higgs curve}
\end{equation}
of all four entries in the $2\times2$ matrix \eqref{eq:matrix1}
and is a simple nodal locus of $W_{4}^{\left(2\right)}$. Since $\tilde{Z}_{14}$
is never zero along $C_{\mathbf{\bar{5}}}^{\left(44\right)}$ we conclude
that 
\[
\left(D_{1}\cup D_{4}\right)\cap C_{\mathbf{\bar{5}}}^{\left(44\right)}=\textrm{Ø}
\]
where $C_{\mathbf{\bar{5}}}^{\left(44\right)}$ is given on $\left\{ Z_{23}=0\right\} $
by the equations
\[
\begin{array}{c}
\vartheta_{1}\frac{\vartheta_{2}}{\tilde{Z}_{14}}+a_{4}=0\\
\left(a_{5}-2a_{3}\right)\frac{\vartheta_{2}}{\tilde{Z}_{14}}-\left(a_{4}+2a_{2}+2a_{5}\right)Z_{0}=0.
\end{array}
\]

Since $Z_{0}$ and $Z_{23}$ are symmetric and all other variables
are skew, the four entries in \eqref{eq:matrix1} are all invariant
under the action of $\tilde{\beta}_{4}/\beta_{3}$. Therefore we can
take either small blow-up to complete the crepant resolution.

\subsection{The smooth model $\tilde{W}_{4}$}

After taking the small resolution of the nodal curve $C_{\mathbf{\bar{5}}}^{\left(44\right)}$
in $W_{4}^{\left(2\right)}$ as described just above, we obtain the
inclusion of proper transform 
\[
\tilde{\Delta}\subseteq\tilde{W}_{4}
\]
where $\tilde{\Delta}$ is non-singular.

We retain the notations $D_{j}$, $j=0,\ldots,4$, for the proper
transforms in $\tilde{W}_{4}$ of the corresponding divisors in $W_{4}^{\left(2\right)}$.
We will let $G_{j}$ denote the canonical section of the line bundle
determined by the divisor $D_{j}$, that is 
\[
D_{j}=\left\{ G_{j}=0\right\} .
\]
\begin{rem}
The exceptional divisors $D_{j}$ are identified with the positive
simple roots of $SU\left(5\right)_{gauge}$ in such a way that the
involution $\beta_{3}$ on $B_{3}$ and the action
\[
\frac{w_{0}}{\vartheta_{0}}\mapsto-\frac{w_{0}}{\vartheta_{0}}
\]
on the fibers of $\tilde{W}_{4}/B_{3}$ induces the non-trivial geometric
action 
\begin{equation}
\begin{array}{c}
D_{0}=-\left(D_{1}+D_{2}+D_{3}+D_{4}\right)\mapsto-D_{0}\\
D_{1}\mapsto-D_{4}\\
D_{2}\mapsto-D_{3}\\
D_{3}\mapsto-D_{2}\\
D_{4}\mapsto-D_{1}
\end{array}\label{eq:Croot2}
\end{equation}
on the exceptional fibers of $\tilde{W}_{4}$. As explained in \cite{Clemens-1},
the action of $\tilde{\beta}_{4}$ on the roots $D_{1},\ldots,D_{4}$
reverses the choice of positive Weyl chamber used in making the dictionary
between exceptional divisors of $\tilde{W}_{4}$ and the $SU\left(5\right)_{gauge}$-roots.
This reversal of positive Weyl chamber exactly reverses the non-trivial
involution \eqref{eq:Croot2} thereby preserving the $SU\left(5\right)_{gauge}$-symmetry
of the quotient $W_{4}^{\vee}$.\footnote{This assertion is reflected in the fact that the action of $\tilde{\beta}_{4}$
leaves all entries of equation (4.10) of \cite{Clemens-1} invariant.} 
\end{rem}

The identity 
\[
a_{5}=Z_{0}\vartheta_{1}+Z_{23}\vartheta_{2}
\]
intertwines $D_{0}$ and $D_{2}+D_{3}$ above $\left\{ z=a_{5}=0\right\} $.
After the first modification, the equation of the spectral variety
on $\left\{ Z_{0}=0\right\} $ is 
\[
\left(a_{5}\vartheta_{1}+a_{4}Z_{14}\right)^{2}-4a_{5}a_{420}Z_{14}^{2}=0
\]
whereas the equations of the center of the second modification are
\[
a_{5}-\vartheta_{1}=Z_{14}=0.
\]
Since $\vartheta_{1}$ and $Z_{14}$ cannot vanish simultaneously,
therefore after the second modification, the proper transform of $\left\{ Z_{0}=0\right\} $
will only intersect $\left\{ Z_{23}=0\right\} $ along the fiber of
$\left\{ Z_{23}=0\right\} $ over $\left\{ z=a_{5}=0\right\} $. If
the modifications were done in the opposite order, the divisors $\left\{ Z_{0}=0\right\} $
and $\left\{ Z_{23}=0\right\} $ in $\tilde{\Delta}$ would have been
flopped.

\subsection{The spectral divisor in $\tilde{W}_{4}$\label{sec:The-spectral-variety}}

Referring to \eqref{eq:specdiv}, the spectral divisor is obviously
invariant under the involution $\left(x,y\right)\mapsto\left(x,-y\right)$.
To compute the image 
\[
\mathcal{C}_{0}^{\left(4\right)}+\mathcal{C}_{0}^{\left(1\right)}\subseteq Q
\]
of $\mathcal{D}^{\left(4\right)}+\mathcal{D}^{\left(1\right)}$ we
write 
\[
t=\frac{y}{x}=\frac{Y-\left(z+a_{420}\right)z^{2}w}{X+z^{2}w}
\]
so that 
\[
t-z=\frac{Y-zX-\left(2z+a_{420}\right)z^{2}w}{X+z^{2}w}
\]
and 
\[
t+z=\frac{Y+zX-a_{420}z^{2}w}{X+z^{2}w}.
\]
So, setting $\frac{X}{z^{2}w}=1$ and $\frac{Y}{z^{2}w}=\vartheta$,
we have
\[
\begin{array}{c}
t=\frac{\vartheta_{0}-\left(z+a_{420}\right)}{2}\\
t-z=\frac{\vartheta_{0}-\left(3z+a_{420}\right)}{2}\\
t+z=\frac{\vartheta_{0}+z-a_{420}}{2}
\end{array}
\]
and so by \eqref{eq:preHiggs} the affine equations 
\begin{equation}
\begin{array}{c}
\left(\vartheta_{0}-\left(z+a_{420}\right)\right)^{2}\left(\left(a_{5}\left(\vartheta_{0}-\left(z+a_{420}\right)\right)+2a_{4}z\right)\left(\vartheta_{0}+z-a_{420}\right)-4a_{420}z^{2}\right)\\
-8a_{0}z^{3}\left(\vartheta_{0}+z-a_{420}\right)=0
\end{array}\label{eq:speceq}
\end{equation}
and 
\[
\vartheta_{0}-\left(3z+a_{420}\right)=0
\]
are the respective equations for the image $\mathcal{C}_{0}^{\left(4\right)}+\mathcal{C}_{0}^{\left(1\right)}\subseteq Q$
of the components of the spectral divisor. Setting $z=0$ we obtain
{]}
\begin{equation}
\begin{array}{c}
a_{5}\left(\vartheta_{0}-a_{420}\right)^{4}=0\\
\vartheta_{0}-a_{420}=0
\end{array}\label{eq:sGUT}
\end{equation}
consistent with the fact that the inverse image of both $\mathcal{C}_{0}^{\left(4\right)}$
and $\mathcal{C}_{0}^{\left(1\right)}$ in $\overline{W}_{4}$ are
reducible, only one of their two components correspond to the image
\begin{equation}
\mathcal{\bar{C}}_{Higgs}=\mathcal{\bar{C}}_{Higgs}^{\left(4\right)}+\mathcal{\bar{C}}_{Higgs}^{\left(1\right)}\subseteq\overline{W}_{4}\label{eq:specHiggs}
\end{equation}
of $\mathcal{D}^{\left(4\right)}+\mathcal{D}^{\left(1\right)}$ and
only $\mathcal{C}_{0}^{\left(4\right)}$ intersects the proper transform
$\left\{ X=0\right\} $ of $\left(\zeta\right)+\left(\tau\right)$
and it only intersects simply along one of the two sections. We let
\[
\mathcal{C}_{Higgs}=\mathcal{C}_{Higgs}^{\left(4\right)}+\mathcal{C}_{Higgs}^{\left(1\right)}\subseteq\tilde{Q}=Q^{\left(2\right)}
\]
denote the proper (also the total) transform of \eqref{eq:specHiggs}
and 
\[
\mathcal{\tilde{C}}_{Higgs}\subseteq\tilde{W}_{4}
\]
denote the proper (also the total) transform of $\mathcal{\bar{C}}_{Higgs}$.
Then 
\[
\mathcal{C}_{Higgs}:=\left(\pi_{\tilde{Q}}\right)_{\ast}\left(\tilde{\mathcal{C}}_{Higgs}\right).
\]

Under the first modification
\[
\begin{array}{c}
t=\frac{\vartheta_{1}Z_{0}-\left(Z_{0}Z_{14}+a_{420}\right)}{2}\\
t-z=\frac{\vartheta_{1}Z_{0}-\left(3Z_{0}Z_{14}+a_{420}\right)}{2}\\
t+z=\frac{\vartheta_{1}Z_{0}-\left(Z_{0}Z_{14}+a_{420}\right)}{2}
\end{array}
\]
 so that \eqref{eq:speceq} becomes
\[
\begin{array}{c}
\left(\vartheta_{1}Z_{0}-\left(Z_{14}Z_{0}+a_{420}\right)\right)^{2}\left(\begin{array}{c}
\left(a_{5}\left(\vartheta_{1}Z_{0}-\left(Z_{14}Z_{0}+a_{420}\right)\right)+2a_{4}Z_{14}Z_{0}\right)\left(\vartheta_{1}Z_{0}+Z_{14}Z_{0}-a_{420}\right)\\
-4a_{420}\left(Z_{14}Z_{0}\right)^{2}
\end{array}\right)\\
-8a_{0}\left(Z_{14}Z_{0}\right)^{3}\left(\vartheta_{1}Z_{0}+Z_{14}Z_{0}-a_{420}\right)=0.
\end{array}
\]
This last implies that $\mathcal{C}_{Higgs}$ is the total transform
of $\mathcal{C}_{0}$ so that 
\[
\mathcal{C}_{Higgs}\cdot Z_{0}\subseteq\left\{ a_{420}^{4}\text{·}a_{5}=Z_{0}=0\right\} 
\]
and therefore

\begin{equation}
\mathcal{C}_{Higgs}^{\left(4\right)}\times_{B_{3}}S_{\mathrm{GUT}}\subseteq\left\{ Z_{14}=0\right\} \label{eq:specGUT}
\end{equation}
where its equation is
\[
\left(\vartheta_{1}Z_{0}-a_{420}\right)^{4}=0
\]
while the equation for $\mathcal{C}_{Higgs}^{\left(1\right)}$ is
just $\vartheta_{1}Z_{0}-a_{420}=0$. Since $\vartheta_{1}$ and $Z_{14}$
cannot vanish simultaneously, $\vartheta_{1}$ can only vanish on
$\mathcal{C}_{Higgs}\times_{B_{3}}S_{\mathrm{GUT}}$ when $a_{420}=0$
and neither component of $\tilde{\mathcal{C}}_{Higgs}\times_{B_{3}}S_{\mathrm{GUT}}$
can lie in the inherited component $D_{1}\subseteq\tilde{W}_{4}$
containing the proper transform of $\left(\zeta\right)$. Thus 
\begin{equation}
\mathcal{\tilde{C}}_{Higgs}\times_{B_{3}}S_{\mathrm{GUT}}\subseteq\left\{ D_{4}=0\right\} .\label{eq:spec4}
\end{equation}

Under the second modification
\[
\begin{array}{c}
a_{5}=\vartheta_{1}Z_{0}+\vartheta_{2}Z_{23}\\
Z_{14}=\tilde{Z}_{14}Z_{23}.
\end{array}
\]
Now $\mathcal{C}_{Higgs}\times S_{\mathrm{GUT}}$ cannot lie in $\left\{ Z_{23}=0\right\} $
since there $\vartheta_{1}Z_{0}=a_{5}$ whereas on $\mathcal{C}_{Higgs}\times S_{\mathrm{GUT}}$
we have $\vartheta_{1}Z_{0}=a_{420}$. Therefore
\[
\mathcal{C}_{Higgs}\times S_{\mathrm{GUT}}\subseteq\tilde{Z}_{14}.
\]

Furthermore
\[
\begin{array}{c}
t=\frac{\vartheta_{1}Z_{0}-\left(Z_{0}Z_{14}+a_{420}\right)}{2}=\frac{a_{5}-\vartheta_{2}Z_{23}-\left(Z_{0}\tilde{Z}_{14}Z_{23}+a_{420}\right)}{2}\\
t-z=\frac{\vartheta_{1}Z_{0}-\left(3Z_{0}Z_{14}+a_{420}\right)}{2}=\frac{a_{5}-\vartheta_{2}Z_{23}-\left(3Z_{0}\tilde{Z}_{14}Z_{23}+a_{420}\right)}{2}\\
t+z=\frac{\vartheta_{1}Z_{0}-\left(Z_{0}Z_{14}+a_{420}\right)}{2}=\frac{a_{5}-\vartheta_{2}Z_{23}-\left(Z_{0}\tilde{Z}_{14}Z_{23}+a_{420}\right)}{2}.
\end{array}
\]
So substituting as above one can produce a (somewhat opaque) equation
for $\mathcal{C}_{Higgs}$ in $\tilde{Q}$.

Finally, by \eqref{eq:speceq}, $\mathcal{C}_{Higgs}^{\left(4\right)}+\mathcal{C}_{Higgs}^{\left(1\right)}$
has cohomology class 
\begin{equation}
\left(4\left(X_{2}\right)+5N\right)+\left(\left(X_{2}\right)+N\right)\label{eq:WF}
\end{equation}
in $Q.$

\subsubsection{Locating matter curves and the Higgs curve }

Referring to \eqref{eq:sGUT} $S_{\mathrm{GUT}}\subseteq B_{3}$ is
lifted into $\tilde{W}_{4}$ via 
\[
\left\{ \vartheta_{1}Z_{0}-a_{420}=\tilde{Z}_{14}=0\right\} \subseteq\tilde{Q}
\]
with the property that $\vartheta_{1}$ is never zero on the lifting.
The matter and Higgs curves will all lie in this lifting.

The matter curves will be given by 
\begin{equation}
\Sigma_{\mathbf{10}}^{\left(4\right)}:=\left\{ a_{5}=\tilde{Z}_{14}=\vartheta_{1}Z_{0}-a_{420}=0\right\} \subseteq\mathcal{C}_{Higgs}^{\left(4\right)}\label{eq:intinfo}
\end{equation}
and, since $\vartheta_{1}$ and $\tilde{Z}_{14}$ cannot vanish simultaneously,

\begin{equation}
\Sigma_{\mathbf{\bar{5}}}^{\left(41\right)}:=\left\{ Z_{0}=a_{420}=\tilde{Z}_{14}=0\right\} \subseteq\mathcal{C}_{Higgs}^{\left(4\right)}..\label{eq:intinfo'}
\end{equation}

By \eqref{eq:Ccurve} and \eqref{eq:Higgs curve} the singular curve
$C_{\mathbf{\bar{5}}}^{\left(44\right)}$ determining the Higgs curve
must lie on 
\[
\begin{array}{c}
\vartheta_{1}\vartheta_{2}+a_{4}\tilde{Z}_{14}=Z_{23}=\\
\left(a_{5}-2a_{3}\right)\vartheta_{2}-\left(a_{4}+2a_{2}-2a_{5}\right)Z_{0}\tilde{Z}_{14}=0
\end{array}
\]
as well as 
\[
\left\{ \left|\begin{array}{cc}
a_{4} & -a_{5}\\
a_{0}+a_{3} & a_{3}
\end{array}\right|=0\right\} \cap\left\{ Z_{23}=0\right\} 
\]
so that $\vartheta_{1}\vartheta_{2}\text{\ensuremath{\neq}}0$ and
so by Lemma \ref{lem:Singular2}
\[
\frac{\vartheta_{2}}{\tilde{Z}_{14}}=\frac{\left(a_{4}+2a_{2}+2a_{5}\right)Z_{0}}{a_{5}-2a_{3}}
\]
and therefore
\begin{equation}
_{\Sigma_{\mathbf{\bar{5}}}^{\left(44\right)}:=\left\{ \left|\begin{array}{cc}
a_{4} & a_{5}\\
a_{0}+a_{3} & -a_{3}
\end{array}\right|=\tilde{Z}_{14}=\vartheta_{1}Z_{0}-a_{420}=0\right\} \subseteq\mathcal{C}_{Higgs}^{\left(4\right)}.}\label{eq:sigma44}
\end{equation}
Furthermore the two values 
\[
\frac{\tilde{t}}{\tilde{z}}=\text{\textpm}\sqrt{\frac{a_{0}}{a_{54}}}=\frac{\left(a_{5}-\vartheta_{2}Z_{23}-a_{420}\right)-Z_{0}\tilde{Z}_{14}Z_{23}}{2}
\]
in \eqref{eq:Higgspush} lie in $\mathcal{C}_{Higgs}^{\left(4\right)}$
but do not meet the singular curve $C_{\mathbf{44}}\subseteq\tilde{\Delta}$. 

The lifted spectral curves $\tilde{\Sigma}_{\mathbf{10}}^{\left(4\right)},\,\tilde{\Sigma}_{\mathbf{\bar{5}}}^{\left(41\right)},\,\tilde{\Sigma}_{\mathbf{\bar{5}}}^{\left(44\right)}\subseteq\mathcal{\tilde{C}}_{Higgs}$
and by \eqref{eq:spec4} have associated spectral surfaces given by
\begin{equation}
\begin{array}{c}
\tilde{E}_{\mathbf{10}}^{\left(4\right)}:=\tilde{\Sigma}_{\mathbf{10}}^{\left(4\right)}\times_{S_{\mathrm{GUT}}}D_{4}\cong E_{\mathbf{10}}^{\left(4\right)}:=\Sigma_{\mathbf{10}}^{\left(4\right)}\times_{S_{\mathrm{GUT}}}\left\{ \tilde{Z}_{14}=0\right\} .\\
\tilde{E}_{\mathbf{\bar{5}}}^{\left(41\right)}:=\tilde{\Sigma}_{\mathbf{\bar{5}}}^{\left(41\right)}\times_{S_{\mathrm{GUT}}}D_{4}\cong\Sigma_{\mathbf{\bar{5}}}^{\left(41\right)}:=\Sigma_{\mathbf{\bar{5}}}^{\left(41\right)}\times_{S_{\mathrm{GUT}}}\left\{ \tilde{Z}_{14}=0\right\} \\
\tilde{E}_{\mathbf{\bar{5}}}^{\left(44\right)}:=\tilde{\Sigma}_{\mathbf{\bar{5}}}^{\left(44\right)}\times_{S_{\mathrm{GUT}}}D_{4}\cong\Sigma_{\mathbf{\bar{5}}}^{\left(44\right)}:=\Sigma_{\mathbf{\bar{5}}}^{\left(44\right)}\times_{S_{\mathrm{GUT}}}\left\{ \tilde{Z}_{14}=0\right\} .
\end{array}\label{eq:specsetup}
\end{equation}

\subsection{Topology of $\tilde{W}_{4}$}

\subsubsection{Picard groups}

The Picard group of $\tilde{Q}$ is freely generated by generators
\[
\begin{array}{c}
\pi_{B_{3}}^{\ast}\left(\mathrm{Pic}\left(B_{3}\right)\right)\\
\left(X\right),\left(Z_{0}\right),\left(Z_{23}\right)
\end{array}
\]
where $\left(X\right)$ denotes the divisor given by $X_{0}=0$, etc.
The Picard group of $\tilde{W}_{4}$ is generated by 
\[
\pi_{\tilde{Q}}^{\ast}\mathrm{Pic}\left(\tilde{Q}\right),\left(\tilde{\zeta}\right),\left(\tilde{\tau}\right),D_{0},\ldots,D_{4}
\]
with relations 
\begin{equation}
\begin{array}{c}
\pi_{\tilde{Q}}^{\ast}\left(\left(X_{0}\right)\right)=\left(\tilde{\zeta}\right)+\left(\tilde{\tau}\right)\\
\pi_{\tilde{Q}}^{\ast}\left(\left(Z_{34}\right)\right)=D_{2}+D_{3}\\
\pi_{\tilde{Q}}^{\ast}\left(\left(Z_{1}\right)\right)=D_{0}\\
\tilde{\pi}_{B_{3}}^{\ast}\left(S_{\mathrm{GUT}}\right)=\sum_{j=0}^{4}D_{j}.
\end{array}\label{eq:Picards}
\end{equation}
The cohomology class of $\tilde{\Delta}\subseteq\tilde{Q}$ is given
by 
\begin{equation}
4\left(\left(X_{2}\right)+N\right)-2\left(\left(Z_{0}\right)+\left(Z_{23}\right)\right)\label{eq:Higgshom-1}
\end{equation}
and 
\[
K_{\tilde{Q}}=-\left(2\left(X_{2}\right)+2N\right)+\left(\left(Z_{0}\right)+\left(Z_{23}\right)\right)
\]
where, as before, $N=\tilde{\pi}^{\ast}\left(c_{1}\left(\mathcal{N}\right)\right)$
and $\mathcal{N}$ is the line bundle on $B_{3}$ whose sections include
$z$ and the $a_{j}$. Thus the Picard group of $\tilde{\Delta}$
admits the effective square root 
\[
2\left(\left(X_{2}\right)+N\right)-\left(\left(Z_{0}\right)+\left(Z_{23}\right)\right)
\]
of the branch locus of $\pi_{\tilde{Q}}$. Also 
\begin{equation}
\begin{array}{c}
\mathcal{C}_{Higgs}=\mathcal{C}_{Higgs}^{\left(4\right)}+\mathcal{C}_{Higgs}^{\left(1\right)}=\mathcal{C}_{Higgs}^{\left(4\right)}+image\left(\tilde{\tau}\right)\\
\equiv5\left(X_{2}\right)+5N\in Pic\left(\tilde{Q}\right).
\end{array}\label{eq:Higgsint}
\end{equation}
We have the following linear equivalences on $\tilde{Q}$: 
\begin{equation}
\begin{array}{c}
\left(\vartheta_{0}\right)=\left(Y_{0}\right)-\left(X_{0}\right)\equiv N\\
\left(\vartheta_{1}\right)=\left(Y_{1}\right)-\left(X_{1}\right)=\left(\left(Y_{0}\right)-\left(Z_{0}\right)\right)-\left(X_{0}\right)\\
\left(\vartheta_{2}\right)=\left(Y_{2}\right)-\left(X_{2}\right)=\left(\left(Y_{0}\right)-\left(Z_{23}\right)\right)-\left(X_{0}\right)\\
\left(\vartheta_{1}\right)+\left(Z_{0}\right)=\left(\vartheta_{2}\right)+\left(Z_{23}\right)\equiv N.
\end{array}\label{eq:inttbl}
\end{equation}
Finally, since none of the blow-ups in the resolution over $B_{3}$
touch $\left(X_{0}\right)$, from now on we will simply identify 
\[
\left(X\right):=\left(X_{0}\right)=\left(X_{1}\right)=\left(X_{2}\right).
\]

\subsubsection{Intersections in $Q^{\left(2\right)}$}

We compute push-forwards to $B_{3}$ of intersections in $\tilde{Q}$
as follows. From the fact that $\left(\pi_{B_{3}}\right)_{\ast}\left(\left(Z_{0}\right)\cdot\left(Z_{23}\right)\right)$
is supported on the curve $\left\{ z=a_{5}=0\right\} $ we conclude
that 
\[
\left(\pi_{B_{3}}\right)_{\ast}\left(\left(Z_{0}\right)\cdot\left(Z_{23}\right)\right)\equiv0.
\]
Appealing to \eqref{eq:Picards} and \eqref{eq:Higgshom-1}, we compute
\begin{equation}
\begin{array}{c}
\left(\pi_{B_{3}}\right)_{\ast}\left(\left(X\right)\cdot_{\tilde{Q}}\left(X\right)\right)\equiv-N\\
\left(\pi_{B_{3}}\right)_{\ast}\left(\left(X\right)\cdot_{\tilde{Q}}\left(Z_{0}\right)\right)\equiv\left(\pi_{B_{3}}\right)_{\ast}\left(\left(X\right)\cdot_{\tilde{Q}}\left(Z_{23}\right)\right)\equiv0\\
\left(\pi_{B_{3}}\right)_{\ast}\left(\left(X\right)\cdot_{\tilde{Q}}\mathcal{C}_{Higgs}\right)\equiv0\\
\left(\pi_{B_{3}}\right)_{\ast}\left(\left(Z_{0}\right)^{2}\right)\equiv-N\\
\left(\pi_{B_{3}}\right)_{\ast}\left(\left(Z_{23}\right)^{2}\right)\equiv-N.
\end{array}\label{eq:Higgsgen}
\end{equation}

\section{\label{sec:Higgs-line-bundle}Higgs line bundle and the $G$-flux }

\subsection{Physical interpretation}

We begin with the projection map 
\[
\mathcal{D}\subseteq B_{3}\times\mathbb{P}_{\left[\tilde{t},\tilde{z}\right]}\overset{\psi}{\longrightarrow}B_{3}
\]
Recall that $E_{8}$ has the subgroup $SU\left(5\right)_{gauge}\times SU\left(4\right)_{Higgs}\times U\left(1\right)_{X}$
and the Higgs operator is a non-trivial element of the center of the
enveloping algebra of $SU\left(4\right)_{Higgs}$. We use it to break
$E_{8}$-symmetry of the subgroup $SU\left(4\right)_{Higgs}$. To
accomplish this we must identify a line bundle 
\[
\mathcal{L}_{Higgs}
\]
on $\mathcal{D}$ such that the first Chern class of the rank-4 vector
bundle 
\[
\mathcal{E}:=\psi_{\ast}\left(\mathcal{L}_{Higgs}\right)
\]
is zero. We then geometrically specify the eigenvalues of the Higgs
operator acting on the complexified Cartan subalgeba of $SU\left(4\right)_{Higgs}$.We
identify the fibers of $\mathcal{E}$ with the representation space
for the fundamental representation of $SU\left(4\right)$ so that
the roots 
\[
\frac{t}{z}=\tilde{t}_{1},\ldots,\tilde{t}_{4}
\]
of 
\[
a_{5}t^{4}+a_{54}t^{3}z-a_{20}t^{2}z^{2}-a_{0}z^{3}\left(t+z\right)=0
\]
in \eqref{eq:preHiggs} become the eigenvalues of the Higgs operator
on the eigenvalues of the standard representation of $SU\left(4\right)_{Higgs}$.
The condition 
\[
c_{1}\left(\mathcal{E}\right)=0
\]
becomes the condition that the sum of these eigenvalues is zero.

\subsection{The Higgs bundle}

So we begin with any line bundle $\mathcal{L}$ on $\mathcal{D}^{\left(4\right)}$.
By the Grothendieck Riemann-Roch theorem, 
\[
c_{1}\left(\psi_{\ast}\left(\mathcal{L}\right)\right)=\psi_{\ast}\left(c_{1}\left(\mathcal{L}\right)\right)-\psi_{\ast}\left(c_{1}\left(\mathcal{O}_{\mathcal{D}}\left(\frac{R}{2}\right)\right)\right)
\]
where $R$ is the ramification divisor of $\psi$.

To construct $\mathcal{L}_{Higgs}$ we proceed as follows. Applying
the formula for the discriminant of a fourth-degree equation \eqref{eq:loc1}
we obtain that the discriminant has class $6\text{·}N$ on $B_{3}$.
Thus we must have 
\[
\psi_{\ast}\left(c_{1}\left(\mathcal{O}_{\mathcal{D}}\left(\frac{R}{2}\right)\right)\right)=c_{1}\left(\mathcal{N}^{3}\right)
\]
so that we must choose an effective divisor class on $\mathcal{D}^{\left(4\right)}$
whose push-forward to $B_{3}$ has class $c_{1}\left(\mathcal{N}^{3}\right)$.

One obvious line bundle to use is 
\begin{equation}
\left.\mathcal{O}_{B_{3}}\left(N\right)\boxtimes\mathcal{O}_{\mathbb{P}_{\left[\tilde{t},\tilde{z}\right]}}\left(-1\right)\right|_{\mathcal{D}}\label{eq:approxHiggs}
\end{equation}
that pushes forward on $B_{3}$ to 
\[
3N.
\]

We denote the line bundle on $\mathcal{D}^{\left(4\right)}$ given
by this divisor as 
\[
\mathcal{L}_{Higgs}^{\left(0\right)}=\left.\mathcal{O}_{B_{3}}\left(N\right)\boxtimes\mathcal{O}_{\mathbb{P}_{\left[U,V\right]}}\left(-1\right)\right|_{\mathcal{D}^{\left(4\right)}}.
\]
From Subsection \ref{sec:The-spectral-variety} the matter curves
are given by 
\[
\begin{array}{c}
\Sigma_{\mathbf{10}}^{\left(4\right)}=\left\{ t/z=\infty\right\} \text{·}\mathcal{D}^{\left(4\right)}\\
\Sigma_{\mathbf{\bar{5}}}^{\left(41\right)}=\left\{ t/z=1\right\} \text{·}\mathcal{D}^{\left(4\right)}
\end{array}
\]
and, as in \cite{Clemens-2} the Higgs curve is given by a branched
double cover $\hat{\Sigma}_{\mathbf{\bar{5}}}^{\left(44\right)}\subseteq\mathcal{D}^{\left(4\right)}$
of the curve 
\[
\Sigma_{\mathbf{\bar{5}}}^{\left(44\right)}=\left\{ z=\left|\begin{array}{cc}
a_{4} & -a_{5}\\
a_{3}+a_{0} & a_{3}
\end{array}\right|=0\right\} \subseteq B_{3}.
\]
We will need the skew component of the push-forward to $\Sigma_{\mathbf{\bar{5}}}^{\left(44\right)}$
of $\mathcal{L}_{Higgs}^{\left(0\right)}$ . Since the skew component
of the restriction of $\mathcal{O}_{\mathbb{P}_{\left[\tilde{t},\tilde{z}\right]}}\left(-1\right)$
is trivial, we need only compute the skew component of the push-forward
of the pull-back of $\mathcal{O}_{B_{3}}\left(N\right)$. Again by
the Grothendieck Riemann-Roch theorem, this is a line bundle on $\Sigma_{\mathbf{\bar{5}}}^{\left(44\right)}$
whose square is the restriction of $\mathcal{O}_{B_{3}}\left(2N\right)$.
We will denote this line bundle as $\mathcal{L}_{Higgs}^{\left(0,-\right)}$.
Since the restriction of $\mathcal{O}_{B_{3}}\left(3N\right)$ to
$\Sigma_{\mathbf{\bar{5}}}^{\left(44\right)}$ is its canonical bundle,
the restriction of $\mathcal{L}_{Higgs}^{\left(0,-\right)}$ will
be a theta-characteristic of $\Sigma_{\mathbf{\bar{5}}}^{\left(44\right)}$
.

\subsection{Restrictions to the matter and Higgs curves\label{subsec:Restrictions-to-the}}

From (4.1) of \cite{Clemens-2} we have 
\[
c_{1}\left(K_{B_{3}}^{-1}\right)^{3}=12
\]
and so 
\[
\begin{array}{c}
deg\left(K_{\Sigma_{\mathbf{10}}^{\left(4\right)}}\right)=N^{3}=12\\
deg\left(K_{\Sigma_{\mathbf{\bar{5}}}^{\left(41\right)}}\right)=N^{3}=12\\
deg\left(K_{\Sigma_{\mathbf{\bar{5}}}^{\left(44\right)}}\right)=4N^{3}=48.
\end{array}
\]
Since the degree of the canonical bundle of a curve is $2g-2$, the
genus of $\Sigma_{\mathbf{\bar{5}}}^{\left(44\right)}$ is $25$ and
the genus of both $\Sigma_{\mathbf{10}}^{\left(4\right)}$ and $\Sigma_{\mathbf{\bar{5}}}^{\left(41\right)}$
is $7$.

To compute with $\mathcal{L}_{Higgs}^{\left(0\right)}$ we work on
the preimage $\mathcal{D}$ of $\mathcal{C}_{Higgs}$ in $B_{3}\times\mathbb{P}_{\left[U,V\right]}$
where 
\[
\mathcal{L}_{Higgs}^{\left(0\right)}=K_{B_{3}}^{-1}\boxtimes\mathcal{O}_{\mathbb{P}_{\left[U,V\right]}}\left(-1\right).
\]
Thus 
\[
\left.\mathcal{L}_{Higgs}^{\left(0\right)}\right|_{\Sigma_{\mathbf{10}}^{\left(4\right)}}=\left.\mathcal{N}\right|_{\Sigma_{\mathbf{10}}^{\left(4\right)}}=\left.K_{B_{3}}^{-1}\right|_{\Sigma_{\mathbf{10}}^{\left(4\right)}}=K_{\Sigma_{\mathbf{10}}^{\left(4\right)}}
\]
and 
\[
\left.\mathcal{L}_{Higgs}^{\left(0\right)}\right|_{\Sigma_{\mathbf{\bar{5}}}^{\left(41\right)}}=\left.\mathcal{N}\right|_{\Sigma_{\mathbf{\bar{5}}}^{\left(41\right)}}=\left.K_{B_{3}}^{-1}\right|_{\Sigma_{\mathbf{\bar{5}}}^{\left(41\right)}}=K_{\Sigma_{\mathbf{\bar{5}}}^{\left(41\right)}}.
\]
As we showed just above 
\[
\left.\mathcal{L}_{Higgs}^{\left(0,-\right)}\right|_{\Sigma_{\mathbf{\bar{5}}}^{\left(44\right)}}=\left.\mathcal{N}\right|_{\Sigma_{\mathbf{\bar{5}}}^{\left(44\right)}}
\]
so that, in particular 
\[
\left(\left.\mathcal{L}_{Higgs}^{\left(0,-\right)}\right|_{\Sigma_{\mathbf{\bar{5}}}^{\left(44\right)}}\right)^{2}=K_{\Sigma_{\mathbf{\bar{5}}}^{\left(44\right)}}.
\]
In fact, $\left.\mathcal{L}_{Higgs}^{\left(0\right)}\right|_{\Sigma_{\mathbf{\bar{5}}}^{\left(44\right)}}$
is nothing more than the theta-characteristic $Z_{2}\text{·}Z$ in
§5.1 of \cite{Clemens-2} for $Z_{2}=\Sigma_{\mathbf{\bar{5}}}^{\left(44\right)}$.

Thus 
\[
\begin{array}{c}
h^{0}\left(\left.\mathcal{L}_{Higgs}^{\left(0\right)}\right|_{\Sigma_{\mathbf{10}}^{\left(4\right)}}\right)-h^{1}\left(\left.\mathcal{L}_{Higgs}^{\left(0\right)}\right|_{\Sigma_{\mathbf{10}}^{\left(4\right)}}\right)=7-1=6\\
h^{0}\left(\left.\mathcal{L}_{Higgs}^{\left(0\right)}\right|_{\Sigma_{\mathbf{\bar{5}}}^{\left(41\right)}}\right)-h^{1}\left(\left.\mathcal{L}_{Higgs}^{\left(0\right)}\right|_{\Sigma_{\mathbf{\bar{5}}}^{\left(41\right)}}\right)=7-1=6\\
h^{0}\left(\left.\mathcal{L}_{Higgs}^{\left(0,-\right)}\right|_{\Sigma_{\mathbf{\bar{5}}}^{\left(44\right)}}\right)-h^{1}\left(\left.\mathcal{L}_{Higgs}^{\left(0,-\right)}\right|_{\Sigma_{\mathbf{\bar{5}}}^{\left(44\right)}}\right)=0.
\end{array}
\]
We therefore have the desired Euler characteristics for the Higgs
bundle on the matter and Higgs curves. However the ranks of the relevant
spaces of sections are not quite right.

We will rectify the undesired outcome by modifying $\mathcal{L}_{Higgs}^{\left(0\right)}$
by recalling from §5.1 of \cite{Clemens-2} that 
\[
\left(\left(S_{\mathrm{GUT}}\cap\left(F_{\dot{x}\dot{z}\ddot{y}\ddot{w}}\cup F_{\dot{y}\dot{w}\ddot{x}\ddot{z}}\right)\right)\times\mathbb{P}_{\left[U,V\right]}\right)\subseteq\mathcal{D}.
\]
This allows us to define 
\[
\mathcal{L}_{Higgs}=\mathcal{L}_{Higgs}^{\left(0\right)}\otimes\mathcal{O}_{\mathcal{D}}\left(m\text{·}\left(\left(S_{\mathrm{GUT}}\cap\left(F_{\dot{x}\dot{z}\ddot{y}\ddot{w}}-F_{\dot{y}\dot{w}\ddot{x}\ddot{z}}\right)\right)\times\mathbb{P}_{\left[U,V\right]}\right)\right)
\]
for $m\text{\ensuremath{\neq}}0$ as in Lemma 6 of \cite{Clemens-2}
and define 
\[
\begin{array}{c}
\mathcal{L}_{\mathbf{10}}^{\left(4\right)}:=\left.\mathcal{L}_{Higgs}\right|_{\Sigma_{\mathbf{10}}^{\left(4\right)}}\\
\mathcal{L}_{\mathbf{\bar{5}}}^{\left(41\right)}:=\left.\mathcal{L}_{Higgs}\right|_{\Sigma_{\mathbf{\bar{5}}}^{\left(41\right)}}\\
\mathcal{L}_{\mathbf{\bar{5}}}^{\left(44\right)}:=\left.\mathcal{L}_{Higgs}^{\left(0,-\right)}\right|_{\Sigma_{\mathbf{\bar{5}}}^{\left(44\right)}}.
\end{array}
\]
Since the restriction of $\mathcal{O}_{\mathcal{D}}\left(m\text{·}\left(\left(S_{\mathrm{GUT}}\cap\left(F_{\dot{x}\dot{z}\ddot{y}\ddot{w}}-F_{\dot{y}\dot{w}\ddot{x}\ddot{z}}\right)\right)\times\mathbb{P}_{\left[U,V\right]}\right)\right)$
to the respective curves is not the trivial bundle, we conclude that
\begin{equation}
\begin{array}{c}
h^{0}\left(\mathcal{L}_{\mathbf{10}}^{\left(4\right)}\right)=6\\
h^{1}\left(\mathcal{L}_{\mathbf{10}}^{\left(4\right)}\right)=0
\end{array}\label{eq:yes1}
\end{equation}
and 
\begin{equation}
\begin{array}{c}
h^{0}\left(\mathcal{L}_{\mathbf{\bar{5}}}^{\left(41\right)}\right)=6\\
h^{1}\left(\mathcal{L}_{\mathbf{\bar{5}}}^{\left(41\right)}\right)=0.
\end{array}\label{eq:yes2}
\end{equation}

Furthermore by Lemma 7 of \cite{Clemens-2} we have with respect to
the eigenbundles of the involution $C_{u,v}$ that 
\begin{equation}
\begin{array}{c}
h^{0}\left(\left(\pi_{\Sigma_{\mathbf{\bar{5}}}^{\left(44\right)}}\right)_{\ast}\left(\mathcal{L}_{\mathbf{\bar{5}}}^{\left(44\right)}\right)^{\left[+1\right]}\right)=0\\
h^{1}\left(\left(\pi_{\Sigma_{\mathbf{\bar{5}}}^{\left(44\right)}}\right)_{\ast}\left(\left.\mathcal{L}_{Higgs}\right|_{\hat{\Sigma}_{\mathbf{\bar{5}}}^{\left(44\right)}}^{\left[-1\right]}\right)^{\left[+1\right]}\right)=0\\
h^{0}\left(\left(\pi_{\Sigma_{\mathbf{\bar{5}}}^{\left(44\right)}}\right)_{\ast}\left(\left.\mathcal{L}_{Higgs}\right|_{\hat{\Sigma}_{\mathbf{\bar{5}}}^{\left(44\right)}}^{\left[-1\right]}\right)^{\left[-1\right]}\right)=1\\
h^{1}\left(\left(\pi_{\Sigma_{\mathbf{\bar{5}}}^{\left(44\right)}}\right)_{\ast}\left(\left.\mathcal{L}_{Higgs}\right|_{\hat{\Sigma}_{\mathbf{\bar{5}}}^{\left(44\right)}}^{\left[-1\right]}\right)^{\left[-1\right]}\right)=1.
\end{array}\label{eq:yes3}
\end{equation}
This gives three-generations 
\begin{equation}
\begin{array}{c}
h^{0}\left(\Sigma_{\mathbf{10}}^{\left(4\right)\vee};\mathcal{L}_{\mathbf{10}}^{\left(4\right)\left[\text{\textpm}1\right]}\right)=h^{0}\left(\Sigma_{\mathbf{\bar{5}}}^{\left(41\right)\vee};\mathcal{L}_{\mathbf{\bar{5}}}^{\left(41\right)\left[\text{\textpm}1\right]}\right)=3\\
h^{1}\left(\Sigma_{\mathbf{10}}^{\left(4\right)\vee};\mathcal{L}_{\mathbf{10}}^{\left(4\right)\left[\text{\textpm}1\right]}\right)=h^{1}\left(\Sigma_{\mathbf{\bar{5}}}^{\left(41\right)\vee};\mathcal{L}_{\mathbf{\bar{5}}}^{\left(41\right)\left[\text{\textpm}1\right]}\right)=0
\end{array}\label{eq:yes!}
\end{equation}
as well as

\begin{equation}
\begin{array}{c}
h^{0}\left(\Sigma_{\mathbf{\bar{5}}}^{\left(44\right)\vee};\mathcal{L}_{\mathbf{\bar{5}}}^{\left(44\right)\left[+1\right]}\right)=h^{1}\left(\Sigma_{\mathbf{\bar{5}}}^{\left(44\right)\vee};\mathcal{L}_{\mathbf{\bar{5}}}^{\left(44\right)\left[+1\right]}\right)=0\\
h^{0}\left(\Sigma_{\mathbf{\bar{5}}}^{\left(44\right)\vee};\mathcal{L}_{\mathbf{\bar{5}}}^{\left(44\right)\left[-1\right]}\right)=h^{1}\left(\Sigma_{\mathbf{\bar{5}}}^{\left(44\right)\vee};\mathcal{L}_{\mathbf{\bar{5}}}^{\left(44\right)\left[-1\right]}\right)=1
\end{array}\label{eq:Higgs count}
\end{equation}
where the sign indicates the $C_{u,v}$-eigen-line-bundle.

\subsection{Action of $\tilde{\beta}_{4}$ on $\mathcal{L}_{Higgs}$\label{subsec:Action-of-and flat}}

The isomorphisms

\begin{equation}
\left(\tilde{\zeta}\right)\leftrightarrow B_{3}\leftrightarrow\left(\tilde{\tau}\right)\label{eq:match}
\end{equation}
induced by $\tilde{\beta}_{4}$ and $\pi_{B_{3}}$ identifies sections
of $\mathcal{O}_{B_{3}}\left(N\right)$ with sections of the co-normal
bundles $\mathcal{N}_{\left(\tilde{\zeta}\right)|\tilde{W}_{4}}^{\vee}$
and $\mathcal{N}_{\left(\tilde{\tau}\right)|\tilde{W}_{4}}^{\vee}$
respectively. A section of $\mathcal{O}_{B_{3}}\left(N\right)$ pulls
back to sections of $\mathcal{N}_{\left(\tilde{\zeta}\right)|\tilde{W}_{4}}^{\vee}$
and $\mathcal{N}_{\left(\tilde{\tau}\right)|\tilde{W}_{4}}^{\vee}$
respectively that are interchanged under the action of $\tilde{\beta}_{4}$
thereby identifying $\mathcal{N}_{\left(\tilde{\zeta}\right)|\tilde{W}_{4}}^{\vee}$
with $\tilde{\beta}^{\ast}\mathcal{N}_{\left(\tilde{\tau}\right)|\tilde{W}_{4}}^{\vee}$.
In \eqref{eq:yes3} we have used that identification to define an
involution on the rank-$2$ vector bundle $\left(\pi_{\mathcal{C}_{Higgs}^{\vee}}\right)_{\ast}\mathcal{L}_{Higgs}$
that decomposes into the direct sum of line bundles consisting of
sections that respectively symmetric and skew-symmetric. It is the
skew-symmetric line bundle on $\mathcal{C}_{Higgs}^{\vee}$ that we
call $\mathcal{L}_{Higgs}^{\vee}$.

Let 
\begin{equation}
\check{\mathcal{L}}_{\mathbf{10}}^{\left(4\right)},\ \mathcal{\check{\mathcal{L}}}_{\mathbf{\bar{5}}}^{\left(41\right)},\ \check{\mathcal{L}}_{\mathbf{\bar{5}}}^{\left(44\right)}\label{eq:iso}
\end{equation}
denote the restriction of the line bundle $\mathcal{L}_{Higgs}^{\vee}$
to the respective quotient curves.

\subsection{The $G$-flux in $\tilde{W}_{4}$\label{subsec:The--G-flux-in}}

The $G$-flux $\mathcal{G}_{\tilde{W}_{4}}$ is defined as the push-forward
into $\tilde{W}_{4}$ of the Chern class of $\mathcal{L}_{Higgs}$
modified by a $\tilde{\beta}_{4}$-invariant two cycle in $\left(X\right)$
so that it has intersection number zero with $\pi_{B_{3}}^{\ast}\left(C\right)$
for any curve $C$ in $S_{\mathrm{GUT}}$. Since we can essentially
identify $\mathcal{C}_{Higgs}\subseteq\tilde{Q}$ with its lifting
$\mathcal{\tilde{C}}_{Higgs}\subseteq\tilde{W}_{4}$ we will first
work to compute the class of the $G$-flux $\mathcal{G}_{\tilde{Q}}\subseteq\tilde{Q}$.
By \eqref{eq:Higgsint} 
\[
\mathcal{C}_{Higgs}^{\left(4\right)}\equiv4\left(X\right)+5N\in Pic\left(\tilde{Q}\right).
\]
Therefore the class of the push-forward of $c_{1}\left(\mathcal{L}_{Higgs}\right)$
to $\tilde{Q}$ is 
\begin{equation}
\left(4\left(X\right)+5N\right)\text{·}_{\tilde{Q}}\left(N+mF\right).\label{eq:G}
\end{equation}

Since $F=F_{+}-F_{-}$ projects into the curve $\left\{ z=\vartheta_{0}+a_{420}=0\right\} \subseteq B_{3}$
\[
\left(\pi_{B_{3}}\right)_{\ast}\left(mF\right)=0
\]
and we have by the projection formula that the intersection of \eqref{eq:G}
with $\pi_{B_{3}}^{\ast}\left(C\right)$ projects to $C\cdot_{B_{3}}\left(4N\right)$,
we must have 
\[
\mathcal{G}_{\tilde{Q}}\equiv\left(4N+5N\right)\text{·}_{\tilde{Q}}\left(N+mF\right).
\]
So 
\[
\left(\pi_{B_{3}}\right)_{\ast}\left(\mathcal{G}_{\tilde{Q}}^{2}\right)=0
\]
and therefore, recalling that $\mathcal{G}_{\tilde{W}_{4}}$ projects
birationally to $\mathcal{G}_{\tilde{Q}}$, 
\begin{equation}
\mathcal{G}_{\tilde{Q}}^{2}=\mathcal{G}_{\tilde{W}_{4}}^{2}=0.\label{eq:G-int}
\end{equation}

\subsection{Numerical conditions on the square of the $G$-flux\label{subsec:Numerical-conditions-on}}

The standard conditions that $\mathcal{G}_{W_{4}^{\vee}}$ will have
to satisfy are that 
\begin{equation}
\begin{array}{c}
\mathcal{G}_{W_{4}^{\vee}}\cdot_{W_{4}^{\vee}}\mathcal{G}_{W_{4}^{\vee}}\geq0\\
\mathcal{G}_{W_{4}^{\vee}}\cdot_{W_{4}^{\vee}}\mathcal{G}_{W_{4}^{\vee}}-\frac{\chi\left(W_{4}^{\vee}\right)}{24}\leq0
\end{array}\label{eq:important-1}
\end{equation}
where the subscript on the intersection dot indicates the space in
which the intersection is taking place. But by \eqref{eq:G-int},
\[
\mathcal{G}_{W_{4}^{\vee}}\cdot_{W_{4}^{\vee}}\mathcal{G}_{W_{4}^{\vee}}=\mathcal{G}_{\tilde{W}_{4}}^{2}=0
\]
so we must only check that $\frac{\chi\left(W_{4}^{\vee}\right)}{24}$
is a non-negative integer.

Now the involution $\tilde{\beta}_{4}$ on $W_{4}$ has four fixpoints
over each of the eight fixpoints of the action of the involution $\beta_{3}$
on $B_{3}$. As in formula (5.3) of \cite{Marsano:2012yc} each orbifold
point will affect the computation of $\chi\left(W_{4}^{\vee}\right)$
in \eqref{eq:important-1}. That is, since the (string-theoretic)
orbifold Euler characteristic of the fiber of $\tilde{W}_{4}/B_{3}$
over each fixpoint of $\beta_{3}$ in $B_{3}$ is $6,$ 
\[
\chi\left(W_{4}^{\vee}\right)=\frac{1}{2}\chi\left(\tilde{W}_{4}\right)+8\text{·}6.
\]
Since $\tilde{W}_{4}$ is smooth with a $\left(4+1\right)$-split
of $\mathcal{\tilde{C}}_{Higgs}$ we can apply the formula 
\[
\chi\left(\tilde{W}_{4}\right)=6\int_{B_{3}^{\vee}}\left(24c_{1}^{3}-44c_{1}^{2}\text{·}N+27c_{1}\text{·}N^{2}-5N^{3}+2c_{1}c_{2}\right)
\]
given in (5.19) of \cite{Marsano:2012yc} where $c_{j}$ denotes the
$j$-th Chern class of $B_{3}$. Now in our case $N=c_{1}$ and, since
$\chi\left(\mathcal{O}_{B_{3}}\right)=1$, the Riemann-Roch theorem
gives that $\int_{B_{3}}c_{1}\text{·}c_{2}=24$. Furthermore, by (4.1)
of \cite{Clemens-2} $c_{1}^{3}=12$. So 
\[
\chi\left(\tilde{W}_{4}\right)=12\int_{B_{3}^{\vee}}\left(c_{1}^{3}+c_{1}c_{2}\right)=12\left(12+24\right)
\]
and 
\[
\chi\left(W_{4}^{\vee}\right)=6\text{·}\left(36+8\right)=24\text{·}11.
\]

\subsection{\label{sec:The-D-term}The $D$-term}

We next consider the $U\left(1\right)_{X}$ factor associated to the
$\left(4+1\right)$-decomposition of the spectral variety. Following
\cite{Grimm:2010ez} this $D$ term must vanish, otherwise we might
break $R$-parity in the low energy theory.

After resolution, 
\[
\tilde{W}_{4}-E_{\mathbf{\bar{5}}}^{\left(44\right)}\rightarrow\tilde{Q}-C_{\mathbf{\bar{5}}}^{\left(44\right)}
\]
is a branched double cover, so that $\tilde{W}_{4}$ will inherit
ample divisors from sufficiently ample divisors on $\tilde{Q}$. So,
to compute the $D$-term we must first adjust the divisor 
\[
\mathcal{C}_{Higgs}^{\left(4\right)}-4\mathcal{C}_{Higgs}^{\left(1\right)}\equiv N
\]
so that its intersections with all curves of the form $\left(X\right)\cdot_{\tilde{Q}}D$
and $\pi_{B_{3}}^{\ast}\left(point\right)$ are zero. Now 
\[
\left(X\right)\cdot_{\tilde{Q}}D\cdot_{\tilde{Q}}N=\pi_{B_{3}}\left(\left(X\right)\cdot_{\tilde{Q}}D\right)\cdot_{B_{3}}N
\]
so that we must adjust by $-N$. But the $D$-term is computed by
identifying an ample divisor on $\tilde{Q}$ that has zero intersection
number with the two-class 
\begin{equation}
\mathcal{G}_{\tilde{Q}}\cdot_{\tilde{Q}}\left(\mathcal{C}_{Higgs}^{\left(4\right)}-4\mathcal{C}_{Higgs}^{\left(1\right)}-\hat{S}_{\mathrm{GUT}}\right).\label{eq:curve}
\end{equation}
But 
\[
\mathcal{G}_{\tilde{Q}}\cdot_{\tilde{Q}}\left(\mathcal{C}_{Higgs}^{\left(4\right)}-4\mathcal{C}_{Higgs}^{\left(1\right)}-S_{\mathrm{GUT}}\right)=\mathcal{G}_{\tilde{Q}}\cdot_{\tilde{Q}}0=0
\]
so any ample divisor trivially has zero intersection with \eqref{eq:curve}.

\section{Wilson line: Symmetry-breaking to the Standard Model\label{subsec:Symmetry-breaking}}

We write
\[
\left(\pi_{B_{3}^{\vee}}\right)_{\ast}\mathcal{O}_{B_{3}}=\mathcal{O}_{B_{3}^{\vee}}\oplus\mathcal{O}_{B_{3}^{\vee}}\left(\varepsilon_{u,v}\right)
\]
with connection isomorphism
\[
\varepsilon_{u,v}:\pi_{1}\left(B_{3}^{\vee}\right)\rightarrow\left\{ \pm1,\text{·}\right\} .
\]

Let
\[
Y=\left(\begin{array}{ccccc}
-1/3\\
 & -1/3\\
 &  & -1/3\\
 &  &  & 1/2\\
 &  &  &  & 1/2
\end{array}\right)\in\mathfrak{su}\left(5\right)
\]
denote the hypercharge direction when acting on the fundamental representation.
The Wilson line is the flat rank-$5$ vector bundle $L_{Y}$ given
by the homomorphism

\[
\pi_{1}\left(S_{\mathrm{GUT}}^{\vee}\right)\rightarrow SU\left(5\right)_{gauge}
\]
that takes the generator to $\exp\left(6\pi i\text{·}Y\right)$ when
viewed as acting on the fundamental representation of $SU\left(5\right)_{gauge}$.
This has the effect of breaking $SU\left(5\right)_{gauge}$-representations
over $S_{\mathrm{GUT}}$ to $SU\left(3\right)\times SU\left(2\right)\times U\left(1\right)_{Y}$-representations
over $S_{\mathrm{GUT}}^{\vee}$. Via tensor and exterior product,
these representations yield eigenspaces respect to the $\mathbb{Z}_{2}\times\mathbb{Z}_{2}=\left\{ \pm1\right\} \times\left\{ \exp\left(6\pi i\text{·}Y\right)\right\} $
holonomy group. Only eigenspaces with eigenvalue pairs in the diagonal
of $\mathbb{Z}_{2}\times\mathbb{Z}_{2}$ remain in the theory. The
flat $U\left(1\right)_{Y}$-bundle gives mass to the gauge states
in $SU\left(5\right)/\left[SU\left(3\right)\times SU\left(2\right)\times U_{Y}\left(1\right)\right]$
at a scale given by the cycle $\propto1/R_{cycle}$ on the $\mathrm{GUT}$
surface. \footnote{Referring to Lemma \ref{lem:The-sections-,Mordell} the flat $U\left(1\right)_{Y}$-bundle
wraps the non-contractible cycle $\left(\upsilon\right)$ while $\left(\zeta\right)$
and $\left(\tau\right)$ are interchanged. } It endows the quotient 
\[
W_{4}^{\vee}/B_{3}^{\vee}
\]
with the Standard Model gauge symmetry $SU\left(3\right)\times SU\left(2\right)\times U\left(1\right)_{Y}$.

We are about to incorporate the very particular translation by $\tau\left(b_{3}\right)-\zeta\left(b_{3}\right)$
in the fiber direction into the definition of the $\mathbb{Z}_{2}$-action
on $\tilde{W}_{4}/B_{3}$. This translation commutes with all the
properties we have attributed to $\tilde{\beta}_{4}/\beta_{3}$ up
to this point. The purpose of including the change of basepoint into
the involution $\tilde{\beta}_{4}/\beta_{3}$ is to eliminate vector-like
exotics from the bulk spectrum. If we were \textit{not} to incorporate
the translation $\tau\left(b_{3}\right)-\zeta\left(b_{3}\right)$
into $\tilde{\beta}_{4}/\beta_{3}$ , the bulk spectrum data for $SU\left(3\right)\times SU\left(2\right)\times U\left(1\right)_{Y}$
would be detailed in the following Table, in which the non-vanishing
of the cohomology groups in the last two rows indicate the existence
of vector-like exotics \cite{Beasley:2008kw,Donagi:2008ca,Marsano:2012yc}:\smallskip{}
 
\begin{center}
\begin{tabular}{|c|c|c|}
\hline 
Representation  & Type of multiplet  & Cohomology group dimension\tabularnewline
\hline 
\hline 
$\left(\mathbf{8},\mathbf{1}\right)_{0}$  & Vector  & $h^{2}\left(S_{\mathrm{GUT}}^{\vee},K_{S_{\mathrm{GUT}}^{\vee}}\right)=1$\tabularnewline
\hline 
$\left(\mathbf{1},\mathbf{3}\right)_{0}$  & Vector  & $h^{2}\left(S_{\mathrm{GUT}}^{\vee},K_{S_{\mathrm{GUT}}^{\vee}}\right)=1$\tabularnewline
\hline 
$\left(\mathbf{1},\mathbf{1}\right)_{0}$  & Vector  & $h^{2}\left(S_{\mathrm{GUT}}^{\vee},K_{S_{\mathrm{GUT}}^{\vee}}\right)=1$\tabularnewline
\hline 
$\left(\mathbf{8},\mathbf{1}\right)_{0}$  & Chiral  & $h^{0}\left(S_{\mathrm{GUT}}^{\vee},K_{S_{\mathrm{GUT}}^{\vee}}\right)\oplus h^{1}\left(S_{\mathrm{GUT}}^{\vee},K_{S_{\mathrm{GUT}}^{\vee}}\right)=0$\tabularnewline
\hline 
$\left(\mathbf{1},\mathbf{3}\right)_{0}$  & Chiral  & $h^{0}\left(S_{\mathrm{GUT}}^{\vee},K_{S_{\mathrm{GUT}}^{\vee}}\right)\oplus h^{1}\left(S_{\mathrm{GUT}}^{\vee},K_{S_{\mathrm{GUT}}^{\vee}}\right)=0$\tabularnewline
\hline 
$\left(\mathbf{1},\mathbf{1}\right)_{0}$  & Chiral  & $h^{0}\left(S_{\mathrm{GUT}}^{\vee},K_{S_{\mathrm{GUT}}^{\vee}}\right)\oplus h^{1}\left(S_{\mathrm{GUT}}^{\vee},K_{S_{\mathrm{GUT}}^{\vee}}\right)=0$\tabularnewline
\hline 
$\left(\mathbf{3},\mathbf{2}\right)_{-5/6}$  & Vector  & $h^{0}\left(S_{\mathrm{GUT}}^{\vee},\mathcal{O}_{S_{\mathrm{GUT}}^{\vee}}\left(\varepsilon_{u,v}\right)\right)=0$\tabularnewline
\hline 
$\left(\bar{\mathbf{3},}\mathbf{2}\right)_{5/6}$  & Vector  & $h^{0}\left(S_{\mathrm{GUT}}^{\vee},\mathcal{O}_{S_{\mathrm{GUT}}^{\vee}}\left(\varepsilon_{u,v}\right)\right)=0$\tabularnewline
\hline 
$\left(\mathbf{3},\mathbf{2}\right)_{-5/6}$  & Chiral  & $h^{1}\left(S_{\mathrm{GUT}}^{\vee},\mathcal{O}_{S_{\mathrm{GUT}}^{\vee}}\left(\varepsilon_{u,v}\right)\right)\oplus h^{2}\left(S_{\mathrm{GUT}}^{\vee},\mathcal{O}_{S_{\mathrm{GUT}}^{\vee}}\left(\varepsilon_{u,v}\right)\right)=1$\tabularnewline
\hline 
$\left(\bar{\mathbf{3},}\mathbf{2}\right)_{5/6}$  & Chiral  & $h^{1}\left(S_{\mathrm{GUT}}^{\vee},\mathcal{O}_{S_{\mathrm{GUT}}^{\vee}}\left(\varepsilon_{u,v}\right)\right)\oplus h^{2}\left(S_{\mathrm{GUT}}^{\vee},\mathcal{O}_{S_{\mathrm{GUT}}^{\vee}}\left(\varepsilon_{u,v}\right)\right)=1$\tabularnewline
\hline 
\end{tabular}\smallskip{}
\par\end{center}

Reviewing the construction of the 'semi-stable degeneration' in Section
\ref{sec:NC} one sees that all the modification steps 
\[
\begin{array}{ccccc}
\overline{W}_{4}^{\wedge}/B_{3}^{\wedge} & \rightarrow & \bar{V}_{3}^{\wedge}/B_{2} & \subseteq & \overline{dP}{}_{a}\cup\overline{dP}{}_{b}\\
\downarrow &  & \downarrow &  & \downarrow\\
\overline{W}_{4}/B_{3} & \rightarrow & \bar{V}_{3}/B_{2} & \subseteq & \overline{dP}{}_{a}\cup\overline{dP}{}_{b}\\
\uparrow &  & \uparrow &  & \uparrow\\
W_{4}^{\left(1\right)}/B_{3} & \rightarrow & V_{3}^{\left(1\right)}/B_{2} & \subseteq & dP_{a}^{\left(1\right)}\cup dP_{b}^{\left(1\right)}\\
\uparrow &  & \uparrow &  & \uparrow\\
W_{4}^{\left(2\right)}/B_{3} & \rightarrow & V_{3}^{\left(2\right)}/B_{2} & \subseteq & dP_{a}^{\left(2\right)}\cup dP_{b}^{\left(2\right)}\\
\uparrow &  & \uparrow &  & \uparrow\\
\tilde{W}_{4}/B_{3} & \rightarrow & V_{3}/B_{2} & \subseteq & dP_{a}\cup dP_{b}
\end{array}
\]
are equivariant with respect to the action of the automorphism group
of $B_{3}$ given by the permutation group $S_{4}\subseteq W\left(SU\left(5\right)\right)$
as in \cite{Clemens-2}. Now on the Heterotic side $\overline{dP}{}_{a}\cup\overline{dP}{}_{b}$
is simply the blow-up of $W_{4,0}$ with center $\left(\upsilon\right)$.

Since 
\[
\left(x,y\right)\mapsto\left(x,-y\right)
\]
before we introduce the translation into the definition of $\tilde{\beta}_{4}$
we have that 
\[
\begin{array}{c}
y\left(\zeta\left(\beta_{3}\left(b_{3}\right)\right)\right)=-y\left(\zeta\left(b_{3}\right)\right)\\
y\left(\tau\left(\beta_{3}\left(b_{3}\right)\right)\right)=-y\left(\tau\left(b_{3}\right)\right)
\end{array}
\]
whereas composing with the translation we have that 
\[
\begin{array}{c}
y\left(\zeta\left(\beta_{3}\left(b_{3}\right)\right)\right)=-y\left(\tau\left(b_{3}\right)\right)\\
y\left(\tau\left(\beta_{3}\left(b_{3}\right)\right)\right)=-y\left(\zeta\left(b_{3}\right)\right)
\end{array}
\]
so that 
\[
y\left(\tau\left(\beta_{3}\left(b_{3}\right)\right)\right)-y\left(\zeta\left(\beta_{3}\left(b_{3}\right)\right)\right)=-y\left(\zeta\left(b_{3}\right)\right)-\left(-y\left(\tau\left(b_{3}\right)\right)\right)
\]
allowing the translation to descend to the quotient. 

Previous construction utilized a single section that was invariant
under the involution, thereby forcing the existence of vector-like
exotics. The existence of two sections and the incorporation of the
translation between them into the $\mathbb{Z}_{2}$-action allows
the elimination of the vector-like exotics.

It is $\left(\upsilon\right)$ on which the non-contractible cycle
of $dP_{a}\cup dP_{b}$ is wrapped, while $\left(\zeta\right)$ and
$\left(\tau\right)$ are interchanged. Thus it is $\left(\upsilon\right)$
that endows the quotient of the action by $\tilde{\beta}_{4}$ with
a section. 

\section{\label{sec:Skew-twisted-cohomology}Bulk and chiral spectra with
Wilson Line on $W_{4}^{\vee}$}

\subsection{Bulk spectrum with Wilson line }

Since the canonical bundle of $S_{\mathrm{GUT}}$ is trivial, the
introduction of the translation into the definition of $\tilde{\beta}_{4}$
will allow us to replace the trivial line bundle by
\begin{equation}
\mathcal{O}_{\tilde{W}_{4}}\left(\left(\tilde{\tau}\right)-\left(\tilde{\zeta}\right)\right)\in\mathrm{Pic}\left(\tilde{W}_{4}\right)\label{eq:shift-1}
\end{equation}
in the computation of the bulk spectrum. More specifically, the semi-stable
degeneration is the geometric bridge between Heterotic theory and
$F$-theory. The introduction of the translation into the definition
of $\tilde{\beta}_{4,0}$ affects the dictionary \eqref{eq:LFMW}
that is given in terms of the Weierstrass form on the $dP_{9}$-bundles
coming from the $F$-theory side and flat line bundles on the elliptic
fibers on the Heterotic side. Change of basepoint from the one given
by $\zeta$ to the one given by $\tau$ does not change the uniquely
given Weierstrass form of the elliptic fiber however on the other
side it does change the sum of eight flat line bundles given by the
restriction of the $E_{8}$-bundle since they are given by the differences
between each of the eight points marked by the (asymptotic) Tate form
and the identity element of the elliptic fiber as a group and the
involution changes this last from the intersection of the torus fiber
with $\left(\zeta\right)$ to the intersection of the torus fiber
with $\left(\tau\right)$.

Therefore instead of pushing forward the canonical bundle $K_{S_{\mathrm{GUT}}}=\mathcal{O}_{S_{\mathrm{GUT}}}$
to $S_{\mathrm{GUT}}^{\vee}$ we must incorporate the twist by the
bundle \eqref{eq:shift-1} in order to compute the bulk spectrum for
$W_{4}^{\vee}/B_{3}^{\vee}$, that is, we must compute the (derived)
cohomology of the push-forward of the twisted structure sheaf 
\[
\left(\mathcal{O}_{S_{\mathrm{GUT}}\times_{\hat{B}_{3}}\tilde{W}_{4}}\right)\otimes\mathcal{O}_{\tilde{W}_{4}}\left(\left(\tilde{\zeta}\right)-\left(\tilde{\tau}\right)\right)
\]
that goes to minus itself under the isomorphism $\tilde{\beta}_{4}$.
The push-forward of the line bundle \eqref{eq:shift-1} to $S_{\mathrm{GUT}}^{\vee}\times_{B_{3}^{\vee}}W_{4}^{\vee}$
breaks into the sum 
\begin{equation}
\left(\pi^{\ast}K_{S_{\mathrm{GUT}}^{\vee}}\otimes\mathcal{O}_{S_{\mathrm{GUT}}^{\vee}\times_{B_{3}^{\vee}}W_{4}^{\vee}}\left(\varepsilon_{u,v}\right)\right)\oplus\left(\mathcal{O}_{S_{\mathrm{GUT}}^{\vee}\times_{B_{3}^{\vee}}W_{4}^{\vee}}\left(\varepsilon_{u,v}\text{·}\tilde{\tau}-\varepsilon_{u,v}\text{·}\tilde{\zeta}\right)\right)\label{eq:summands}
\end{equation}
of symmetric and a skew-symmetric line bundle summands where $\varepsilon_{u,v}$
is the flat orbifold line bundle induced by the involution $\tilde{\beta}_{4}$
with finite fixpoint set.

Thus the derived push-forwards of the summands in \eqref{eq:summands}
are respectively 
\[
\mathcal{O}_{S_{\mathrm{GUT}}^{\vee}}\oplus\left\{ 0\right\} 
\]
are all zero since the divisor $\left(\left(\tilde{\tau}\right)-\left(\tilde{\zeta}\right)\right)$
is not linearly equivalent to zero anywhere, in particular nowhere
over $S_{\mathrm{GUT}}$. Therefore referring to the Table in Section
\ref{subsec:Symmetry-breaking} but now with the incorporation of
the translation $\left(\tilde{\tau}\right)-\left(\tilde{\zeta}\right)$
bulk spectrum data is detailed in the following Table:\smallskip{}
 
\begin{center}
\begin{tabular}{|c|c|c|}
\hline 
Representation  & Type of multiplet  & Cohomology group dimension\tabularnewline
\hline 
\hline 
$\left(\mathbf{8},\mathbf{1}\right)_{0}$  & Vector  & $h^{2}\left(S_{\mathrm{GUT}}^{\vee},K_{S_{\mathrm{GUT}}^{\vee}}\right)=h^{0}\left(\mathcal{O}_{S_{\mathrm{GUT}}^{\vee}}\right)=1$\tabularnewline
\hline 
$\left(\mathbf{1},\mathbf{3}\right)_{0}$  & Vector  & $h^{2}\left(S_{\mathrm{GUT}}^{\vee},K_{S_{\mathrm{GUT}}^{\vee}}\right)=h^{0}\left(\mathcal{O}_{S_{\mathrm{GUT}}^{\vee}}\right)=1$\tabularnewline
\hline 
$\left(\mathbf{1},\mathbf{1}\right)_{0}$  & Vector  & $h^{2}\left(S_{\mathrm{GUT}}^{\vee},K_{S_{\mathrm{GUT}}^{\vee}}\right)=h^{0}\left(\mathcal{O}_{S_{\mathrm{GUT}}^{\vee}}\right)=1$\tabularnewline
\hline 
$\left(\mathbf{8},\mathbf{1}\right)_{0}$  & Chiral  & $h^{0}\left(S_{\mathrm{GUT}}^{\vee},K_{S_{\mathrm{GUT}}^{\vee}}\right)\oplus h^{1}\left(S_{\mathrm{GUT}}^{\vee},K_{S_{\mathrm{GUT}}^{\vee}}\right)=0$\tabularnewline
\hline 
$\left(\mathbf{1},\mathbf{3}\right)_{0}$  & Chiral  & $h^{0}\left(S_{\mathrm{GUT}}^{\vee},K_{S_{\mathrm{GUT}}^{\vee}}\right)\oplus h^{1}\left(S_{\mathrm{GUT}}^{\vee},K_{S_{\mathrm{GUT}}^{\vee}}\right)=0$\tabularnewline
\hline 
$\left(\mathbf{1},\mathbf{1}\right)_{0}$  & Chiral  & $h^{0}\left(S_{\mathrm{GUT}}^{\vee},K_{S_{\mathrm{GUT}}^{\vee}}\right)\oplus h^{1}\left(S_{\mathrm{GUT}}^{\vee},K_{S_{\mathrm{GUT}}^{\vee}}\right)=0$\tabularnewline
\hline 
$\left(\mathbf{3},\mathbf{2}\right)_{-5/6}$  & Vector  & $h^{0}\left(\mathcal{O}_{S_{\mathrm{GUT}}^{\vee}\times_{B_{3}^{\vee}}W_{4}^{\vee}}\left(\varepsilon_{u,v}\text{·}\tilde{\tau}-\varepsilon_{u,v}\text{·}\tilde{\zeta}\right)\right)=0$\tabularnewline
\hline 
$\left(\bar{\mathbf{3},}\mathbf{2}\right)_{5/6}$  & Vector  & $h^{0}\left(\mathcal{O}_{S_{\mathrm{GUT}}^{\vee}\times_{B_{3}^{\vee}}W_{4}^{\vee}}\left(\varepsilon_{u,v}\text{·}\tilde{\tau}-\varepsilon_{u,v}\text{·}\tilde{\zeta}\right)\right)=0$\tabularnewline
\hline 
$\left(\mathbf{3},\mathbf{2}\right)_{-5/6}$  & Chiral  & $h^{1}\left(\mathcal{O}_{S_{\mathrm{GUT}}^{\vee}\times_{B_{3}^{\vee}}W_{4}^{\vee}}\left(\varepsilon_{u,v}\text{·}\tilde{\tau}-\varepsilon_{u,v}\text{·}\tilde{\zeta}\right)\right)\oplus h^{2}\left(\ldots\right)=0$\tabularnewline
\hline 
$\left(\bar{\mathbf{3},}\mathbf{2}\right)_{5/6}$  & Chiral  & $h^{1}\left(\mathcal{O}_{S_{\mathrm{GUT}}^{\vee}\times_{B_{3}^{\vee}}W_{4}^{\vee}}\left(\varepsilon_{u,v}\text{·}\tilde{\tau}-\varepsilon_{u,v}\text{·}\tilde{\zeta}\right)\right)\oplus h^{2}\left(\ldots\right)=0$\tabularnewline
\hline 
\end{tabular}\smallskip{}
 
\par\end{center}

Thus there are no vector-like exotics!

\subsection{Matter spectrum }

In fact we have already computed the correct chiral spectra in \eqref{eq:yes1},
\eqref{eq:yes2}, and \eqref{eq:yes3} of Subsection \ref{subsec:Restrictions-to-the}. 
\begin{lem}
\label{lem:With-respect-to}With respect to the $C_{u,v}$-eigen-line-bundle
decomposition of the push-forwards to $B_{3}^{\vee}$ of the (restrictions
of the) Higgs bundle to the matter and Higgs curves in $B_{3}$, 
\end{lem}

i) 
\[
\begin{array}{c}
h^{0}\left(\check{\mathcal{L}}_{\mathbf{10}}^{\left(4\right)\left[\text{\textpm}1\right]}\right)=h^{0}\left(\check{\mathcal{L}}_{\mathbf{\bar{5}}}^{\left(41\right)\left[\text{\textpm1}\right]}\right)=3\\
h^{1}\left(\check{\mathcal{L}}_{\mathbf{10}}^{\left(4\right)\left[\text{\textpm1}\right]}\right)=h^{1}\left(\check{\mathcal{L}}_{\mathbf{\bar{5}}}^{\left(41\right)\left[\text{\textpm1}\right]}\right)=0.
\end{array}
\]
ii) 
\[
\begin{array}{c}
h^{0}\left(\check{\mathcal{L}}_{\mathbf{\bar{5}}}^{\left(44\right)\left[+1\right]}\right)=h^{1}\left(\check{\mathcal{L}}_{\mathbf{\bar{5}}}^{\left(44\right)\left[+1\right]}\right)=0\\
h^{0}\left(\check{\mathcal{L}}_{\mathbf{\bar{5}}}^{\left(44\right)\left[-1\right]}\right)=h^{1}\left(\check{\mathcal{L}}_{\mathbf{\bar{5}}}^{\left(44\right)\left[-1\right]}\right)=1.
\end{array}
\]

Therefore referring to Tables 1 and 2 in Section 7 of \cite{Clemens-2},
the flux distribution associated to the $\mathbb{Z}_{2}$-action given
by the involution $\tilde{\beta}_{4}$ and the Wilson line that is
wrapped by that involution is presented in the Tables below. Namely
the distribution of MSSM matter fields is as follows:

\smallskip{}
\noindent \begin{center}
\begin{tabular}{|c|c|c|c|c|}
\hline 
$\Sigma_{\mathbf{10}}^{\left(4\right)}=\left\{ a_{5}=z=0\right\} $  & $C_{u,v}$  & \textbf{$L_{Y}$} & $\mathcal{L}_{Higgs}$  & $SU\left(3\right)\times SU\left(2\right)\times U\left(1\right)_{Y}$\tabularnewline
\hline 
\hline 
$h^{0}\left(\check{\mathcal{L}}_{\mathbf{10}}^{\left(4\right)\left[\text{\textpm}1\right]}\right)$ & $+1$  & $+1$  & $3$  & $\left(\mathbf{\mathbf{\mathbf{1}}},\mathbf{1}\right)_{+1}$\tabularnewline
\hline 
 & $-1$  & $-1$  &  & $\left(\mathbf{\mathbf{\mathbf{3}}},\mathbf{2}\right)_{+1/6}$\tabularnewline
\hline 
 & $+1$  & $+1$  &  & $\left(\mathbf{\mathbf{\mathbf{\bar{3}}}},\mathbf{1}\right)_{-2/3}$\tabularnewline
\hline 
$h^{1}\left(\check{\mathcal{L}}_{\mathbf{10}}^{\left(4\right)\left[\text{\textpm}1\right]}\right)$  & $+1$  & $+1$  & $0$  & $\left(\mathbf{\mathbf{\mathbf{1}}},\mathbf{1}\right)_{+1}$\tabularnewline
\hline 
 & $-1$  & $-1$  &  & $\left(\mathbf{\mathbf{\mathbf{\bar{3}}}},\mathbf{2}\right)_{+1/6}$\tabularnewline
\hline 
 & $+1$  & $+1$  &  & $\left(\mathbf{\mathbf{\mathbf{3}}},\mathbf{1}\right)_{+2/3}$\tabularnewline
\hline 
\end{tabular}\smallskip{}
\begin{tabular}{|c|c|c|c|c|}
\hline 
$\Sigma_{\mathbf{\bar{5}}}^{\left(41\right)}=\left\{ a_{420}=z=0\right\} $  & $C_{u,v}$  & \textbf{$L_{Y}$} & $\mathcal{L}_{Higgs}$  & $SU\left(3\right)\times SU\left(2\right)\times U\left(1\right)_{Y}$\tabularnewline
\hline 
\hline 
$h^{0}\left(\check{\mathcal{L}}_{\mathbf{\bar{5}}}^{\left(41\right)\left[\text{\textpm1}\right]}\right)$  & $+1$  & $+1$  & $3$  & $\left(\mathbf{\mathbf{\mathbf{\bar{3}}}},\mathbf{1}\right)_{+1/3}$\tabularnewline
\hline 
 & $-1$  & $-1$  &  & $\left(\mathbf{\mathbf{\mathbf{1}}},\mathbf{2}\right)_{-1/2}$\tabularnewline
\hline 
$h^{1}\left(\check{\mathcal{L}}_{\mathbf{\bar{5}}}^{\left(41\right)\left[\text{\textpm1}\right]}\right)$  & $+1$  & $+1$  & $0$  & $\left(\mathbf{\mathbf{\mathbf{3}}},\mathbf{1}\right)_{-1/3}$\tabularnewline
\hline 
 & $-1$  & $-1$  &  & $\left(\mathbf{\mathbf{\mathbf{1}}},\mathbf{2}\right)_{+1/2}$\tabularnewline
\hline 
\end{tabular}\smallskip{}
 
\par\end{center}

As for the Higgs fields we have the following:\smallskip{}
\noindent \begin{center}
\begin{tabular}{|c|c|c|c|c|}
\hline 
$\Sigma_{\mathbf{\bar{5}}}^{\left(44\right)}=\left\{ a_{4}a_{3}+a_{5}\left(a_{0}-a_{3}\right)=z=0\right\} $  & $C_{u,v}$  & $L_{Y}$  & $\mathcal{L}_{Higgs}$  & $SU\left(3\right)\times SU\left(2\right)\times U\left(1\right)_{Y}$\tabularnewline
\hline 
\hline 
$h^{0}\left(\check{\mathcal{L}}_{\mathbf{\bar{5}}}^{\left(44\right)\left[+1\right]}\right)$  & $+1$  & $+1$  & $0$  & $\left(\mathbf{\mathbf{\mathbf{\bar{3}}}},\mathbf{1}\right)_{+1/3}$\tabularnewline
\hline 
$h^{0}\left(\check{\mathcal{L}}_{\mathbf{\bar{5}}}^{\left(44\right)\left[-1\right]}\right)$  & $-1$  & $-1$  & $1$  & $\left(\mathbf{\mathbf{\mathbf{1}}},\mathbf{2}\right)_{-1/2}$\tabularnewline
\hline 
$h^{1}\left(\check{\mathcal{L}}_{\mathbf{\bar{5}}}^{\left(44\right)\left[+1\right]}\right)$  & $+1$  & $+1$  & $0$  & $\left(\mathbf{\mathbf{\mathbf{3}}},\mathbf{1}\right)_{-1/3}$\tabularnewline
\hline 
$h^{1}\left(\check{\mathcal{L}}_{\mathbf{\bar{5}}}^{\left(44\right)\left[-1\right]}\right)$  & $-1$  & $-1$  & $1$  & $\left(\mathbf{\mathbf{\mathbf{1}}},\mathbf{2}\right)_{+1/2}$\tabularnewline
\hline 
\end{tabular}
\par\end{center}

\smallskip{}

In conclusion we have reproduced the spectrum of the minimal supersymmetric
Standard Model.

\section{Asymptotic $\mathbb{Z}_{4}$ R-symmetry\label{sec:Asymptotic--R-symmetry}}

The asymptotic $\mathbb{Z}_{4}$-symmetry $T_{u,v}$ of the semi-stable
$F$-theory limit constructed on $B_{3}^{\vee}$ in Section 6 of \cite{Clemens-2}
acts on sections and twisted sections of the anti-canonical bundle
$B_{3}^{\vee}$. Referring to the results and notation in Sections
5 and 6 of \cite{Clemens-2}, the asymptotic $\mathbb{Z}_{4}$ $\mathbf{R}$-symmetry
$T_{u,v}$ constructed there on
\[
\lim_{\delta\rightarrow0}\tilde{W}_{4,\delta}
\]
 is compatible with the following table from Subsection 8.2 of \cite{Clemens-2}:
\footnote{Given the $\mathbb{Z}_{4}$ R-charges, $i^{q_{f}+1}$, for the fermionic
components of $N=1$ superfields, then the bosonic components of the
chiral superfields are given by $i^{q_{f}}$. This is in accord with
the anti-commuting superspace coordinate transforming as $\theta^{\prime}=i^{-1}\theta$.
Otherwise, these are equivalent to the $\mathbb{Z}_{4}$ R-charges
of \cite{Lee,Lee-1:}. }
\begin{center}
\smallskip{}
\par\end{center}

\begin{center}
\begin{tabular}{|c|c|c|}
\hline 
TABLE 3:\quad{}$T_{u,v}$  & $T_{u,v}$-charge  & space\tabularnewline
\hline 
\hline 
matter fields on $\frac{\Sigma_{\mathbf{10}}^{\left(4\right)}}{\left\{ C_{u,v}\right\} }$ & $-1$  & $H^{0}\left(\frac{\Sigma_{\mathbf{10}}^{\left(4\right)}}{\left\{ C_{u,v}\right\} };\mathcal{L}_{Higgs}^{\vee,\left[\pm1\right]}\right)$\tabularnewline
\hline 
matter fields on $\frac{\Sigma_{\mathbf{\bar{5}}}^{\left(41\right)}}{\left\{ C_{u,v}\right\} }$ & $-1$ & $H^{0}\left(\frac{\Sigma_{\mathbf{\bar{5}}}^{\left(41\right)}}{\left\{ C_{u,v}\right\} };\mathcal{L}_{Higgs}^{\vee,\left[\pm1\right]}\right)$\tabularnewline
\hline 
Higgs fields on $\frac{\Sigma_{\mathbf{\bar{5}}}^{\left(44\right)}}{\left\{ C_{u,v}\right\} }$  & $+i$  & $H^{0}\left(\frac{\Sigma_{\mathbf{\bar{5}}}^{\left(44\right)}}{\left\{ C_{u,v}\right\} };\mathcal{L}_{Higgs}^{\vee,\left[-1\right]}\right)$
/ $H^{1}\left(\frac{\Sigma_{\mathbf{\bar{5}}}^{\left(44\right)}}{\left\{ C_{u,v}\right\} };\mathcal{L}_{Higgs}^{\vee,\left[-1\right]}\right)$\tabularnewline
\hline 
bulk matter on $\frac{S_{\mathrm{GUT}}}{\left\{ C_{u,v}\right\} }$ & $-i$ & $H^{2}\left(K_{\frac{S_{\mathrm{GUT}}}{\left\{ C_{u,v}\right\} }}\right)$\tabularnewline
\hline 
\end{tabular}
\par\end{center}

\noindent \begin{center}
\smallskip{}
\par\end{center}

The $\mathbb{Z}_{4}$ R symmetry forbids the Higgs $\mu$-term and
dimension 4 and 5 baryon and lepton number violating operators. 

\section{Conclusions}

In this paper we have constructed an SU(5) GUT $F$-theory model.
We have shown how to break the GUT group with a non-local Wilson line.
Thus we are able to identify the GUT scale with the compactification
scale of the GUT surface. Our model includes three families of quarks
and leptons and one pair of Higgs doublets. The price to pay for this
result is that we have a mirror world where the mass scales and couplings
of the mirror states may be different than for the MSSM. This mirror
world can, in principle, be the dark matter of the universe. There
may or may not also be direct couplings of the mirror and MSSM sectors
of the theory. These give interesting physics as in \cite{Berezhiani,Chacko,Barbieri-1}.
There are no vector-like exotics in the bulk spectrum or on the matter
curves, neither are there chiral exotics. 
\begin{flushleft}
The existence of the $\mathbb{Z}_{4}$ $\mathbf{R}$-symmetry generated
by the automorphism $T_{u,v}$ that reverses the sign on the holomorphic
$4$-form addresses the issue of dimension-5 proton decay operators
and forbids a $\mu$-term \cite{Lee}. The charges of the matter states
under the $\mathbb{Z}_{4}$ $\mathbf{R}$-symmetry are given in Table
3. Finally, Wilson line symmetry-breaking is addressed in Sections
\ref{subsec:Symmetry-breaking} and \ref{sec:Skew-twisted-cohomology}. 
\par\end{flushleft}

However there are several issues which are not resolved in this paper.
Moduli stabilization is not addressed. We have not generated a $\mu$-term.
In principle the $\mathbb{Z}_{4}$ $\mathbf{R}$-symmetry can be broken
by non-perturbative physics down to matter parity which then allows
for a $\mu$-term of order the weak scale and severely suppressed
dimension-5 proton decay operators. We have not discussed the possible
Yukawa interactions needed to give quarks and leptons mass. We may
or may not also have right-handed neutrinos which would be useful
for a see-saw mechanism of neutrino masses.

Finally, the $U(1)_{X}$ gauge symmetry may be broken to $\mathbb{Z}_{2}$
matter parity via non-perturbative effects at the GUT scale or by
a Stueckelberg mechanism. This would then allow for neutrino Majorana
masses near the GUT scale.

\section*{Acknowledgments}

The authors would like to thank Dave Morrison, Tony Pantev and Sakura
Schäfer-Nameki for their guidance and many helpful conversations over
several years. The authors would also like to thank Sakura Schäfer-Nameki
for her contribution of the Appendix to this paper. However the authors
themselves take sole responsibility for any errors or omissions in
this paper.

S.R. acknowledges partial support from DOE grant DE-SC0011726.

\appendix

\section{$\mathbb{P}^{112}$ Formulations}
\begin{center}
by Sakura Schäfer-Nameki 
\par\end{center}

\subsection{\label{subsec:-App.1}Realizing $\bar{W_{4}}$ in $\mathbb{P}^{112}$}

An alternative formulation of the elliptic fibrations $\overline{W}_{4}$
can be given in terms of the $\mathbb{P}^{112}$-fibration in \cite{Morrison:2012ei}.
Let 
\[
\mathbb{P}^{112}:=\mathbb{P}\left(\mathcal{O}_{B_{3}}\oplus\mathcal{O}_{B_{3}}\left(2\right)\oplus\mathcal{O}_{B_{3}}\left(4\right)\right)
\]
with projective coordinates $[w,x,y]$. For $c_{i}\in H^{0}\left(\mathcal{O}_{B_{3}}\left(8-2i\right)\right)$
define 
\begin{equation}
\begin{array}{c}
c_{0}w^{4}+c_{1}w^{3}x+c_{2}w^{2}x^{2}+c_{3}wx^{3}\\
=y^{2}+\left(b_{0}x^{2}+b_{1}wx+b_{2}w^{2}\right)y.
\end{array}\label{eq:QuarticDef}
\end{equation}
This elliptic fibration has two rational sections 
\[
\begin{array}{c}
w=y=0\\
w=y+b_{0}x^{2}=0
\end{array}
\]
that can be made symmetric by the shift 
\begin{equation}
y\mapsto y-\frac{b_{0}x^{2}}{2}.\label{eq:shift}
\end{equation}
\[
\left|\begin{array}{cc}
h_{0} & -\hat{t}\left(\hat{h}_{1}\hat{z}-\hat{h}_{2}\hat{t}\right)\\
\hat{t}\left(d_{2}\hat{z}+d_{3}\hat{t}\right) & \hat{c}_{4}\hat{t}^{2}\left(\hat{t}^{2}+\hat{z}^{2}\right)+\hat{c}_{3}\hat{t}^{3}\hat{z}+\hat{c}_{1}z^{3}\left(\hat{t}-\hat{z}\right)
\end{array}\right|=0.
\]
\[
\left|\begin{array}{cc}
h_{0} & -\hat{t}\left(\hat{h}_{1}\hat{z}-\hat{h}_{2}\hat{t}\right)\\
\hat{t}\left(d_{2}\hat{z}+d_{3}\hat{t}\right) & \hat{c}_{4}\hat{t}^{2}\left(\hat{t}^{2}+\hat{z}^{2}\right)+\hat{c}_{3}\hat{t}^{3}\hat{z}+\hat{c}_{1}z^{3}\left(\hat{t}-\hat{z}\right)
\end{array}\right|=0.
\]

Denote by 
\[
\begin{array}{c}
\gamma_{i}=ord_{z}\left(c_{i}\right)\\
\beta_{j}=ord_{z}\left(b_{j}\right)
\end{array}
\]
and let 
\[
Q\left(\gamma_{0},\gamma_{1},\gamma_{3},\gamma_{0},\beta_{0},\beta_{1},\beta_{2}\right)
\]
denote the quartic \eqref{eq:QuarticDef} for the given values. Using
Tate's algorithm, there are various ways to degenerate this to an
$I_{5}$ fiber about the locus $\left\{ z=0\right\} $ \cite{Kuntzler:2014ila},
both by specifying the vanishing orders in \eqref{eq:QuarticDef}
without further relations among the leading order coefficients $c_{i}$
and $b_{j}$ (`canonical Tate models'), or by imposing relations among
the coefficients (`non-canonical models').

In particular, the fibration $\overline{W}_{4}$ wth coordinates $\left(X,Y,w\right)$
in Subsection \ref{subsec:The-standard--formulation} becomes in this
language a non-canonical model with coordinates $\left(w,x,y\right)$
given by the substitution
\[
\begin{array}{c}
X\leftrightarrow w\\
Y\leftrightarrow x\\
w\leftrightarrow y+x^{2}
\end{array}
\]
so that \eqref{eq:waffeq} becomes the $Q\left(2,1,0,0,0,\infty,\infty\right)$
given in \eqref{eq:pee112} by 
\[
\begin{array}{c}
y^{2}+2x^{2}y=-2a_{5}x^{3}w\\
+\left(a_{5}^{2}-2a_{4}z-6z^{2}\right)x^{2}w^{2}\\
+\left(2a_{5}a_{4}z+\left(6a_{5}-4a_{420}\right)z^{2}-8z^{3}\right)xw^{3}\\
-\left(\left(-a_{4}^{2}-4a_{5}a_{420}\right)z^{2}+\left(2a_{4}+4a_{2}-4a_{5}\right)z^{3}+3z^{4}\right)w^{4}.
\end{array}
\]

The $I_{5}$ singular fiber enhances to $I_{6}$ along $-a_{53}=a_{420}=0$
and $\left|\begin{array}{cc}
a_{4} & -a_{5}\\
a_{0}+a_{3} & a_{3}
\end{array}\right|=0,$ and we have $I_{1}^{*}$ enhancement along $a_{5}=0$.

The sections intersect the $I_{5}$ fiber along $D_{0}$ and $D_{2}$,
where, as above, the rational curves in the $I_{5}$ associated to
the simple roots $\alpha_{i}$ are denoted by $D_{i}$, with the extended
node corresponding to $D_{0}$. This type of $I_{5}$ fiber with two
rational sections was denoted by $I_{5}^{(0||1)}$, where the separation
of the sections is $\#|-1$. \cite{Kuntzler:2014ila,Lawrie}

\subsection{Other models with similar fiber type}

A different set of $I_{5}^{(0||1)}$ non-canonical models were determined
in \cite{Kuntzler:2014ila} \smallskip{}
 
\begin{table}
\centering %
\begin{tabular}{c|c|c}
Matter loci  & Rep and $U(1)$ charge  & Codim 2 fiber \tabularnewline
\hline 
\hline 
$\sigma_{3}$  & $\mathbf{10}_{1}+\overline{\mathbf{10}}_{-1}$  & $I_{1}^{*(0||1)}$ \tabularnewline
$\sigma_{1}$  & $\mathbf{10}_{-4}+\overline{\mathbf{10}}_{4}$  & $I_{1}^{*(01)}$ \tabularnewline
$\sigma_{2}$  & $\mathbf{5}_{-7}+\overline{\mathbf{5}}_{7}$  & $I_{6}^{(0|1)}$ \tabularnewline
$P_{1}$  & $\mathbf{5}_{-2}+\overline{\mathbf{5}}_{2}$  & $I_{6}^{(0||1)}$ \tabularnewline
$P_{2}$  & $\mathbf{5}_{3}+\overline{\mathbf{5}}_{-3}$  & $I_{6}^{(0|||1)}$ \tabularnewline
\end{tabular}\caption{ Matter loci and $U(1)$ charges for the $I_{5}^{(0||1)}$ fiber given
by $Q(3,2,1,1,0,0,1)|_{P_{0}=0}$. }
\end{table}

This model is based on the $I_{4}$ fiber $Q\left(3,2,1,0,0,0,1\right)$
given by 
\[
\begin{array}{c}
c_{0}z^{3}w^{4}+c_{1}z^{2}w^{3}x+c_{2}zw^{2}x^{2}+c_{3}zwx^{3}\\
=y^{2}+b_{0}x^{2}y+b_{1}ywx+b_{2}zw^{2}y.
\end{array}
\]
where in addition we impose 
\[
P_{0}=\left|\begin{array}{cc}
b_{0} & b_{2}\\
c_{3} & c_{2}
\end{array}\right|=0.
\]
This last is solved as follows: 
\[
\left|\begin{array}{cc}
b_{0} & b_{2}\\
c_{3} & c_{2}
\end{array}\right|=\left|\left(\begin{array}{c}
\sigma_{1}\\
\sigma_{4}
\end{array}\right)\left(\begin{array}{cc}
\sigma_{2} & \sigma_{3}\end{array}\right)\right|=0.
\]
This fibration $Q(3,2,1,1,0,0,1)|_{P_{0}=0}$ has generically $I_{5}$
fiber, with same intersection pattern with the two rational sections
as $W_{4}$, but it provides more matter loci, shown in the table,
with the following expressions for the matter curves 
\[
\begin{array}{c}
P_{1}=\sigma_{4}b_{2,1}^{2}+\sigma_{1}^{2}\sigma_{3}c_{0,3}-\sigma_{1}b_{2,1}c_{1,2}\\
P_{2}=\sigma_{1}\sigma_{2}^{2}\left|\begin{array}{cc}
\sigma_{1} & b_{2,1}\\
\sigma_{4} & c_{1,2}
\end{array}\right|+\sigma_{1}\sigma_{3}^{2}\left|\begin{array}{cc}
\sigma_{1} & b_{0,1}\\
\sigma_{4} & c_{3,2}
\end{array}\right|\\
+\sigma_{3}\sigma_{2}\left(\sigma_{4}\sigma_{1}b_{1,1}-\sigma_{1}^{2}c_{2,2}+\sigma_{4}^{2}\right).
\end{array}
\]
The sections are not symmetric in the form above, but this can again
be remedied by shifting as in \eqref{eq:shift}.

\end{document}